\documentclass[a4paper,onecolumn,accepted=2022-08-08]{quantumarticle}
\pdfoutput=1

\usepackage{amsfonts,amssymb}

\usepackage[T1]{fontenc}

\usepackage{amsmath,amsthm,dsfont}
\usepackage{mathabx}

\usepackage{microtype}

\usepackage[affil-it]{authblk}
\usepackage{enumerate}
\usepackage[english]{babel}
\usepackage{graphicx}	
\usepackage[margin=15pt,font=small,labelfont={bf,sf}]{caption}
\usepackage{subcaption}
\usepackage[margin=3cm]{geometry}
\usepackage{url}
\usepackage{todonotes}
\usepackage{verbatim}
\usepackage{hyperref}
\hypersetup{colorlinks=true,citecolor=blue,linkcolor=blue,filecolor=blue,urlcolor=blue,breaklinks=true,bookmarksnumbered=true, bookmarksopen=true,bookmarksopenlevel=1}
\definecolor{dullmagenta}{rgb}{0.4,0,0.4}
\allowdisplaybreaks
\usepackage{braket}
\usepackage{dsfont}
\usepackage{mathtools}
\usepackage{tikz}
\usetikzlibrary{shapes.geometric,decorations.markings,external}
\usepackage{quantikz}
\usepackage{adjustbox}

\usepackage[ruled,vlined]{algorithm2e}
\usepackage[capitalise]{cleveref}
\crefname{equation}{}{}

\newtheorem{theorem}{Theorem}[section]

\newtheorem{lemma}[theorem]{Lemma}

\newtheorem{claim}{Claim}[section]

\SetKwInput{KwInput}{Input}
\SetKwInput{KwOutput}{Output}
\SetKwProg{Fn}{Def}{:}{}

\SetCommentSty{mycommfont}
\SetKwRepeat{Struct}{struct \{}{\}}%

\DeclareFontFamily{U}{mathb}{\hyphenchar\font45}
\DeclareFontShape{U}{mathb}{m}{n}{
      <5> <6> <7> <8> <9> <10> gen * mathb
      <10.95> mathb10 <12> <14.4> <17.28> <20.74> <24.88> mathb12
      }{}
\DeclareSymbolFont{mathb}{U}{mathb}{m}{n}

\DeclareFontFamily{U}{matha}{\hyphenchar\font45}
\DeclareFontShape{U}{matha}{m}{n}{
      <5> <6> <7> <8> <9> <10> gen * matha
      <10.95> matha10 <12> <14.4> <17.28> <20.74> <24.88> matha12
      }{}
\DeclareSymbolFont{matha}{U}{matha}{m}{n}

\DeclareMathSymbol{\oasterisk}{3}{matha}{"66}
\DeclareMathSymbol{\boxasterisk}{3}{mathb}{"66}
\def\ostar{\oasterisk}
\def\boxstar{\boxasterisk}

\DeclareFontFamily{U}{bbold}{}
\DeclareFontShape{U}{bbold}{m}{n}
 {
  <-5.5> s*[1.069] bbold5
  <5.5-6.5> s*[1.069] bbold6
  <6.5-7.5> s*[1.069] bbold7
  <7.5-8.5> s*[1.069] bbold8
  <8.5-9.5> s*[1.069] bbold9
  <9.5-11> s*[1.069] bbold12 %was 10
  <11-15> s*[1.069] bbold12
  <15-> s*[1.069] bbold17
 }{}

\DeclareRobustCommand{\id}{%
  \text{\usefont{U}{bbold}{m}{n}1}%
}

\usepackage[backend=biber,sorting=none,style=phys-jmr,biblabel=brackets,backref=true,bibencoding=utf8,maxnames=20,eprint=true]{biblatex}
\makeatletter
\pretocmd{\blx@head@bibintoc}{\phantomsection}{}{\ddt}
\makeatother
\DefineBibliographyStrings{english}{%
  backrefpage = {page},% originally "cited on page"
  backrefpages = {pages},% originally "cited on pages"
}
\usepackage{csquotes}

\def\concat{\mathbin\Vert}
\def\cnot{\textsc{cnot}}
\def\notc{\textsc{notc}}

\newcommand{\vmp}[1]{\ensuremath{V_{#1}^{\text{mp}}}}
\newcommand{\vext}[1]{\ensuremath{V_{#1}^{\text{ext}}}}
\newcommand{\ketmp}[1]{\ensuremath{\ket{#1^{\text{mp}}}}}
\newcommand{\ketext}[1]{\ensuremath{\ket{#1^{\text{ext}}}}}
\newcommand{\ketbra}[1]{\ket{#1}\bra{#1}}
\newcommand{\abs}[1]{\left\lvert#1\right\rvert}
\newcommand{\norm}[1]{\left\|#1\right\|}
\newcommand{\qstate}[2]{\ket{Q (#1,#2)}}
\newcommand{\qstateconj}[2]{\bra{Q (#1,#2)}}
\newcommand{\anglestate}[1]{\ket{\sphericalangle (#1)}}
\newcommand{\anglestateconj}[1]{\bra{\sphericalangle (#1)}}
\NewDocumentCommand\quant{g}{%
  \IfNoValueTF{#1}{R_B}{R_B({#1})}%
}
\NewDocumentCommand\quantalphabeta{m +g}{%
  \IfNoValueTF{#2}{\tilde{R}_{#1}}{\tilde{R}_{#1}({#2})}%
}
\DeclareMathOperator*{\argmax}{arg\,max}
\DeclareMathOperator*{\argmin}{arg\,min}
\DeclareMathOperator{\poly}{poly}
\DeclareMathOperator{\polylog}{polylog}
\DeclareMathOperator{\Tr}{Tr} 
\DeclareExpandableDocumentCommand{\uctrl}{O{}m}{|[shape=circle,draw=black,minimum size=5pt,#1,path picture={\fill[black] (path picture bounding box.north east) -- plot (path picture bounding box.south west) -- (path picture bounding box.south east) -- cycle;},label={[phase label,#1]#2}]| {} \qw}

\tikzset{roundnode/.style = {shape=circle, text=black, draw=black, fill=white, inner sep=0.5mm, outer sep=0.0mm}}
\tikzset{quadnode/.style = {shape=rectangle, text=black, draw=black, fill=white, inner sep=0.5mm, outer sep=0.0mm}}
\tikzset{squarenode/.style = {shape=regular polygon,regular polygon sides=4, text=black, draw=black, fill=white, inner sep=0.0mm, outer sep=0.0mm}}
\tikzset{phase/.style = {draw,fill,shape=circle,minimum size=5pt,inner sep=0pt}}
\tikzset{ucontrolled/.style = {fill=white,shape=circle,draw=black,minimum size=5pt,inner sep=0pt,path picture={\fill[black] (path picture bounding box.north east) -- plot (path picture bounding box.south west) -- (path picture bounding box.south east) -- cycle;}}}
\tikzset{circlewc/.style={draw,circle,cross,minimum width=0.3 cm}}

\newlength\eqheight
\def\cequal{\raisebox{0pt}[\eqheight][0pt]{\raisebox{-0.85pt}{$=$}}}

\title{Quantum message-passing algorithm for optimal and efficient decoding}% of CQ channels}
\author{Christophe Piveteau and Joseph M. Renes}
\affil{\normalsize Institute for Theoretical Physics, ETH Z\"urich, Switzerland}
\date{}

\begin{document}

\renewcommand{\abstractname}{\vspace{-2.5\baselineskip}} % https://tex.stackexchange.com/a/53175

\maketitle

\begin{abstract}
Recently, Renes proposed a quantum algorithm called belief propagation with quantum messages (BPQM) for decoding classical data encoded using a binary linear code with tree Tanner graph that is transmitted over a pure-state CQ channel~\cite{renes_2017}, i.e., a channel with classical input and pure-state quantum output. 
The algorithm presents a genuine quantum counterpart to decoding based on the classical belief propagation algorithm, which has found wide success in classical coding theory when used in conjunction with LDPC or Turbo codes. 
More recently Rengaswamy \emph{et al.}~\cite{rengaswamy_2020} observed that BPQM implements the optimal decoder on a small example code, in that it implements the optimal measurement that distinguishes the quantum output states for the set of input codewords with highest achievable probability. 
Here we significantly expand the understanding, formalism, and applicability of the BPQM algorithm with the following  contributions.
First, we prove analytically that BPQM realizes optimal decoding for any binary linear code with tree Tanner graph. 
We also provide the first formal description of the BPQM algorithm in full detail and without any ambiguity. 
In so doing, we identify a key flaw overlooked in the original algorithm and subsequent works which implies quantum circuit realizations will be exponentially large in the code dimension. 
Although BPQM passes quantum messages, other information required by the algorithm is processed globally.  
We remedy this problem by formulating a truly message-passing algorithm which approximates BPQM and has quantum circuit complexity $\mathcal{O}(\poly n, \polylog \frac{1}{\epsilon})$, where $n$ is the code length and $\epsilon$ is the approximation error. 
Finally, we also propose a novel method for extending BPQM to factor graphs containing cycles by making use of approximate cloning. 
We show some promising numerical results that indicate that BPQM on factor graphs with cycles can significantly outperform the best possible classical decoder.
\end{abstract}

%\tableofcontents

\section{Introduction}
Message-passing algorithms on graphical models have proven to be an integral computational tool for many fields such as statistics, information theory, machine learning, and statistical physics. \emph{Belief propagation} (BP), which is perhaps the most well-known message-passing algorithm, approximates the marginal distributions of a given probabilistic graphical model. 
In the context of a spin model in statistical physics, for example, the single-site marginals yield the magnetization of the model. 
BP operates by sending messages across the edges of the graph in question. 
The nodes in the graph receive messages from their neighbors, perform a local computation, and then transmit the output to their neighbors. 
It has found tremendous success in the field of coding theory as a means of decoding data encoded into linear block codes and transmitted  over noisy classical channels. 
For instance, certain low-density parity-check (LDPC) codes have been shown to approach the Shannon capacity of the channel when paired with BP decoding on the Tanner graph describing the parity-checks of the code~\cite{kudekar_2013}.
The BP decoder works by estimating the bits of the codeword individually, which requires marginalization of the conditional probability distribution of the input codewords given the observed noisy channel output. 

The task of decoding noisy channels with quantum output is quite different. 
Even for channels with classical input and quantum output, so-called CQ channels, there is generally no analog to the conditional probability distribution of the inputs. 
Quantum information, famously, does not take specific values, and so there is no sensible notion of conditioning on quantum information in a probability distribution. 
Only if the output quantum states all commute can one find such a description. 
Instead, the decoding task is to implement a measurement of the quantum output which can reliably determine which codeword was input to the channel. 
The analogous situation for classical channels is to regard the decoding task as a procedure for distinguishing between the conditional distributions of the {output} given the {input}. 
But, due to the nature of classical information, it is not necessary to regard the decoding task in this manner since the opposite conditional distribution of the input given the output is available by Bayes' rule. 

A simple approach for decoding CQ channels is to just measure the output quantum system for each transmitted codeword bit individually, which effectively transforms the CQ channel into a classical channel, and then proceed with classical decoding.
Alas, this approach cannot generally achieve the capacity of the underlying CQ channel. 
This is even true for the simplest family of binary-input CQ channels, where the two possible output states are pure states. 
For instance, this channel arises when using Bosonic coherent states as signals over a pure-loss Bosonic channel. 
The maximal rate at which information can be transmitted over a CQ channel is given by the Holevo capacity $\chi$, which in this case is simply $\chi=h_2(\tfrac12(1-\cos\theta))$, where $\cos\theta$ is the fidelity (overlap) of the two pure state outputs and $h_2$ is the binary entropy function. 
Even using the optimal basis for measurement in the individual measurement strategy only results in a Shannon capacity of the effective classical channel of $1-h_2(\tfrac12(1-\sin\theta))$, which is strictly smaller than $\chi$ for all $\theta\in(0,\tfrac\pi2)$. 
Indeed the ratio of the latter to the former goes to zero as $\theta$ goes to zero; in the above example this implies that the individual measurement strategy will be especially suboptimal in the regime of small amplitude coherent states.

Recently, Renes proposed a decoding algorithm for just such pure state CQ channels which generalizes belief propagation decoding~\cite{renes_2017}.
Like BP, it decodes the input codeword bitwise and operates by passing information, now quantum,  among the nodes of cycle-free Tanner graphs.
Unlike BP, though, it implements a measurement to distinguish between the quantum states associated to the two possible values of any chosen codeword bit. 
In the technical sense it is not strictly a belief propagation algorithm as it does not operate by marginalizing a probability distribution. 
But, like BP, it is an algorithm for inference.
There are other quite distinct notions of belief propogation and other applications in the quantum setting, and we comment on these at the end of this section. 

In fact, the algorithm of \cite{renes_2017} performs the optimal distinguishing measurement,  the so-called Helstrom measurement, which determines the value of any given codeword bit with optimal error probability. 
Like BP, this decoder is bitwise optimal on codes with tree Tanner graphs.
In \cite{rengaswamy_2020} Rengaswamy \emph{et al.} proposed several simplifications of the algorithm and coined the name \textit{belief propagation with quantum messages} (BPQM). 
Moreover, quite surprisingly, they also showed analytically that BPQM is blockwise optimal for a simple example code, i.e., it determines the entire input codeword with optimal error probability.
This was unexpected, as bitwise optimality does not guarantee blockwise optimality, and indeed BP is provably not blockwise optimal for the classical binary symmetric channel (BSC) in all parameter regimes. 

Both previous works on the subject of BPQM~\cite{renes_2017,rengaswamy_2020} lack an explicit and formal description of the algorithm. Some details are only implicitly suggested or illustrated by example. A more precise description of the algorithm reveals that the original formulation is, contrary to what is claimed in the previous works, not truly a message-passing algorithm. While BPQM can be described as an algorithm acting on qubits which are passed along the edges of the Tanner graph, the operations performed by the nodes do not only depend solely on those qubits. 
Due to this problem, the resulting quantum circuit is not efficient as claimed, but in fact exhibits a runtime that scales exponentially in the number of encoded bits, $k$.

In this paper we address several of the limitations and open questions of BPQM raised in~\cite{renes_2017,rengaswamy_2020}. 
We first list these in order of presentation before proceeding to sketch out the ideas and methods used for each. 
\begin{itemize}
  \item We give the first formally complete and concise description of BPQM. As part of that, we recognise a flaw in its original formulation which leads to a quantum circuit depth that grows exponentially in the code rank $k$ (\Cref{sec:description}).
  \item We prove analytically that BPQM performs block-optimal decoding for any pure-state CQ channel and any binary linear 
  classical code with tree Tanner graph (\cref{thm:optimality} in \Cref{sec:optimality}).
  \item We provide a new algorithm, which we call \emph{message-passing BPQM}. This new algorithm is truly a message-passing algorithm and does not exhibit the exponential cost present in BPQM. Message-passing BPQM can approximate BPQM to arbitrary precision. We provide analytical and numerical results characterizing the effect of the involved discretization errors (\Cref{sec:message_passing}).
  \item We suggest a strategy for extending BPQM to linear codes with non-tree Tanner graphs by making use of approximate quantum cloning. We numerically investigate this approach and demonstrate that it achieves near-optimal decoding performance in a simple example (\Cref{sec:nontree_graphs}).
\end{itemize}
Further, we provide the source code to reproduce all numerical results and figures included in this work.
\footnote{The source code is available at \href{https://github.com/ChriPiv/quantum-message-passing-paper}{https://github.com/ChriPiv/quantum-message-passing-paper}.}

\subsection{Description of BPQM}
Instead of working directly on the Tanner graph, BPQM utilizes what we call a \emph{message-passing graph} (MPG), a binary tree derived from the Forney-style representation of the factor graph of the code. 
BPQM sequentially decodes one codeword bit at a time, and a separate MPG is required for each individual codeword bit. 
As an example, consider the 5-bit code with two parity-checks $X_1+X_2+X_4=0$ and $X_1+X_3+X_5=0$ (also considered in \cite{rengaswamy_2020}). 
Its Tanner graph is depicted in \cref{fig:5bitTanner} and the MPG associated to the bit $X_1$ is shown in \cref{fig:5bitMPG}. 
Qubits are passed along the edges of the MPG, starting from the leaves and going towards the root of the graph.  
The qubits sent from the leaves are the output qubits obtained from the CQ channel. 
Every subsequent node receives two qubits and performs some processing in order to generate a single-qubit output which is then passed toward the root. 
There are two types of nodes, equality nodes denoted by `$=$' and check nodes denoted by `$+$', and the performed node operation depends on the corresponding type of the node.
Ultimately a single qubit is produced by the node operation at the root, which is then projectively measured to obtain the codeword bit estimate. 

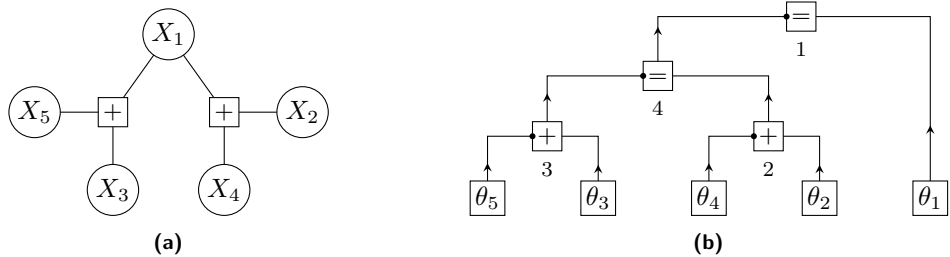
\begin{figure}[h]
  \centering
  \begin{subfigure}[b]{0.3\textwidth}
  \centering
  \tikzsetnextfilename{5bitTannerGraph}
  \begin{tikzpicture}
    \node[roundnode] (X1) {$X_1$};
    \node[squarenode] (c1) [below left=0.6cm and 0.3cm of X1] {$+$};
    \node[squarenode] (c2) [below right=0.6cm and 0.3cm of X1] {$+$};
    \node[roundnode] (X2) [right=0.5cm of c2] {$X_2$};
    \node[roundnode] (X4) [below=0.5cm of c2] {$X_4$};
    \node[roundnode] (X3) [below=0.5cm of c1] {$X_3$};
    \node[roundnode] (X5) [left=0.5cm of c1] {$X_5$};
    \draw (X1) to (c1);
    \draw (X1) to (c2);
    \draw (c2) to (X2);
    \draw (c2) to (X4);
    \draw (c1) to (X3);
    \draw (c1) to (X5);
  \end{tikzpicture}
  \caption{}
  \label{fig:5bitTanner}
  \end{subfigure}
  \hspace{5mm}
  \begin{subfigure}[b]{0.55\textwidth}
  \centering
  \tikzset{square/.style = {shape=regular polygon,regular polygon sides=4, text=black, draw=black, fill=white,inner sep=0.5mm, outer sep=0.0mm,text height=1.5ex,
    text depth=.25ex}}
  \tikzsetnextfilename{5bitMPG}
  \begin{tikzpicture}
[ 
    arrow/.style={postaction={decoration={markings,mark=at position 0.2 with {\arrow[>=stealth]{>}}},decorate}}
]
    \node[squarenode] (W1) {$\theta_1$};
    \node[squarenode] (W2) [left=1cm of W1] {$\theta_2$};
    \node[squarenode] (W4) [left=1cm of W2] {$\theta_4$};
    \node[squarenode] (W3) [left=1cm of W4] {$\theta_3$};
    \node[squarenode] (W5) [left=1cm of W3] {$\theta_5$};
    \node[squarenode] (ch1) [above left=0.4cm and 0.25cm of W2] {$+$};
    \node[below=.25mm of ch1] {\footnotesize 2};
    \node[fill,circle,inner sep=0.8pt] at (ch1.west) {};
    \node[squarenode] (ch2) [above left=0.4cm and 0.25cm of W3] {$+$};
    \node[below=.25mm of ch2] {\footnotesize 3};
    \node[fill,circle,inner sep=0.8pt] at (ch2.west) {};
    \node[squarenode] (eq1) [above left=1.2cm and 0.25cm of W4] {\cequal};
    \node[below=.25mm of eq1] {\footnotesize 4};
    \node[fill,circle,inner sep=0.8pt] at (eq1.west) {};
    \node[squarenode] (eq2) [above right=0.4cm and 1.5cm of eq1] {\cequal};
    \node[below=.25mm of eq2] {\footnotesize 1};
    \node[fill,circle,inner sep=0.8pt] at (eq2.west) {};

    \draw[arrow] (W2) |-  (ch1);
    \draw[arrow] (W4) |- (ch1);
    \draw[arrow] (W3) |- (ch2);
    \draw[arrow] (W5) |-  (ch2);
    \draw[arrow] (ch1) |-  (eq1);
    \draw[arrow] (ch2) |-  (eq1);
    \draw[arrow] (eq1) |-  (eq2);
    \draw[arrow] (W1) |- (eq2);
  \end{tikzpicture}
  \caption{}
  \label{fig:5bitMPG}
  \end{subfigure}
  \caption{\label{fig:5bitTannerMPG} Tanner graph and message-passing graph (MPG) for a five-bit code. The MPG is centered on $X_1$.}
  \end{figure}

However, the equality node operations must be uniformly controlled\footnote{A gate acting on the system $A$ and controlled by the $m$-qubit system $B$ is said to be a \emph{uniformly controlled gate} if the corresponding unitary acting on the joinst system $AB$ is of the form \smash{$\sum_{i=0}^{2^m-1}U^{(i)}_A\otimes\ket{i}\bra{i}_B$} where $U^{(0)}_A,\dots,U^{(2^m-1)}_A$ is a collection of unitaries acting on $A$ and $\ket{i}$ denotes a computational basis state on $B$.} on ancilla qubits which are generated by the node operations of all the preceding check nodes in the tree. 
That is, a different action must be taken on the qubits at the equality node for each of the different values of the ancilla qubits. 
These ancillas accumulate and their total number scales with the number of encoded bits $k$. Provided, as usual, that the universal set of gates supported by the quantum computer only contains gates acting on a bounded number of qubits, such uniformly-controlled gates generally need to be decomposed in a number of basic gates that scales exponentially in $k$. Previous works on BPQM do not explicitly address how to deal with this issue.

To illustrate why BPQM achieves optimal bit-wise decoding (as also described in \cite{renes_2017,rengaswamy_2020}), we employ a formalism where classical and quantum states and channels can be represented by a variant of Forney factor graphs. This formalism is based on~\cite{loeliger_2017,cao_2017} and is introduced in detail in~\cref{sec:fg_states_channels}.
The procedure to decode the $r$-th codeword bit $X_r$ can be understood by considering the factor graph representation of the CQ channel from $X_r$ to the channel outputs $Y_1,\dots,Y_n$ obtained from passing the codeword bits $X_1,\dots,X_n$ through the CQ channel $W[\theta]$, for a random choice of the codeword. 
The operations performed in the BPQM algorithm can be modeled as adding additional nodes to the factor graph, and the operations in BPQM are carefully chosen so that the graph simplifies two qubits at a time until the codeword bit value is encoded in just a single qubit. 
This qubit can then be projectively measured in the appropriate basis to determine the codeword bit value.

On its face, this description seems to differ strongly from the classical BP algorithm, where marginalization is performed by merging the nodes of the graph appropriately, sometimes referred to as “closing the boxes” in Forney-style factor graph language~\cite{loeliger_2004}. We show in~\cref{app:relation_to_bp} that BP can also be formulated in a channel contraction formalism similar to BPQM, which makes clear that BPQM is a quantum analog of BP.

\subsection{Block optimality}
In order to decode the complete codeword, the single-codeword bit procedure is executed sequentially for $k$ independent codeword bits. 
Ideally, each of these $k$ procedures would have access to the original $n$ channel output qubits. 
But since quantum information cannot be cloned, one must attempt to retrieve the original $n$-qubit state after some codeword bit has been decoded, to allow for the decoding for the subsequent codeword bits. 
This is nearly possible: The node operations involved in BPQM are all unitary, and can therefore be reversed. 
The only non-unitary part of the algorithm is the single-qubit projective measurement which produces the estimate for the codeword bit. %  in the $\{\ket{+},\ket{-}\}$ basis at the end.
Thus, for each codeword bit, the POVM elements associated to the two possible bit values are orthogonal projection operators.\footnote{Here orthogonality refers to the fact that the POVM elements $E_i$ fulfill $E_iE_j=\delta_{ij}E_i$.}

A central ingredient for the proof of optimality of the complete codeword decoder is the fact that an optimal decoder is realized by the so-called \emph{pretty-good measurement} (PGM)~\cite{belavkin_1975,hausladen_pretty_1994}, sometimes also called square-root measurement (also pointed out in \cite{rengaswamy_2020,rengaswamy_semiclassical_2021}).
This ensures the optimal measurement is a rank-one projective measurement, due to the linear structure of the code and earlier results on the optimality and orthogonality of the PGM~\cite{ban_optimum_1997,sasaki_quantum_1998,eldar_2001}. 
Orthogonality of the optimal measurement then implies that optimal decoding can in principle be performed bitwise. 
Simply construct projection operators $\Lambda_{r,j}$ for the $r$th codeword bit to have the value $j$ by summing the associated codeword projections $\Pi_{\vec x}$ over codewords $\vec x$ such that $x_r=j$. 
The intersection of the bitwise projection operators thus determines the codeword. 

The proof is therefore reduced to demonstrating that BPQM realizes $\Lambda_{r,j}$. 
To show this, a precise analysis is required that characterizes how the BPQM unitary acts on all possible pure state inputs. 
The result of this analysis is captured by~\cref{lem:purestate_evolution}. 
It can be considered a generalization of the factor graph contraction argument used to demonstrate bitwise optimality, as the latter result can be recovered by simply averaging out the corresponding codeword bits.
As a consequence of the structure of the projections, it follows that the (optimal) decoding performance of BPQM does not depend on how the $k$ independent codeword bits are chosen nor in which order these $k$ codeword bits are decoded.

\subsection{Message-passing BPQM}
As previously mentioned, BPQM is not truly a message-passing algorithm. Consider the qubit that is passed over some edge $e$ of the MPG towards the root. 
The state of this qubit is a superposition of states $\qstate{z}{\varphi} := \cos(\varphi /2)\ket{0}+(-1)^z\sin(\varphi /2)\ket{1}, z\in\{0,1\},\varphi\in(0,\pi)$ for different values of $\varphi$ (but fixed $z$). 
The two states corresponding to the two possible values of $z$ are related by a Pauli-$Z$ operation on the qubit.
The angle $\varphi$ takes values in some set $\{\varphi_{\vec{j}}\}_{\vec{j}\in\{0,1\}^{k_e-1}}$ where $k_e-1$ is the number of check nodes that precede $e$ in the graph.
The information specifying the angle is stored in the ancilla qubits produced by all the check nodes that precede $e$ in the MPG. 
This information is classical, as the states of the ancilla qubits corresponding to the different angles all commute.  
Equality node operations depend on the angle of the two incoming qubits, and for this reason they must be conditioned on the $k_e-1$ ancilla qubits produced by the preceeding check nodes. This is in turn the cause for the non-message-passing nature of BPQM that exhibits an exponential quantum circuit complexity.

We address this problem by introducing a quantum register that directly stores the angle of the qubits passed over the edge $e$. 
The messages passed around the MPG now consist of a qubit as well as the angle register. 
This on-the-fly bookkeeping of the involved angles is the essence of message-passing BPQM.
Were these new quantum registers able to represent all angles $\varphi\in(0,\pi)$ deterministically, they would need to be infinite-dimensional quantum systems. 
To overcome this issue, we instead store the angle in a discretized representation using only $B$ qubits.
%This is not novel, of course; representing a real number with a classical computer a priori also requires an infinite number of bits. 
In principle, the angle register could also be chosen to be a classical register, but in that case it is impossible to revert the action of BPQM at a later point as is required to decode multiple bits. Therefore, the quantum nature of the angle register is imperative for decoding the entire  codeword.

We provide an in-depth analysis of how the discretization errors propagate throughout the execution of the algorithm and how much they reduce the probability of successful decoding. 
The understanding of this error propagation is critical as it fundamentally differs from its classical counterpart: If BPQM operates under the assumption that a qubit involved in a node operation has the angle $\tilde{\varphi}$ which differs from the true value $\varphi$, then the output of the node operation will generally no longer be in a state that exhibits the aforementioned Pauli-$Z$ symmetry. 
The technical difficulty of analyzing the error is that the functions which need to be computed are not uniformly continuous. 
The central result is given by~\cref{thm:error_probab}, which states an explicit upper bound on the error of the decoder. As a direct consequence, one has to choose $B=\mathcal{O}(n,\log\frac{1}{\epsilon})$ in order to achieve a desired decoding error $\epsilon$. Consequently, the quantum circuit width and depth of message-passing BPQM are shown to scale as $\mathcal{O}(\poly n, \polylog \frac{1}{\epsilon})$. 
We also provide some numerical results on how the discretization errors effect the decoding error when decoding a single codeword bit of some $17$-bit binary linear code. As expected, we observe in~\cref{fig:17bitcode_numerical_results} that the decoding error decreases exponentially in $B$.

If one takes the value of $B$ to be fixed, which corresponds to a model of computation in which classical scalar operations take constant time, then the circuit depth of message-passing BPQM scales as $\mathcal{O}(kn)$ where $k$ is the number of encoded bits.
However, if the topology of the MPG allows us to parallelize some of the node operations, this complexity can in some cases even be reduced to quasilinear time $\mathcal{O}(k\log n)$.

\subsection{Codes with cycles}
Finally, we also study the question of how the BPQM decoding algorithm can be extended to codes with cycles in their Tanner graphs. 
As a guiding principle, we consider how classical BP is heuristically extended from tree to non-tree factor graphs by defining a fixed message-passing schedule at each node~\cite{richardson_2008}. 
In general, classical BP on non-tree factor graphs is not guaranteed to produce the correct marginals, let alone converge. 
Nevertheless, it has empirically been observed to give good approximate results in the case where the cycles are large. 
Similarly, when applying BPQM to non-tree graphs, one should not expect to achieve optimal decoding.

Fixing the message-passing schedule for non-tree graphs can be thought of as transforming the original decoding problem, including the actually observed output $\vec y$ as well as the descriptions of both the code $\mathcal C$ and the channel, to a new decoding problem involving a new output $\vec y'$, a tree code $\mathcal C'$, and a new channel model.
The code $\mathcal C'$ in question is obtained by ``unrolling'' the Tanner graph of $\mathcal{C}$ around the variable $X_r$ for a certain number of steps, giving each variable node a unique name.
The obtained graph is sometimes also referred to as \emph{computation tree}.
For example, consider decoding $X_1$ of the (6,3) binary linear code with Tanner graph depicted in~\cref{fig:6bitTanner}. 
By unrolling the Tanner graph around $X_1$ for $h=2$ steps, we obtain a Tanner graph defining the 9-bit code $\mathcal{C}'$, shown in \cref{fig:9bitTanner}. 
A codeword $\vec X$ of $\mathcal C$ specified by $X_1,\dots, X_6$ can be transformed to a codeword $\vec X'$ of $\mathcal C'$ by taking $X'_j=X_j$ for $j=1,\dots,6$, and $X'_7=X_4$, $X'_{8}=X_3$, and $X'_{9}=X_6$. 
This is depicted in \cref{fig:6bitComputational}. 
The same duplication procedure can be applied to the channel output $\vec y$, and the channel parameters for duplicated bits can be taken to be the same as their original values. 
Then, BP decoding of the resulting $\vec y'$ using the duplicated channel model will give an estimate of the original input $X_1$. 
One can check that running classical BP on $\mathcal{C}$ to decode $X_r$ given a noisy channel output $\vec{y}\in\mathbb{F}_2^n$ for $h$ time steps is equivalent to running classical BP on $\mathcal{C}'$ given the noisy channel output $\vec{y}'$. 

\begin{figure}[h]
\centering
\begin{subfigure}[b]{0.3\textwidth}
\centering
\tikzsetnextfilename{6bitTannerGraph}
  \begin{tikzpicture}
    \node[roundnode] (X1) {$X_1$};
    \node[squarenode] (c1) [below left=0.6cm and 0.3cm of X1] {$+$};
    \node[squarenode] (c2) [below right=0.6cm and 0.3cm of X1] {$+$};
    \node[roundnode] (X2) [right=0.5cm of c2] {$X_2$};
    \node[roundnode] (X4) [below=0.5cm of c2] {$X_4$};
    \node[roundnode] (X3) [below=0.5cm of c1] {$X_3$};
    \node[roundnode] (X5) [left=0.5cm of c1] {$X_5$};
    \node[squarenode] (c3) [below right=0.5cm and 0.3 cm of X3] {$+$};
    \node[roundnode] (X6) [below=0.5cm of c3] {$X_6$};
    \draw (X1) to (c1);
    \draw (X1) to (c2);
    \draw (c2) to (X2);
    \draw (c2) to (X4);
    \draw (c1) to (X3);
    \draw (c1) to (X5);
    \draw (c3) to (X6);
    \draw (c3) to (X3);
    \draw (c3) to (X4);
  \end{tikzpicture}
  \caption{}
  \label{fig:6bitTanner}
  \end{subfigure}
  \hspace{2mm}
  \begin{subfigure}[b]{0.3\textwidth}
  \centering
\tikzsetnextfilename{9bitTannerGraph}
  \begin{tikzpicture}
    \node[roundnode] (X1) {$X'_1$};
    \node[squarenode] (c1) [below left=0.6cm and 0.3cm of X1] {$+$};
    \node[squarenode] (c2) [below right=0.6cm and 0.3cm of X1] {$+$};
    \node[roundnode] (X2) [right=0.5cm of c2] {$X'_2$};
    \node[roundnode] (X4) [below=0.5cm of c2] {$X'_4$};
    \node[roundnode] (X3) [below=0.5cm of c1] {$X'_3$};
    \node[roundnode] (X5) [left=0.5cm of c1] {$X'_5$};
    \node[squarenode] (c3) [below=0.5cm  of X3] {$+$};
    \node[squarenode] (c4) [below=0.5cm  of X4] {$+$};
    \node[roundnode] (X6) [below=0.5cm of c3] {$X'_6$};
    \node[roundnode] (X2_2) [left=0.5cm of c3] {$X'_7$};
    \node[roundnode] (X6_2) [below=0.5cm of c4] {$X'_9$};
    \node[roundnode] (X3_2) [right=0.5cm of c4] {$X'_8$};
    \draw (X1) to (c1);
    \draw (X1) to (c2);
    \draw (c2) to (X2);
    \draw (c2) to (X4);
    \draw (c1) to (X3);
    \draw (c1) to (X5);
    \draw (X3) -- (c3) -- (X6);
    \draw (X2_2) -- (c3);
    \draw (X4) -- (c4) -- (X6_2);
    \draw (X3_2) -- (c4);
  \end{tikzpicture}
  \caption{}
  \label{fig:9bitTanner}
  \end{subfigure}
  \hspace{2mm}
  \begin{subfigure}[b]{0.3\textwidth}
  \centering
\tikzsetnextfilename{6bitComputationalGraph}
  \begin{tikzpicture}
    \node[roundnode] (X1) {$X_1$};
    \node[squarenode] (c1) [below left=0.6cm and 0.3cm of X1] {$+$};
    \node[squarenode] (c2) [below right=0.6cm and 0.3cm of X1] {$+$};
    \node[roundnode] (X2) [right=0.5cm of c2] {$X_2$};
    \node[roundnode] (X4) [below=0.5cm of c2] {$X_4$};
    \node[roundnode] (X3) [below=0.5cm of c1] {$X_3$};
    \node[roundnode] (X5) [left=0.5cm of c1] {$X_5$};
    \node[squarenode] (c3) [below=0.5cm  of X3] {$+$};
    \node[squarenode] (c4) [below=0.5cm  of X4] {$+$};
    \node[roundnode] (X6) [below=0.5cm of c3] {$X_6$};
    \node[roundnode] (X2_2) [left=0.5cm of c3] {$X_4$};
    \node[roundnode] (X6_2) [below=0.5cm of c4] {$X_6$};
    \node[roundnode] (X3_2) [right=0.5cm of c4] {$X_3$};
    \draw (X1) to (c1);
    \draw (X1) to (c2);
    \draw (c2) to (X2);
    \draw (c2) to (X4);
    \draw (c1) to (X3);
    \draw (c1) to (X5);
    \draw (X3) -- (c3) -- (X6);
    \draw (X2_2) -- (c3);
    \draw (X4) -- (c4) -- (X6_2);
    \draw (X3_2) -- (c4);
  \end{tikzpicture}
  \caption{}
  \label{fig:6bitComputational}
  \end{subfigure}
  \caption{Unrolling the Tanner graph of a code $\mathcal C$ to generate a new code $\mathcal C'$ and computation tree of $\mathcal C$. (a) The Tanner graph of a 6-bit code $\mathcal C$. (b) Tanner graph of the 9-bit code $\mathcal C'$ defined by unrolling the original Tanner graph around $X_1$ to depth 2. (c) $h=2$ computation tree of $\mathcal C$, obtained by identifying the variables of $\mathcal C'$ with those of $\mathcal C$.}
  \label{fig:unrollingcomputational}
\end{figure}
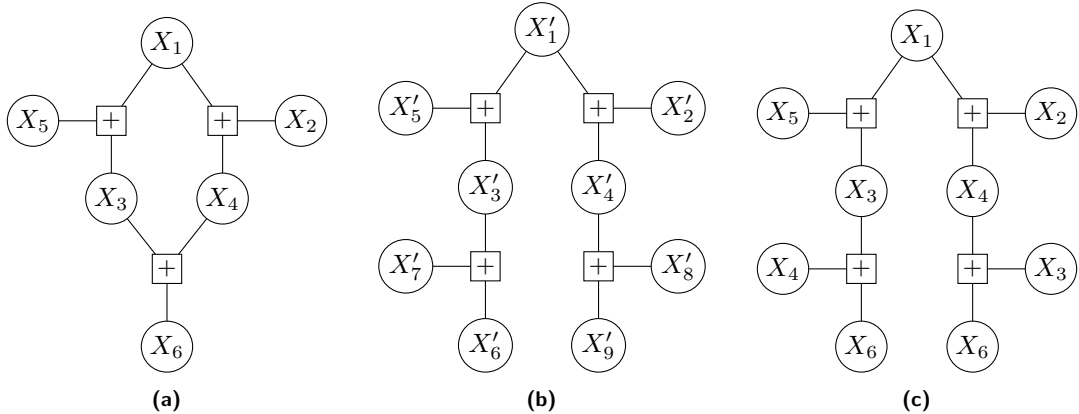

Returning to the case of a CQ channel, we encounter a major problem if we try to apply the same procedure: It is impossible to copy the channel output qubits due to the no-cloning theorem. 
We address this problem by employing approximate cloning, which generates imperfect clones of the channel output qubits. 
Choosing the channel parameters of the duplicated outputs in accordance with the cloning procedure, numerical results for an 8-bit code in~\cref{sec:cloning_numerics} show that the resulting decoder can significantly outperform the best possible classical decoder\footnote{Here the best classical decoding refers to measuring the channel outputs in the $\ket{\pm}$ basis and doing classical maximum likelihood decoding.} using this strategy. 
In fact, it nearly reaches optimal decoding performance.

Note that due to the approximate cloning issue, BPQM decoding on non-tree graphs exhibits a strongly different characteristic to its classical counterpart: Increasing the number of unrolling steps $h$ does not always improve the result. After some point, decoder performance deteriorates, as a larger value of $h$ implies that more approximate cloning operations must be performed and therefore the quality of the involved qubits decreases.

We note that the described unrolling strategy is very unlikely to be the best choice to generate $\mathcal{C}'$ for any possible code $\mathcal{C}$, especially for codes with many small cycles, where the number of approximate cloning operations would increase very quickly with the depth $h$ of the unrolling. Most likely, different families of codes will require different strategies to generate the best possible choice of $\mathcal{C}'$. 
Similarly, which kind of codes lend themselves best to BPQM decoding is an interesting  question for future research. 
LDPC codes are a promising family of codes, as they have proven to exhibit excellent performance for transmission over classical channels and their Tanner graphs exhibit few and large cycles. 

\subsection{Other notions of quantum belief propagation}
The notion of belief propagation has been applied in the quantum setting in many different previous works; here we comment on their relation to the present work. 
Most closely related is the use of belief propagation for decoding quantum information subject to Pauli errors, as studied by Poulin in~\cite{poulin_optimal_2006,poulin_iterative_2008}, as well as many others since~\cite{wang_enhanced_2012,criger_multi-path_2018,liu_neural_2019,rigby_modified_2019,panteleev_degenerate_2019,li_pseudocodeword-based_2019,roffe_decoding_2020,li_pseudocodeword-based_2020}.
Here, however, classical BP is sufficient: The task is to infer which error occured from the classical syndrome information, which only involves the classical conditional probability of the former given the latter. 
Conversely, \cite{kasi_2020} proposes decoding LDPC codes over a classical channel by means of quantum annealing. 
This does not involve decoding information from quantum systems, but instead using quantum systems to perform a classical computation.  
Slightly further afield is the work of Leifer and Poulin on developing quantum graphical models and belief propagation algorithms for marginalization of density operators~\cite{leifer_quantum_2008}. 
This issue does not immediately arise here, as essentially only the leaves of Forney factor graphs we consider involve quantum systems. 
Finally, BP is also related to the Bethe-Peirels approximation of the free energy in statistical mechanics~\cite{betheStatisticalTheorySuperlattices1935,peierls_statistical_1936}. 
In this approximation one computes the free energy from the probability distribution using only the marginals describing single sites and neighboring sites. 
As first observed by Yedidia, Freeman, and Weiss, the fixed points of BP are stationary points of the approximation of the free energy~\cite{yedidia_generalized_2000,yedidia_understanding_2003,yedidia_constructing_2005}. 
Hastings and Poulin have constructed quantum belief propagation algorithms for approximating the free energy in quantum statistical mechanics~\cite{hastings_quantum_2007,poulin_markov_2011}.
Further variations of this approach were investigated by Cao and Vontobel~\cite{cao_quantum_2016,cao_2017}.

\section{Preliminaries}
\subsection{Factor graphs and classical belief propagation decoding}
Factor graphs are a graphical representation of the factorization of a function of several variables, e.g., a joint probability distribution, into a product of separate terms each involving only a subset of the variables. 
The standard variant of factor graphs~\cite{kschischang_2001} uses a bipartite graph consisting of two types of nodes: variable and factor nodes. 
Variable nodes are usually represented with circles and factor nodes with rectangular boxes. 
Consider a factorizable function
\begin{equation}
  f(x_1,\dots,x_n)=\prod\limits_{j\in J} f_j(Z_j)
\end{equation}
where $J$ is some finite index set and the $f_j$ are functions having some subset $Z_j$ of $\{x_1,\dots,x_n\}$ as arguments. The corresponding factor graph consists of $n$ variable nodes $X_1,\dots,X_n$ as well as a factor node $f_j$ for each $j\in J$. A variable node $X_i$ and a factor node $f_j$ are connected by an edge if and only if $x_i\in Z_j$. 

Consider an $(n,k)$ binary linear code $\mathcal{C}$, i.e., a linear subspace of $\mathbb{F}_2^n$ of dimension $k$. A parity-check matrix $H$ of this code directly tells us how to represent the membership indicator function
\begin{equation}
  I_{\mathcal{C}}:\mathbb{F}_2^n\rightarrow \{0, 1\} \, , \, I_{\mathcal{C}}(x) := \begin{cases}1 \text{ if } x\in \mathcal{C} \\ 0 \text{ else}\end{cases}
\end{equation}
in terms of such a factor graph.
The parity-checks manifest as check factor nodes (denoted by `$+$') in the factor graph, which represent the function
\begin{equation}
  x_1,\dots,x_m \mapsto \begin{cases}1 \text{ if } x_1+\dots+x_m \equiv 0 \pmod{2} \\ 0 \text{ else}\end{cases}
\end{equation}
where $m$ is the number of edges connected to the check node.
The resulting factor graph is commonly called a Tanner graph of the code. 
Generally a Tanner graph is not unique, as there exist different parity-check matrices $H$ associated with the code $\mathcal{C}$. 
Consider again the $(5,3)$ code with Tanner graph in \cref{fig:5bitTanner}. 
It has the following parity-check and generator matrices:
\begin{equation}
  H = \begin{pmatrix}1&1&0&1&0\\1&0&1&0&1\end{pmatrix} \, , \quad G = \begin{pmatrix}1&0&0&1&1\\0&1&0&1&0\\0&0&1&0&1\end{pmatrix} \, .
\end{equation}

Throughout the majority of this work, it will prove advantageous to work with a different graphical representation of factorizable functions, namely the so-called \emph{Forney-style factor graphs}~\cite{forney_2001,loeliger_2004}. Here variables are represented by the edges of the graph whereas the factor terms of the function are represented by the nodes of the graph. Similarly, an edge and a node of the graph are connected if and only if the corresponding variable is part of the corresponding factor term. The Forney-style representation of the Tanner graph of the 5-bit code is depicted in~\cref{fig:5bitcode_forney}. 
When translating a factor graph to the Forney-style representation, one might encounter the issue that some variables are involved in more than two factors, as depicted in the example in~\cref{fig:eqnode_example}. 
In that case, equality factor nodes must be introduced in the Forney-style factor graph. 
Equality nodes are depicted by the symbol $=$ and represent the function that returns $1$ if all connected variables have the same value and $0$ otherwise. Thus, the Forney-style representation of the Tanner graph of any binary linear code contains only check nodes and equality nodes.

\begin{figure}
  \centering
  % \begin{subfigure}[b]{0.18\textwidth}
  % \begin{adjustbox}{width=\textwidth}
  % \begin{tikzpicture}
  %   \node[roundnode] (X1) {$X_1$};
  %   \node[squarenode] (c1) [above right=0.1cm and 0.6cm of X1] {$+$};
  %   \node[squarenode] (c2) [below right=0.1cm and 0.6cm of X1] {$+$};
  %   \node[roundnode] (X2) [below=0.5cm of c2] {$X_2$};
  %   \node[roundnode] (X4) [right=0.5cm of c2] {$X_4$};
  %   \node[roundnode] (X3) [right=0.5cm of c1] {$X_3$};
  %   \node[roundnode] (X5) [above=0.5cm of c1] {$X_5$};
  %   \draw (X1) to (c1);
  %   \draw (X1) to (c2);
  %   \draw (c2) to (X2);
  %   \draw (c2) to (X4);
  %   \draw (c1) to (X3);
  %   \draw (c1) to (X5);
  % \end{tikzpicture}
  % \end{adjustbox}
  % \caption{}
  % \label{fig:5bitcode_tanner}
  % \end{subfigure}
  \begin{subfigure}[b]{0.25\textwidth}
  \begin{adjustbox}{width=\textwidth}
  \tikzsetnextfilename{5bitForney}
  \begin{tikzpicture}
    \node[squarenode] (c1) [] {$+$};
    \node[squarenode] (c2) [right=1.5cm of c1] {$+$};
    \draw (c1) to[out=90,in=180] ++(0.9cm,1.0cm) node[above] {$X_1$} to[out=0,in=90] (c2);
    \draw (c1) to node[left] {$X_2$} ++(-0.4cm,-1.0cm);
    \draw (c1) to node[right] {$X_4$} ++(0.4cm,-1.0cm);
    \draw (c2) to node[left] {$X_3$} ++(-0.4cm,-1.0cm);
    \draw (c2) to node[right] {$X_5$} ++(0.4cm,-1.0cm);
  \end{tikzpicture}
  \end{adjustbox}
  \caption{}
  \label{fig:5bitcode_forney}
  \end{subfigure}
  \hspace{5mm}
  \begin{subfigure}[b]{0.2\textwidth}
  \begin{adjustbox}{width=\textwidth}
  \tikzsetnextfilename{StandardToNormal}
  \begin{tikzpicture}
    \node[roundnode] (X4) {$X_4$};
    \node[quadnode] (f1) [above right=0.6cm and 0.8cm of X4.center] {$f_1$};
    \node[quadnode] (f2) [above right=0.0cm and 0.8cm of X4.center] {$f_2$};
    \node[quadnode] (f3) [above right=-0.6cm and 0.8cm of X4.center] {$f_3$};
    \node[roundnode,inner sep=0pt] (X1) [right=0.6cm of f1.center] {$X_1$};
    \node[roundnode,inner sep=0pt] (X2) [right=0.6cm of f2.center] {$X_2$};
    \node[roundnode,inner sep=0pt] (X3) [right=0.6cm of f3.center] {$X_3$};
    \node[rotate=270,anchor=center] (equiv) [below right=0.1cm and 0.4cm of f3] {$\rightarrow$};
    \node[squarenode] (eq) [below right=2.0cm and -0.1cm of X4] {\cequal};
    \node[quadnode] (f12) [above right=0.6cm and 0.8cm of eq.center] {$f_1$};
    \node[quadnode] (f22) [above right=0.0cm and 0.8cm of eq.center] {$f_2$};
    \node[quadnode] (f32) [above right=-0.6cm and 0.8cm of eq.center] {$f_3$};

    \draw (X4) to (f1);
    \draw (X4) to (f2);
    \draw (X4) to (f3);
    \draw (f1) to (X1);
    \draw (f2) to (X2);
    \draw (f3) to (X3);
    \draw (eq) to node[above left] {$X_4$} (f12);
    \draw (eq) to (f22);
    \draw (eq) to (f32);
    \draw (f12) to ++(1.0cm,0) node[above] {$X_1$};
    \draw (f22) to ++(1.0cm,0) node[above] {$X_2$};
    \draw (f32) to ++(1.0cm,0) node[above] {$X_3$};
  \end{tikzpicture}
  \end{adjustbox}
  \caption{}
  \label{fig:eqnode_example}
  \end{subfigure}
  \hspace{5mm}
  \begin{subfigure}[b]{0.28\textwidth}
  \begin{adjustbox}{width=\textwidth}
  \tikzsetnextfilename{ChannelTanner}
  \begin{tikzpicture}
    \node[roundnode] (X1) {$X_1$};
    \node[squarenode] (c1) [above right=0.1cm and 0.6cm of X1] {$+$};
    \node[squarenode] (c2) [below right=0.1cm and 0.6cm of X1] {$+$};
    \node[roundnode] (X2) [below=0.5cm of c2] {$X_2$};
    \node[roundnode] (X4) [right=0.5cm of c2] {$X_4$};
    \node[roundnode] (X3) [right=0.5cm of c1] {$X_3$};
    \node[roundnode] (X5) [above=0.5cm of c1] {$X_5$};
    \node[quadnode] (W1) [above left=0.4cm and -0.6cm of X1] {$P_{Y=y_1|X}$};
    \node[quadnode] (W5) [above right=0.3cm and 0.3cm of X5] {$P_{Y=y_5|X}$};
    \node[quadnode] (W3) [above right=0.4cm and -0.5cm of X3] {$P_{Y=y_3|X}$};
    \node[quadnode] (W4) [below right=0.4cm and -0.5cm of X4] {$P_{Y=y_4|X}$};
    \node[quadnode] (W2) [below right=0.3cm and 0.3cm of X2] {$P_{Y=y_2|X}$};
    \draw (X1) to (c1);
    \draw (X1) to (c2);
    \draw (c2) to (X2);
    \draw (c2) to (X4);
    \draw (c1) to (X3);
    \draw (c1) to (X5);
    \draw (X1) to (W1);
    \draw (X2) to (W2);
    \draw (X3) to (W3);
    \draw (X4) to (W4);
    \draw (X5) to (W5);
  \end{tikzpicture}
  \end{adjustbox}
  \caption{}
  \label{fig:5bitcode_tanner_channelnodes}
  \end{subfigure}
  \caption{%(a) Tanner graph of the 5-bit code. 
  (a) Forney-style representation of the Tanner graph of the 5-bit code. 
  (b) The graph on top is an example of a (standard) factor graph containing a variable $X_4$ that is connected to more than two factor nodes. The corresponding Forney-style representation of the factor graph, depicted on the bottom, correspondingly contains an equality node. 
  (c) (Standard) factor graph characterizing the joint distribution of the codeword bits of the 5-bit code given that some channel output $\vec{y}$ was observed.}
\end{figure}
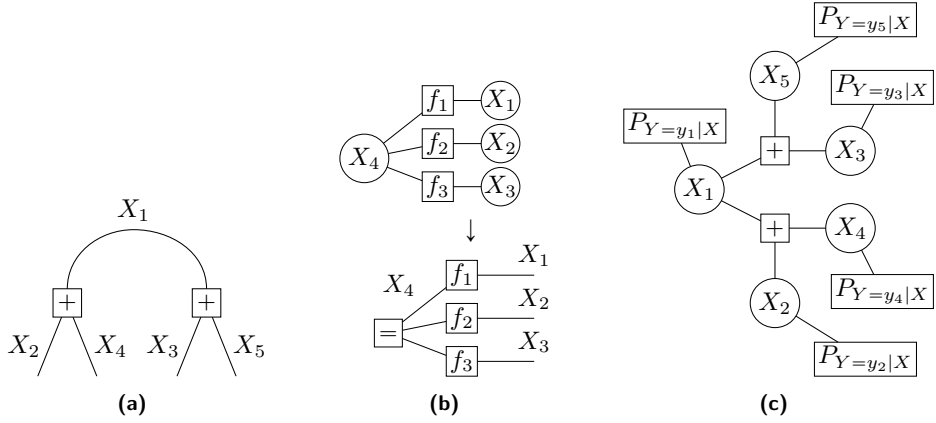

Given a multivariate function $f(x_1,\dots,x_n)$, consider the task of finding the marginal w.r.t. the variable $x_i$, i.e., the function
\begin{equation}
  f_i(x_i) := \sum\limits_{x_1,x_2,\dots,x_{i-1},x_{i+1},\dots,x_n} f(x_1,\dots,x_n) \, .
\end{equation}
In general this task is very expensive as the runtime of the brute force computation scales exponentially in $n$. 
However, if we are given some additional information about the factorizability of $f$ in form of a factor graph, then the belief propagation (BP) algorithm can utilize this additional structure to approximate the marginal at a reduced complexity. 
BP is a message-passing algorithm, i.e., it can be realized as a sequence of messages being passed between nodes of the factor graphs.\footnote{Belief propagation can be formulated either for the standard factor graphs or for Forney-style factor graphs.} Each node receives some input messages from its neighbors, performs a certain computation given that input, and produces output messages that it passes to its neighbors. Under the assumption that the factor graph is a tree, BP is able to compute the exact marginals, otherwise it will generally only return an approximate result. For a more complete description of factor graphs and classical belief propagation, we refer the reader to~\cite{loeliger_2004,richardson_2008}.

BP can be directly applied to the task of communicating classical information over a noisy classical channel described by the transition probability $P_{Y|X}$ which maps some bit $x\in\mathbb{F}_2$ into some output alphabet $\mathcal{Y}$.
Reliable communication is achieved by coding our message using some $(n,k)$ binary linear code.
Assume that the transmitted codeword is chosen uniformly randomly.
If we observe the output of the channel to be some vector $\vec{y}\in\mathcal{Y}^n$, that new knowledge about the joint probability distribution of the codeword bits is reflected by adding new factor nodes to the codeword bits in the factor graph, as depicted in~\cref{fig:5bitcode_tanner_channelnodes}. 
Marginalizing this new joint probability distribution with respect to the $i$-th codeword bit $X_i$ using BP gives the conditional probability density $\mathbb{P}_{X_i|\vec{Y}=\vec{y}}$. 
From this we can easily compute the bitwise maximum a-posteriori (bit-MAP) estimate for the $i$-th bit, defined by $\argmax_{x\in\{0,1\}}\mathbb{P}_{X_i|\vec{Y}=\vec{y}}(x)$.

\subsection{Commmunication over a pure state classical-quantum channel}
We consider the task of achieving reliable communication of classical information over a CQ channel, i.e., a channel with classical input and quantum output.
We restrict ourselves to the study of pure-state CQ channels $W[\theta]$ that produce qubit states of the form
\begin{equation}
  \{0,1\} \ni x \mapsto \qstate{x}{\theta} := \cos\frac{\theta}{2}\ket{0} + (-1)^x\sin\frac{\theta}{2}\ket{1}
\end{equation}
for some angle $\theta\in(0,\pi)$. Note that any pure-state CQ channel taking a single bit as input can be brought into this form by performing a corresponding unitary transformation on the output of the channel, since the two possible pure-state outcomes span a two-dimensional subspace of the total Hilbert space.

To achieve reliable communication of classical information over this channel, we can redundantly encode our information using a binary linear code, completely analogous to the procedure for transmitting classical information over a noisy classical channel. The only difference is that the channel outputs $Y_1,\dots,Y_n$ now denote two-dimensional quantum systems (i.e.~qubits) instead of classical random variables. The question now arises, how one can decode the original codeword given access to this quantum system.

A first naive approach is to convert the CQ channel to a classical channel, by estimating for each channel output separately if it is more likely in the state $\qstate{0}{\theta}$ or $\qstate{1}{\theta}$. This can be done optimally by the Helstrom measurement~\cite{helstrom_1976}, which in this case corresponds to measuring the corresponding qubit in the $\ket{\pm}$ basis. The resulting effective classical channel is a binary symmetric channel with error probability $p=\frac{1-\sin\theta}{2}$ and capacity $C=1-h(p)$.
We refer to this decoding strategy as a \emph{classical decoder}, since we can make use of any classical decoding algorithm and we do not make use of the fact that the channel output is a quantum system.

However, a classical decoder cannot realize the best possible decoding performance. This is illustrated by the celebrated Holevo-Schumacher-Westmoreland (HSW) theorem which states that the maximal possible rate of classical information that can be reliably transmitted is given by the Holevo capacity of the channel, which for $W[\theta]$ turns out to be $\chi=h_2(\frac{1+\cos\theta}{2})$ where $h_2$ is the binary entropy function. It can be easily checked that $\chi$ can be significantly larger than $C$, and in fact one find that $\chi /C\rightarrow\infty$ in the limit $\theta\rightarrow 0$. 
To approach the Holevo capacity, it is therefore imperative to design truly quantum decoders which explicitly make use of the quantum nature of the channel output. 
See Guha~\cite{guha_structured_2011} for more on this point. %In contrary to classical BP, BPQM is a truly quantum decoder and can generally significantly outperform the best classical decoder.

When the Tanner graph of the considered code is a tree, classical BP can achieve bitwise optimal decoding (i.e.~bit-MAP decoding) of information transmitted over a classical channel.
Message passing on such tree graphs is the foundation of belief propagation in general, though codes with cycle-free Tanner graphs are known to perform badly, since they exhibit a distance $d\leq 2$ as soon as the rate $k/n$ is greater or equal one half~\cite{etzion_1999}. 
Interestingly, when decoding codes with cycle-free Tanner graphs transmitted over pure-state CQ channels with BPQM, one does not only attain bitwise optimal decoding, but even the block optimal decoding.

\subsection{Representing states and channels using Forney-style factor graphs}\label{sec:fg_states_channels}
This section explains how we use the Forney-style factor graph formalism to represents classical and quantum states and channels throughout this work. The introduced formalism uses the language of Forney-style factor graphs, but at its heart it is very closely related to the tensor network formalism (for an introduction, see~\cite{bridgeman_hand-waving_2017}). 
The main advantage is that the formalism naturally represents states and operations that involve both classical and quantum systems.
The concepts introduced in this section are necessary to understand the central ideas in the proof of bit optimality of BPQM in~\cref{sec:contraction}, but they have no further importance for the rest of this manuscript.

For the rest of this section we will always restrict ourselves to (random) variables over finite alphabets as well as finite-dimensional quantum systems. 
Consider a Forney-style factor graph representing the factorization of some function in $m+l$ variables $f(V_1,\dots,V_m,H_1,\dots,H_l)$ where one subset of the variables $V_1,\dots,V_m$, called the 'visible' variables, are represented by half-edges, i.e., they are each connected to exactly one factor node. 
The other variables $H_1,\dots,H_l$ are 'hidden' inside the graph, i.e., they are each connected to exactly two factor nodes. We define the \emph{exterior function} of the factor graph to be the joint marginal of the visible variables, i.e.,
\begin{equation}
  v_1,\dots,v_m \mapsto \sum\limits_{h_1,\dots,h_l} f(v_1,\dots,v_m,h_1,\dots,h_l) \, .
\end{equation}
We define the \emph{partition function} of the factor graph to be the sum over all variables, i.e.,
\begin{equation}
  \sum\limits_{v_1,\dots,v_m,h_1,\dots,h_l} f(v_1,\dots,v_m,h_1,\dots,h_l) \, .
\end{equation}

Consider some rank-$k$ tensor\footnote{Note that we use the term \emph{tensor} in the sense of a multidimensional array (as is common in machine learning or tensor network theory) and not of a multilinear map. The rank of a tensor is its number of indices, e.g., a rank-$1$ tensor is a vector and a rank-$2$ tensor is a matrix.} $T_{i_1,\dots,i_k}$ where the indices $i_j$ take values in the alphabets $\mathcal{I}_j$.
We say that a factor graph $G$ represents the tensor $T$ if $G$ has $k$ visible variables $I_1,\dots,I_k$ defined on the alphabets $\mathcal{I}_1,\dots,\mathcal{I}_k$ and the exterior function of $G$ is exactly equal to
\begin{equation}
  i_1,\dots,i_k \mapsto T_{i_1,\dots,i_k}\,.
\end{equation}

The contraction of two tensors can be represented by the factor graph obtained by connecting half-edges of the factor graphs representing the two original tensors. Consider for instance two factor graphs $G$ and $G'$ representing the tensors $T_{i_1,\dots,i_k}$ and $T'_{j_1,\dots,j_{k'}}$ where the indices $i_l$ and $j_l$ run over the same alphabets for $l=1,\dots,m$, for $m\leq \min\{k,k'\}$. Consider the factor graph $\tilde{G}$ defined by connecting the half-edges for $I_l$ and $J_l$ between $G$ and $G'$ for $l=1,\dots,m$. $\tilde{G}$ then represents the tensor $\tilde{T}$ obtained by contracting $T$ and $T'$:
\begin{equation}
  \tilde{T}_{i_{m+1},\dots,i_k,j_{m+1},\dots,j_{k'}} = \sum\limits_{k_1,\dots,k_m} T_{k_1,\dots,k_m,i_{m+1},\dots,i_k}T'_{k_1,\dots,k_m,j_{m+1},\dots,j_{k'}}\,.
\end{equation}

The central quantities and operations in probability theory can be expressed in terms of tensors and tensor contractions. 
For instance, a probability distribution $P_X$ of some random variable $X$ defined over the alphabet $\mathcal{X}$ can be simply thought of as a rank-$1$ tensor where the index takes values in $\mathcal{X}$. 
The following graph is a trivial example of a factor graph representing the distribution $P_X$:
\begin{equation}
\tikzsetnextfilename{PXfigure}
  \begin{tikzpicture}[baseline]
    \node[quadnode] (PX) {$P_X$};
    \draw (PX.east) to node[above] {$X$} ++(1.5cm,0);
  \end{tikzpicture}\,\,.
\end{equation}
As another example, consider the Forney-style representation of the Tanner graph of a binary linear $(n,k)$ code where every variable is represented by a half-edge, as in~\cref{fig:5bitcode_randomstate} for the 5-bit code. 
Since the constructed factor graph corresponds to the membership indicator function, it represents the probability distribution of a uniformly randomly chosen codeword $(X_1,\dots,X_n)$, up to a normalization constant ($1/8$ in the case of the 5-bit code).

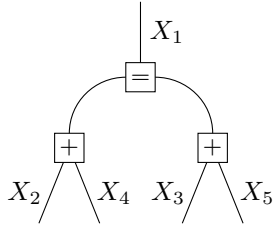
\begin{figure}
  \centering
    \tikzsetnextfilename{5bitFFG}
  \begin{tikzpicture}
    \node[squarenode] (c1) [] {$+$};
    \node[squarenode] (c2) [right=1.5cm of c1] {$+$};
    \node[squarenode] (eq) [above right=0.55cm and 0.55cm of c1] {\cequal};
    \draw (c1) to[out=90,in=180] (eq) to[out=0,in=90] (c2);
    \draw (c1) to node[left] {$X_2$} ++(-0.4cm,-1.0cm);
    \draw (c1) to node[right] {$X_4$} ++(0.4cm,-1.0cm);
    \draw (c2) to node[left] {$X_3$} ++(-0.4cm,-1.0cm);
    \draw (c2) to node[right] {$X_5$} ++(0.4cm,-1.0cm);
    \draw (eq) to node[right] {$X_1$} ++(0,1.0cm);
  \end{tikzpicture}
  \caption{Factor graph representing the distribution of a uniformly random codeword of the 5-bit code.}
  \label{fig:5bitcode_randomstate}
\end{figure}
In general, when we represent probability distributions (and also channels as we will see later) with factor graphs in this manuscript, we do not bother to properly normalize the tensors.
This is done to simplify the notation.
The correct normalization can always be inferred by context, e.g., the entries of a tensor representing a probability distribution should always sum up to one.

Channels are linear maps that are uniquely represented by the matrix containing all the transition probabilities, and therefore they can also thought of as a tensor. 
Moreover, the application of a channel on some state consists of a contraction of the two involved tensors which can accordingly be represented using our factor graphs.
Note that the tensors corresponding to probability distributions and channels contain only non-negative entries.
Correspondingly, the contraction of two tensors with non-negative entries is itself again a tensor with non-negative entries.

Consider as an example three random variables $X_1,X_2,X_3$ distributed as $P_{X_1},P_{X_2},P_{X_3}$ as well as a channel $P_{Z|X_1,X_2}$ that maps $X_1,X_2$ to a random variable denoted by $Z$. The joint state of $X_3$ and $Z$ is represented by the following factor graph:
\begin{equation}
\tikzsetnextfilename{PX3Figure}
\begin{tikzpicture}[baseline=-20pt]
\node[quadnode] (PX1) {$P_{X_1}$};
\node[quadnode] (PX2) [below=0.2cm of PX1] {$P_{X_2}$};
\node[quadnode] (PX3) [below=0.2cm of PX2] {$P_{X_3}$};
\node[quadnode] (channel) [below right=-0.2cm and 0.4cm of PX1] {$P_{Z|X_1,X_2}$};
\draw (PX1) to (channel);
\draw (PX2) to (channel);
\draw (channel.east) to node[above] {$Z$} +(0.5cm,0);
\draw (PX3) to node[above] {$X_3$} +(2.7cm,0);
\end{tikzpicture}\,\,.
\end{equation}

For further illustration, we provide a non-trivial example of a factor graph representation of a channel that will play an important role later in this work. 
Consider the Forney-style representation of the Tanner graph of some $(n,k)$ binary linear code, such that the half-edges consist exactly of the codeword bits $X_i, i=1,\dots,n$, and one of the codeword bits $X_r$ for $r\in\{1,\dots,n\}$ has two half-edges (which can be achieved by the use of an equality node).
\Cref{fig:5bitcode_CCchannel} depicts this factor graph for the 5-bit code and $r=1$.
This factor graph represents the channel from $X_r$ to $X_1,\dots,X_n$, i.e., the channel that given some fixed value $X_r=x_r$ generates one of the $2^{k-1}$ possible codewords $\vec{z}\in\mathcal{C}$ with $z_r=x_r$ uniformly randomly.\footnote{Note that here we implicitly assume that there exists some codeword $z\in\mathcal{C}$ where $z_r\neq 0$.}
Here again, the tensor represented by the factor graph differs from the channel by a normalization constant.

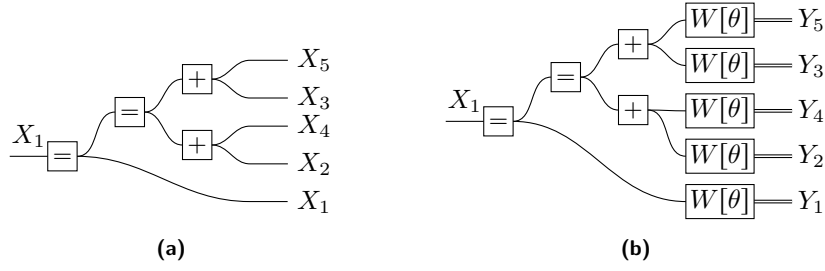
\begin{figure}
\centering
\begin{subfigure}[b]{0.4\textwidth}
\centering
\tikzsetnextfilename{5bitFFGX1}
\begin{tikzpicture}
\pgfdeclarelayer{bg}
\pgfsetlayers{bg,main}

\node[squarenode] (eq1) {\cequal};
\node[squarenode] (eq2) [above right=0.2cm and 0.5cm of eq1] {\cequal};
\node[squarenode] (ch1) [above right=0.05cm and 0.5cm of eq2] {$+$};
\node[squarenode] (ch2) [below right=0.05cm and 0.5cm of eq2] {$+$};
%\node[quadnode] (norm) [above left=0.6cm and 0cm of eq1] {$\frac{1}{4}$};

\draw (eq1.east) to[out=0, in=180] (eq2.west);
\draw (eq2.east) to[out=0, in=180] (ch1.west);
\draw (eq2.east) to[out=0, in=180] (ch2.west);

\draw (ch1.east) to[out=0,in=180] ++(0.5cm,0.25cm) to ++(0.5cm,0) node[right] {$X_5$};
\draw (ch1.east) to[out=0,in=180] ++(0.5cm,-0.25cm) to ++(0.5cm,0) node[right] {$X_3$};
\draw (ch2.east) to[out=0,in=180] ++(0.5cm,0.25cm) to ++(0.5cm,0) node[right] {$X_4$};
\draw (ch2.east) to[out=0,in=180] ++(0.5cm,-0.25cm) to ++(0.5cm,0) node[right] {$X_2$};
\draw (eq1.east) to[out=0,in=180] ++(2.27cm,-0.6cm) to ++(0.5cm,0) node[right] {$X_1$};

\draw (eq1.west) -- node[above] {$X_1$} ++(-0.5cm,0);
\end{tikzpicture}
\caption{}
\label{fig:5bitcode_CCchannel}
\end{subfigure}
\begin{subfigure}[b]{0.4\textwidth}
\centering
\tikzsetnextfilename{5bitFFGX1Channel}
\begin{tikzpicture}
\pgfdeclarelayer{bg}
\pgfsetlayers{bg,main}

\node[squarenode] (eq1) {\cequal};
\node[squarenode] (eq2) [above right=0.2cm and 0.5cm of eq1] {\cequal};
\node[squarenode] (ch1) [above right=0.05cm and 0.5cm of eq2] {$+$};
\node[squarenode] (ch2) [below right=0.05cm and 0.5cm of eq2] {$+$};
%\node[quadnode] (norm) [above left=0.6cm and 0cm of eq1] {$\frac{1}{4}$};
\node[quadnode] (W5) [above right=-0.1cm and 0.5cm of ch1] {$W[\theta]$};
\node[quadnode] (W3) [below=0.15cm of W5] {$W[\theta]$};
\node[quadnode] (W4) [below=0.15cm of W3] {$W[\theta]$};
\node[quadnode] (W2) [below=0.15cm of W4] {$W[\theta]$};
\node[quadnode] (W1) [below=0.15cm of W2] {$W[\theta]$};

\draw (eq1.east) to[out=0, in=180] (eq2.west);
\draw (eq2.east) to[out=0, in=180] (ch1.west);
\draw (eq2.east) to[out=0, in=180] (ch2.west);
\draw (ch1.east) to[out=0, in=180] (W5.west);
\draw (ch1.east) to[out=0, in=180] (W3.west);
\draw (ch2.east) to[out=0, in=180] (W4.west);
\draw (ch2.east) to[out=0, in=180] (W2.west);
\draw (eq1.east) to[out=0, in=180] (W1.west);

\draw (eq1.west) -- node[above] {$X_1$} ++(-0.5cm,0);
% We have to put the double dashed lines in a background layer, otherwise they look ugly
\begin{pgfonlayer}{bg}
  \draw [style=double] (W5.east) to node[right=0.2cm] {$Y_5$} ++(0.5cm,0);
  \draw [style=double] (W3.east) to node[right=0.2cm] {$Y_3$} ++(0.5cm,0);
  \draw [style=double] (W4.east) to node[right=0.2cm] {$Y_4$} ++(0.5cm,0);
  \draw [style=double] (W2.east) to node[right=0.2cm] {$Y_2$} ++(0.5cm,0);
  \draw [style=double] (W1.east) to node[right=0.2cm] {$Y_1$} ++(0.5cm,0);
\end{pgfonlayer}
\end{tikzpicture}
\caption{}
\label{fig:5bitcode_fg}
\end{subfigure}
\caption{(a) Factor graph describing the classical channel from $X_1$ to $X_1,X_2,X_3,X_4,X_5$ for the 5-bit code. (b) Factor graph describing the CQ channel from $X_1$ to the CQ channel outputs $Y_1,Y_2,Y_3,Y_4,Y_5$ for the 5-bit code.}
\end{figure}

Forney-style factor graphs have also been generalized to describe quantum mechanical systems~\cite{loeliger_2017,cao_2017}. Here the central quantities of interest are mixed states (density matrices) and quantum channels, which are again tensors and thus representable with our factor graph formalism. Consider for instance a density matrix $\rho$ on some joint quantum system $Q_1Q_2\dots Q_m$ described by the Hilbert space $\mathcal{H}_1\otimes\dots\otimes\mathcal{H}_m$ where $d_i:=\mathrm{dim}(H_i)$. One possibility to represent $\rho$ with a tensor works as follows: We identify $\rho$ with a rank-$2m$ tensor $T$ defined as
\begin{equation}
T_{i_1^+,i_1^-,\dots,i_m^+,i_m^-} = \left( \bra{i_1^+}_1\otimes\dots\otimes\bra{i_m^+}_m\right) \rho \left(\ket{i_1^-}_1\otimes\dots\otimes\ket{i_m^-}_m\right)\,,
\end{equation}
where for all $j=1,\dots,m$ we have implicitly made some choice of basis $\{\ket{0}_j,\dots,\ket{d_j-1}_j\}$ of $\mathcal{H}_j$.
Notice that the indices of $T$ always appear in pairs $I_j^+,I_j^-$ that run over the same alphabet. Following the notation introduced by~\cite{cao_2017}, we depict both random variables with a single double-edged wire that we simply denote with the symbol $I_j$.
For instance, the state $\rho_Q$ of a quantum system $Q$ can be trivially described by the factor graph
\begin{equation}
\tikzsetnextfilename{rhoQfigure}
\begin{tikzpicture}[baseline]
\pgfdeclarelayer{bg}
\pgfsetlayers{bg,main}
\node[quadnode] (PX) {$\rho_Q$};
\begin{pgfonlayer}{bg}
  \draw[double] (PX.east) to node[above] {$Q$} ++(2cm,0);
\end{pgfonlayer}
\end{tikzpicture}\,\,.
\end{equation}
Since quantum channels are linear maps on the space of endomorphisms of the involved Hilbert spaces, they can also be represented by tensors. Fully analogous to the case of classical channels, an application of a quantum channel on some quantum state can be represented by a contraction of the two corresponding factor graphs. In fact, by correctly identifying half-edges of a factor graph describing some channel, one can verify that the tensor described by the factor graph is the Choi matrix of the involved channel. Therefore the fundamental quantities of quantum information theory (density matrices and Choi matrices) are positive semi-definite operators, in contrast to the non-negative tensors in classical probability theory. The contraction of two positive semi-definite operators can again be shown to be itself a positive semi-definite operator.

In analogy to the example in the classical case presented previously, consider three quantum systems $Q_1,Q_2,Q_3$ in the separable state $\rho_1\otimes\rho_2\otimes\rho_3$. A quantum channel $\mathcal{E}$ takes the systems $Q_1,Q_2$ as input and produces an output described by the quantum system $O$. The joint state of $Q_3O$ is then described by the factor graph
\begin{equation}
\tikzsetnextfilename{rhoQ3figure}
\begin{tikzpicture}[baseline=-20pt]
\pgfdeclarelayer{bg}
\pgfsetlayers{bg,main}
\node[quadnode] (PX1) {$\rho_{1}$};
\node[quadnode] (PX2) [below=0.2cm of PX1] {$\rho_{2}$};
\node[quadnode] (PX3) [below=0.2cm of PX2] {$\rho_{3}$};
\node[quadnode] (channel) [below right=-0.1cm and 0.6cm of PX1] {$\mathcal{E}$};
\begin{pgfonlayer}{bg}
\draw[double] (PX1) to (channel);
\draw[double] (PX2) to (channel);
\draw[double] (channel.east) to node[above] {$O$} +(0.5cm,0);
\draw[double] (PX3) to node[above] {$Q_3$} +(1.75cm,0);
\end{pgfonlayer}
\end{tikzpicture}\,.
\end{equation}

We illustrate one special family of a quantum channels that will be encountered later in this work: uniformly-controlled gates, which apply the unitary $\sum_{i=0,\dots,2^m-1}\ket{i}\bra{i}_C\otimes (U_i)_A$ on the $m$-qubit control system $C$ and the target system $A$, where the $U_i,i=0,\dots,2^m-1$, are some unitaries on the system $A$. The Following factor graph can be shown to describe the corresponding channel
\begin{equation}
\tikzsetnextfilename{ControlUDEFG}
\begin{tikzpicture}[baseline=-10pt]
\pgfdeclarelayer{bg}
\pgfsetlayers{bg,main}
\node[squarenode] (phase) {$\bullet$};
\node[squarenode] (Ui) [below=0.5cm of phase] {$U_i$};
\begin{pgfonlayer}{bg}
\draw [style=double] (phase) to node[above] {$C$} ++(-1.0cm,0);
\draw [style=double] (phase) to node[above] {$C'$} ++(1.0cm,0);
\draw [style=double] (Ui) to node[above] {$A$} ++(-1.0cm,0);
\draw [style=double] (Ui) to node[above] {$A'$} ++(1.0cm,0);
\draw [style=double] (phase) to node[right] {$I$} (Ui);
\end{pgfonlayer}
\end{tikzpicture}\,,
\end{equation}
where $I$ is a $m$-qubit system and the factor node $\bullet$ stands for an equality node for each separate edge, i.e.,
\begin{equation}
\tikzsetnextfilename{DotControlDEFG}
\begin{tikzpicture}[baseline]
\pgfdeclarelayer{bg}
\pgfsetlayers{bg,main}
\node[squarenode] (phase) {$\bullet$};
\begin{pgfonlayer}{bg}
\draw [double] (phase.west) to ++(-0.5cm,0);
\draw [double] (phase.east) to ++(0.5cm,0);
\draw [double] (phase.south) to ++(0,-0.5cm);
\end{pgfonlayer}

\node (equiv) [right=.8cm of phase] {$:=$};

\pgfdeclarelayer{bg}
\pgfsetlayers{bg,main}
\node[squarenode] (eq1) [above right=0.2cm and 1.0cm of equiv] {\cequal};
\node[squarenode] (eq2) [below right=0.1cm and 0.1cm of eq1]{\cequal};
\draw (eq1.west) to ++(-0.75cm,0);
\draw (eq1.east) to ++(1.2cm,0);
\draw (eq1.south) to ++(0,-1.2cm);
\draw (eq2.west) to ++(-1.2cm,0);
\draw (eq2.east) to ++(0.75cm,0);
\draw (eq2.south) to ++(0,-0.75cm);
\end{tikzpicture}\,\,,
\end{equation}
and $U_i$ is the quantum channel corresponding to the isometry $\sum_i\bra{i}_I\otimes (U_i)_A$ acting on the system $IA$. Note that the subscript $i$ in the factor node $U_i$ is chosen to indicate which system (here $I$) acts as the control for choosing which unitary is applied.
Moreoever, for unitaries we slightly abuse the convention of labels of nodes corresponding to channels and simply label the node by the (name of the) unitary, say $U$, rather than the corresponding channel $\mathcal U:\rho\to U\rho U^\dagger$.

Finally, this factor graph formalism can be straightfowardly generalized to also describe CQ and QC channels, following the same procedure to describe classical and quantum objects introduced above. In this case, the corresponding factor nodes will have a mix of single-edged and double-edged wires.
For instance, the factor node
\begin{equation}
\tikzsetnextfilename{FFGCQchannel}
\begin{tikzpicture}
\pgfdeclarelayer{bg}
\pgfsetlayers{bg,main}
\node[squarenode] (W) {\cequal};
\begin{pgfonlayer}{bg}
  \draw [style=double] (W.east) to ++(0.5cm,0);
  \draw (W.west) to ++(-0.5cm,0);
\end{pgfonlayer}
\end{tikzpicture}
\end{equation}
represents the CQ channel which embeds a bit in the computational basis (i.e.~$0\mapsto \ket{0}\bra{0}$ and $1\mapsto \ket{1}\bra{1}$), or when read from the other side as the QC channel corresponding to non-selectively measuring a qubit. Similarly, the CQ channel $W[\theta]$ can be represented by a corresponding factor node with a single-edged and double-edged wire. For example,~\cref{fig:5bitcode_fg} depicts the channel from the codeword bit $X_1$ to the CQ channel outputs $Y_1,\dots,Y_5$ for the 5-bit code, i.e., the channel that given some value $X_1=x_1$ generates a uniformly random codeword $\vec{z}\in\mathcal{C}$ such that $z_1=x_1$ and then passes that codeword $\vec{z}$ through five instances of the channel $W[\theta]$.

To end this section, we introduce a generalization of the channel $W[\theta]$ that will play an important role in later sections. We consider the CQ channel, that takes some bit $x\in\{0,1\}$ as well as a bit string $i\in\{0,1\}^m$ as input for some $m\geq 0$. The channel output is a qubit $Q$ prepared in the state $\qstate{x}{\varphi_i}$ where the angle is chosen from some pre-determined list of angles $\varphi_0,\dots,\varphi_{2^m-1}$ according to the value of $i$. Visually, we will depict this channel with following factor graph node:
\begin{equation}
\tikzsetnextfilename{2InputCQChannelFFG}
\begin{tikzpicture}[baseline]
\pgfdeclarelayer{bg}
\pgfsetlayers{bg,main}
\node[quadnode] (W) {$W[\varphi_i]$};
\begin{pgfonlayer}{bg}
  \draw (W) to node[above] {$X$} ++(-1.5cm,0);
  \draw [double] (W) to node[above] {$Q$} ++(1.5cm,0);
  \draw (W) to node[right] {$I$} ++(0,1.0cm);
\end{pgfonlayer}
\end{tikzpicture}\,\,\,.
\end{equation}
Note that we can recover the original channel $W[\theta]$ by choosing $m=0$. The subscript $i$ in the channel node $W[\varphi_i]$ is used to indicate that the system $I$ controls the chosen angle.

\section{Description of BPQM}\label{sec:description}
Consider a $(n,k)$ binary linear code $\mathcal{C}$ that exhibits a Tanner graph that is a tree.
A codeword $(X_1,\dots,X_n)$ is chosen uniformly randomly from $\mathcal{C}$.
For $i=1,\dots,n$, the $i$th codeword bit $X_i$ is passed through the CQ channel $W[\theta_i]$ for some $\theta_i\in (0,\pi)$ and the output qubit of the channel is denoted by $Y_i$.
BPQM is an algorithm that takes as input the qubits $Y_1,\dots,Y_n$ as well as the channel parameters $\theta_1,\dots,\theta_n$, and its objective is to correctly guess the original codeword $(X_1,\dots,X_n)$.

To decode the full codeword $(X_1,\dots,X_n)$, it suffices to decode $k$ independent codeword bits from which the value of the remaining codeword bits can be uniquely determined. 
By relabelling, we can assume that $X_1,X_2,\dots,X_k$ form such an independent set.
The BPQM decoder of \cite{renes_2017} consists of $k$ sequentional single-bit decoding operations for the codeword bits $X_j,j=1,\dots,k$, each specified by a unitary $V_j$ acting on the qubits $Y_1,\dots,Y_n$ followed by a single-qubit measurement in the $\ket{\pm}$ basis. 
The measurement outcome is the estimate of the corresponding codeword bit.
After the bit $X_j$ has been decoded, the decoding operations are rewound by applying \smash{$V_j^\dagger$}.
For example, the quantum circuit realizing the codeword decoding for the previously introduced 5-bit code is depicted in~\cref{fig:5bitcode_circuit_complete}.

\Cref{sec:single_bit_decoding} is concerned with describing an algorithm that generates a sequence of operations on $Y_1,\dots,Y_n$, which can be represented using a quantum circuit, that realize the single-bit decoding unitary $V_j$.
That is, the single-bit BPQM decoding algorithm consists of two parts, the \emph{quantum circuit} realization of $V_j$ and the \emph{compiler} which generates the quantum circuit description. 
Many steps involved in the single-bit decoding of BPQM might seem somewhat arbitrary at first glance.
In section~\cref{sec:contraction}, the action of the single bit decoding procedure will be analyzed using the factor graph formalism from~\cref{sec:fg_states_channels}, and by doing so it becomes naturally clear why optimal bit decoding is realized.

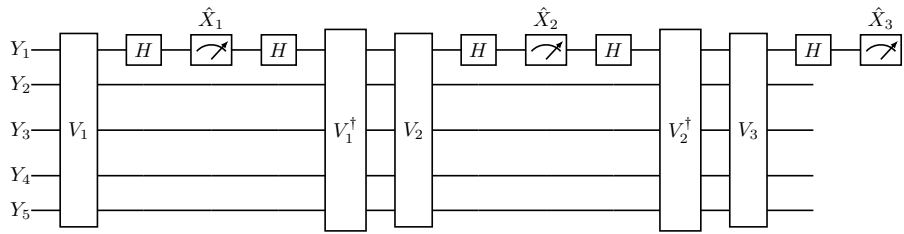
\begin{figure}[h]
  \centering
  \begin{adjustbox}{width=0.8\textwidth}
    \tikzsetnextfilename{CodewordDecoderCircuit}
  \begin{quantikz}
    Y_1 & \gate[wires=5]{V_1} & \gate{H} & \meter{$\hat{X}_1$} & \gate{H} & \gate[wires=5]{V_1^{\dagger}} & \gate[wires=5]{V_2} & \gate{H} & \meter{$\hat{X}_2$} & \gate{H} & \gate[wires=5]{V_2^{\dagger}} & \gate[wires=5]{V_3} & \gate{H} & \meter{$\hat{X}_3$} \\ [-0.4 cm]
    Y_2 &                            & \qw      & \qw                 & \qw      &                                      &                            & \qw      & \qw      & \qw      &                                      &                                       & \qw      &                     \\
    Y_3 &                            & \qw      & \qw                 & \qw      &                                      &                            & \qw      & \qw      & \qw      &                                      &                                       & \qw      &                     \\
    Y_4 &                            & \qw      & \qw                 & \qw      &                                      &                            & \qw      & \qw      & \qw      &                                      &                                       & \qw      &                     \\ [-0.4 cm]
    Y_5 &                            & \qw      & \qw                 & \qw      &                                      &                            & \qw      & \qw      & \qw      &                                      &                                       & \qw      &                     \\
  \end{quantikz}
  \end{adjustbox}
  \caption{Quantum circuit for decoding the complete codeword for the 5-bit code. The input of the circuit are the 5 channel output qubits $Y_1,\dots,Y_5$. The classical output of this circuit are the estimates $\hat{X}_1,\hat{X}_2,\hat{X}_3$ of the bits $X_1,X_2,X_3$.}
  \label{fig:5bitcode_circuit_complete}
\end{figure}

On a simple 5 bit code, it has been observed in~\cite{rengaswamy_2020} that the quantum circuits for decoding subsequent bits can be simplified by constructing them from an updated factor graph where the information of previously decoded codeword bits is incorporated. We do not address this optimization here and leave its analysis as an open question for further research.

\subsection{Single-bit decoding}\label{sec:single_bit_decoding}
\subsubsection{Message passing graph}
\label{sec:mpgdef}
The BPQM algorithm to decode a single bit $X_r$ for some $r\in\{1,\dots,k\}$ is most easily expressed in terms of what we call a \emph{message-passing graph} (MPG) for the code $\mathcal C$ w.r.t.\ $X_r$. 
It is a binary tree associated to a codeword bit $X_r$ which simplifies the bookkeeping of all the conditional operations required in the quantum circuit implementing the decoder.
It can be constructed from the Tanner graph as follows.
First, determine the Forney-style description of the Tanner graph of $\mathcal C$, adding equality nodes as necessary so that to each codeword bit is associated with a half-edge and there is no other half-edge in the graph. 
The resulting graph consists of check and equality nodes of varying degrees.
Using the decomposition rules depicted in~\cref{fig:degree_reduction}, one adds additional nodes and edges to replace any node of degree larger than three by nodes of degree three.
Observe that doing so does not alter the tree structure of the graph. 
Any node of degree two can be removed, as the two involved random variables are essentially equal.
The resulting factor graph therefore only has nodes of degree three.
Connect additional ``channel'' nodes to the half-edges, labelled by channel parameter $\theta_j$ for the codeword bit $X_j$. 
Place an equality node on the edge corresponding to $X_r$, and label this node as the root of the graph. 
Give a direction to each edge by pointing to its incident node which is closer to the root.  
For each node, choose an ordering of its two directly preceding nodes as first and second. 
The ordering is indicated on the graph by adding a dot to the incident edge from the node chosen to be first.
Either choice is valid, but the particular choice will have implications for the quantum circuit. 
However, the decoding performance of BPQM is unaffected by this choice.
Finally, name the check and variable nodes arbitrarily with integers $1,\dots,n-1$.

\begin{figure}
\centering
  \begin{subfigure}[b]{0.45\textwidth}
  \centering
  \tikzsetnextfilename{EqualityFFGDecomposition}
  \begin{tikzpicture}
    \node[squarenode] (eqlhs) {\cequal};
    \node[squarenode] (eq1rhs) [right=1.5cm of eqlhs] {\cequal};
    \node[squarenode] (eq2rhs) [right=2.5cm of eq1rhs] {\cequal};
    \node (dotslhs) [below=0.5cm of eqlhs] {$\dots$};
    \node (dotsrhs) [below=0.5cm of eq1rhs] {$\dots$};
    \node (equiv) [right=0.5cm of eqlhs]{$\cong$};
   
    \draw (eq1rhs) to (eq2rhs);
    \draw (eqlhs) to node[left] {$Z_1$} ++(-0.5cm, -1.0cm);
    \draw (eqlhs) to node[right] {$Z_d$} ++(0.5cm, -1.0cm);
    \draw (eq1rhs) to node[left] {$Z_1$} ++(-0.5cm, -1.0cm);
    \draw (eq1rhs) to node[right] {$Z_{d-2}$} ++(0.5cm, -1.0cm);
    \draw (eq2rhs) to node[left] {$Z_{d-1}$} ++(-0.3cm, -1.0cm);
    \draw (eq2rhs) to node[right] {$Z_{d}$} ++(0.3cm, -1.0cm);
  \end{tikzpicture}
  \caption{}
  \end{subfigure}
  \hspace{6mm}
  \begin{subfigure}[b]{0.45\textwidth}
  \centering
    \tikzsetnextfilename{CheckFFGDecomposition}
  \begin{tikzpicture}
    \node[squarenode] (chlhs) {$+$};
    \node[squarenode] (ch1rhs) [right=1.5cm of chlhs] {$+$};
    \node[squarenode] (ch2rhs) [right=2.5cm of ch1rhs] {$+$};
    \node (dotslhs) [below=0.5cm of chlhs] {$\dots$};
    \node (dotsrhs) [below=0.5cm of ch1rhs] {$\dots$};
    \node (chuiv) [right=0.5cm of chlhs]{$\cong$};
   
    \draw (ch1rhs) to (ch2rhs);
    \draw (chlhs) to node[left] {$Z_1$} ++(-0.5cm, -1.0cm);
    \draw (chlhs) to node[right] {$Z_d$} ++(0.5cm, -1.0cm);
    \draw (ch1rhs) to node[left] {$Z_1$} ++(-0.5cm, -1.0cm);
    \draw (ch1rhs) to node[right] {$Z_{d-2}$} ++(0.5cm, -1.0cm);
    \draw (ch2rhs) to node[left] {$Z_{d-1}$} ++(-0.3cm, -1.0cm);
    \draw (ch2rhs) to node[right] {$Z_{d}$} ++(0.3cm, -1.0cm);
  \end{tikzpicture}
  \caption{}
  \end{subfigure}
  \caption{Decomposition of degree $d\geq 4$ equality and check nodes into multiple nodes of lower degree. The equality $\cong$ can be understood in terms of the exterior functions of the two factor graphs. By iteratively applying these two rules, one can obtain a graph where every equality and check has exactly degree three.}
  \label{fig:degree_reduction}
\end{figure}
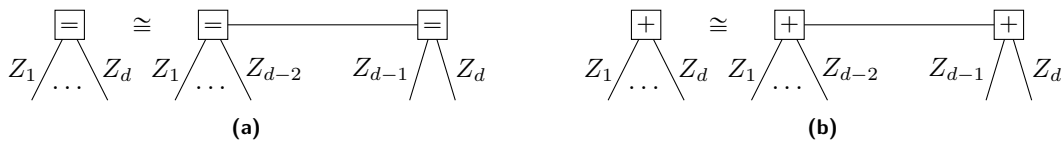

The resulting graph is defined as an MPG of $\mathcal{C}$ w.r.t.\ $X_r$.
Note that an MPG is not unique, as the original Tanner graph from which it is derived is generally also not unique.
As an example,~\cref{fig:5bitMPG} depicts an MPG of the 5-bit code w.r.t.~$X_1$.
An MPG is always a binary tree with $n$ leaves (channel nodes) which is full, meaning every node has either two or zero predecessors.\footnote{As opposed to descendants, since the edges are directed toward the root.}
The number of check and equality nodes must therefore be $n-1$, because a full binary tree with $n$ leaves always has $2n-1$ nodes.\footnote{This can be seen by induction, since adding a degree three node increases the number of leaves by one and the root equality node has degree two.}

Out of these $n-1$ nodes, precisely $k-1$ are check nodes and $n-k$ are equality nodes.
One way to show this fact is with a degree of freedom counting argument, as follows.
For this purpose, let us again interpret the MPG without channel nodes as a Forney-style factor graph.
It contains precisely $2n-2$ edges, meaning $2n-2$ binary variables are involved.
But the number of degrees of freedom is smaller than $2n-2$, as the check and equality nodes impose linear constraints on these binary variables.
More precisely, every check node imposes a linear constraint on the three involved variables, whereas every equality node imposes two linear constraints on the three involved variables.
If $n_c$ denotes the number of check nodes, the number of equality nodes is thus $n-1-n_c$.
Observe that the root equality node only has degree $2$ and therefore only imposes a single linear constraint.
All these constraints can easily be seen to be linearly independent,\footnote{Put differently: The matrix representing this homogeneous system of linear equations has the rank $n_c + 2(n-1-n_c-1) + 1$.} since the factor graph in question is a tree.
The final number of degrees of freedom is thus
\begin{equation}
  n_{\text{dof}} = (2n-2) - n_c - 2(n-1-n_c-1) - 1 = n_c+1\, .
\end{equation}
Since the MPG under consideration here was derived from a factor graph encapsulating the constraints of a linear code of dimension $k$, there must be $k$ independent random variables corresponding to the encoded bits, so $n_{\text{dof}}=k$. 
This directly implies that $n_c=k-1$.

\begin{table}[h]
  \centering
  \begin{tabular}{||c|c|c||} 
  \hline
  Node & Check node list & Branch list \\ [0.5ex] 
  \hline\hline
  $\theta_5$ & $\varnothing$ & $(\varnothing,\theta_5,1)$ \\ 
  \hline
  $\theta_3$ & $\varnothing$ & $(\varnothing,\theta_3,1)$ \\ 
  \hline
  $\theta_4$ & $\varnothing$ & $(\varnothing,\theta_4,1)$ \\ 
  \hline
  $\theta_2$ & $\varnothing$ & $(\varnothing,\theta_2,1)$ \\ 
  \hline
  $\theta_1$ & $\varnothing$ & $(\varnothing,\theta_1,1)$ \\ 
  \hline
  $3$ & $(3)$ & $(i, \theta_5\boxstar_i\theta_3, p_{\boxstar}(\theta_5,\theta_3,i))_{i=0,1}$ \\
  \hline
  $2$ & $(2)$ & $(j, \theta_4\boxstar_j\theta_2, p_{\boxstar}(\theta_4,\theta_2,j))_{j=0,1}$ \\
  \hline
  $4$ & $(3,2)$ & $(ij, (\theta_5\boxstar_i\theta_3)\ostar (\theta_4\boxstar_j\theta_2), p_{\boxstar}(\theta_5,\theta_3,i)p_{\boxstar}(\theta_4,\theta_2,j))_{i,j=0,1}$ \\
  \hline
  $1$ & $(3,2)$ & $(ij, (\theta_5\boxstar_i\theta_3)\ostar (\theta_4\boxstar_j\theta_2)\ostar \theta_1, p_{\boxstar}(\theta_5,\theta_3,i)p_{\boxstar}(\theta_4,\theta_2,j))_{i,j=0,1}$ \\
  \hline
  \end{tabular}
  \caption{Node lists computed by BPQM for the MPG depicted in~\cref{fig:5bitMPG}. The check and equality nodes are referred to by their name (integer ranging from $1$ to $n$) which they are given in the MPG. To shorten the expressions, we use the notation $p_{\boxstar}(\varphi_1,\varphi_2,l):=(1+(-1)^l\cos\varphi_1\cos\varphi_2)/2$ for $\varphi_1,\varphi_2\in (0,\pi)$ and $l\in\{0,1\}$. The branch list of the root node $1$ contains four elements.}
  \label{tab:5bitcode_mpg_labels}
\end{table}

\subsubsection{Check and branch lists}
The MPG will be a central ingredient to define the quantum operations that constitute the unitary $V_r$.
But first, additional classical information needs to be generated from it.
This is the compilation step, which keeps track of conditional operations that need to be performed in the circuit.  
More precisely, the BPQM quantum circuit compiler generates two lists for each node of the MPG.
The first is a length-$m$ list of distinct integers in $\{1,\dots,n-1\}$, where $m\in\{0,\dots,k-1\}$ is the number of check nodes in the set containing the node in question as well as all its precedessors in the MPG.
We refer to this first list as the check node list.
The second is a length-$2^m$ list of triples $(s,\theta,p)$, where $s\in \mathbb Z_2^m$, $\theta\in(0,\pi)$, and  $p\in [0,1]$. %,  $h\in \mathbb Z_2^m$. 
We refer to this as the branch list. 
The final entry in the triple, a probability, will not play a role here in specifying the decoding quantum circuit, but will be useful in analyzing block optimality later.

At the leaves, $m=0$ and the check node list is empty, denoted by $\varnothing$, and the branch list contains the single entry $(\varnothing,\theta_j,1)$ at the $j$th leaf node.
The check node and branch lists of the remaining nodes are computed iteratively, as for any node the two lists are defined from those of their direct predecessors.
Let $L$ and $L'$ be the check node lists of the two predecessors, of lengths $m$ and $m'$, respectively.  
Furthermore, let $(s_i, \varphi_i,p_i)_i$ and $(s'_j, \varphi'_j,p'_j)_j$ be the respective branch lists.  
At equality nodes, the check node list is simply the length-$(m+m')$ list $\smash{L\concat L'}$, where `$\smash{\concat}$' denotes concatenation. 
The branch list contains the $\smash{2^{m+m'}}$ triples $\smash{(s_i\concat s'_j\,,\varphi_i\ostar \varphi'_j\,,p_ip'_j)}$ for all $i$ and $j$, where 
\begin{equation}
\label{eq:ostarangle}
\varphi\ostar \varphi':=\arccos \left(\cos \varphi \cdot \cos \varphi'\right)\,.
\end{equation}
Meanwhile for check nodes the lists have lengths $m+m'+1$ and $2^{m+m'+1}$, respectively. 
The check node list is simply ${t}\concat L\concat L'$ where $t$ is the name of the check node in question.
To construct the branch list, first define, for $l\in \mathbb Z_2$,
\begin{equation}
\label{eq:boxstarangle}
\varphi\boxstar_l \varphi'=\arccos\left(\frac{\cos\varphi+(-1)^l \cos\varphi'}{1+(-1)^l \cos\varphi\cos\varphi'}\right)\,,
\end{equation}
\begin{equation}
p_{\boxstar}(\varphi,\varphi',l):=\tfrac12(1+(-1)^l\cos\varphi\cos\varphi')\,.
\end{equation}
Then the branch list contains the $2^{m+m'+1}$ triples
\begin{equation}
\left(l\concat s_i\concat s'_j,\,\,\varphi_i\boxstar_l \varphi'_j,\,\,p_ip'_j\cdot p_{\boxstar}(\varphi_i,\varphi'_j,l)\right)\,,
\end{equation}
where $l$ ranges over $\{0,1\}$ and again $i$ and $j$ range over all their respective values.
As an example, \cref{tab:5bitcode_mpg_labels} gives the computed lists for the MPG of the 5-bit code w.r.t~$X_1$ depicted in~\cref{fig:5bitMPG}.

Observe that the functions $\ostar$ and $\boxstar_0$ are symmetric in their arguments, but $\boxstar_1$ is not. 
In fact, $\beta\boxstar_1\alpha=\pi-\alpha\boxstar_1\beta$. 
Therefore, the choice of which of the two predecessor nodes corresponds to the first argument results in different angle entries in the branch list. 

Note that the first entries $s$ of the triples range over all elements in $\mathbb Z_2^m$. 
Thus, it is not strictly necessary to include $s$ in the triple, as it could be specified by order of the triple (now pair) in the branch list. 
However, we have opted for this redundancy to avoid having to explicitly deal with the ordering in the description of the compiliation step.

\subsubsection{Quantum circuit}
Using the MPG and the generated check and branch lists, we can now easily construct a quantum circuit to realize $V_r$. 
The circuit can be thought of as the process of passing qubits along the edges of the MPG.
First, every leaf node $\theta_j$ passes the channel output qubit $Y_j$ to its immediate successor.
Once a check or equality node has received a qubit from both of its immediate predecessors, it performs some unitary gate on these two qubits, resulting in two output qubits.
One of these qubits is then passed to the immediate successor of the node, whereas the other remains at the node.
We call the former \emph{data qubit} and the later \emph{ancilla qubit}.
The operation that a node performs on its two input qubits depends on whether the node is an equality or a check node.
We call the node operation \emph{equality node operation} or \emph{check node operation} correspondigly.
The final data qubit produced by the root of the graph is precisely the qubit that is to be measured in the $\ket{\pm}$ basis in order to obtain the estimate $\hat{X}_r$ of $X_r$.
To fully specify the BPQM algorithm and the quantum circuit implementation of $V_r$, it thus suffices to specify the quantum gates realizing the check and equality node operations. 

The check node operation is simply a \cnot{} gate.
More specifically, choose the control qubit to be the input data qubit coming from the first predecessor in the MPG, as indicated by the dot in the MPG, and the second data qubit to the be target.  
The control qubit is passed on further as a data qubit, while the target is an ancilla qubit and remains at the check node.  

The equality node operation is more involved. 
It consists of a two-qubit unitary applied to the incoming data qubits, uniformly controlled by the ancilla qubits produced by the predecessor check nodes in the MPG. 
The unitary in question is given by\footnote{Note that this is a slight variation of the unitary appearing in \cite{renes_2017,rengaswamy_2020}; the second row is negated. Additionally, $b_+$ and $b_-$ are interchanged, for reasons which will become apparent in \cref{sec:quantummessagepassing}.}
\begin{align}
\label{eq:equalitynodeU}
    U_{\ostar}(\alpha,\beta) &:= \begin{pmatrix}a_+&0&0&a_- \\ -a_-&0&0&a_+ \\ 0&b_-&b_+&0 \\ 0&b_+&-b_-&0\end{pmatrix}\,,\\
    \label{eq:apmdef}
    a_{\pm}&:=\frac{1}{2}\frac{\cos(\frac{\alpha-\beta}{2})\pm \cos(\frac{\alpha+\beta}{2})}{\abs{\cos(\frac{\alpha\ostar\beta}{2})}}, \, \\
    \label{eq:bpmdef}
    b_{\pm}&:=\frac{1}{2}\frac{\sin(\frac{\alpha+\beta}{2})\pm \sin(\frac{\alpha-\beta}{2})}{\abs{\sin(\frac{\alpha\ostar\beta}{2})}}\,,
  \end{align}
for parameters $\alpha,\beta\in (0,\pi)$.
Again, the first qubit is chosen to be the qubit arriving from the first predecessor in the MPG. 
At the output of the gate, the first qubit sent to the next node as the data qubit while the second is the ancilla which is retained.
The check and branch lists of the MPG specify the precise conditional $U_{\ostar}$ operation that is required. 
For a given equality node, denote its check list by $L$ and its branch list by $(s_i,\varphi_i,p_i)_i$.
Then the equality node operation is the uniformly-controlled unitary  
\begin{equation}
\label{equalitynodeunitary}
\sum_{i,\ell}U_{\ostar}(\varphi_i,\varphi'_\ell)\otimes \ketbra{s_i}\otimes \ketbra{s'_\ell}\,,
\end{equation}
where $U_{\ostar}$ acts on the data qubits and the projectors act on ancilla qubits. 
Precisely which are the ancilla qubits is recorded in the check node lists, as one ancilla qubit has been generated at each check node passed thus far in the algorithm.

\begin{figure}
  \centering
  \tikzsetnextfilename{5bitBPQMCircuit}
  \resizebox{.9\textwidth}{!}{
  \begin{quantikz}
    \lstick{$Y_5$} & \ctrl{1} & \qw      & \gate[wires=2]{U_{\ostar}(\theta_5\boxstar_i\theta_3,\theta_4\boxstar_j\theta_2)} & \qw      & \gate[wires=2]{U_{\ostar}((\theta_5\boxstar_i\theta_3)\ostar(\theta_4\boxstar_j\theta_2), \theta_1)} & \qw & \gate{H} & \meter{0/1} \\ [-0.4cm]
    \lstick{$Y_3$} & \targ{}  & \swap{1} &                                                                           & \swap{3} &                                                                                         & \qw & \qw     & \qw \\ [-0.4cm]
    \lstick{$Y_4$} & \ctrl{1} & \targX{} & \uctrl{i}\vqw{-1}                                                         & \qw      & \uctrl{i}\vqw{-1}                                                                       & \qw & \qw     & \qw \\ [-0.3cm]
    \lstick{$Y_2$} & \targ{}  & \qw      & \uctrl{j}\vqw{-1}                                                         & \qw      & \uctrl{j}\vqw{-1}                                                                       & \qw & \qw     & \qw \\ [-0.1cm]
    \lstick{$Y_1$} & \qw      & \qw      & \qw                                                                       & \targX{} & \qw                                                                                     & \qw & \qw     & \qw \\
  \end{quantikz}
  }
  \tikzexternaldisable
  \caption{Quantum circuit implementing BPQM decoding of $X_1$ for the 5-bit code. Without the Hadamard gate and measurement at the end, this circuit realizes the unitary $V_1$. The symbol {\protect\tikz\protect \node [ucontrolled] (c) {}; } denotes the qubits on which a gate is uniformly controlled.}
  \tikzexternalenable
  \label{fig:5bitcode_circuit}
\end{figure}

\bigskip
The above algorithm is not truly a message-passing algorithm, since the equality node operation does not act on the two received data qubits alone. 
Furthermore, the uniformly-controlled $U_{\ostar}$ gates can generally only be decomposed into a number of two-qubit gates scaling with the number of different control patterns, i.e.\ $2^{k-1}$~\cite{bergholm_2005}. 
Hence, for codes of finite rate such that $k$ scales linearly with $n$, the resulting quantum circuit depth will be exponential in the blocklength of the code. 
In fact the exponential overhead already shows up before executing the circuit, in the compilation, which requires an exponential amount of (classical) memory to store the branch list.

This significant drawback was overlooked in previous literature, mainly because the exponential runtime of the circuit can be avoided in the case where one only decodes a single codeword bit. 
In this case the ancilla qubits can simply be measured as soon as they are produced, and the appropriate angle arguments for equality node gates $U_{\ostar}$ can be generated for the specific observed cases. 
However, this hybrid classical-quantum approach is not useful when decoding the complete codeword, as it prevents the node operations from being rewound (which can normally be done by executing $V_r^\dagger$).

\subsection{Bit optimality of BPQM}\label{sec:contraction}
To establish that the single-bit decoding operation for $X_r$ realizes the Helstrom measurement, and thus constitutes the optimal bit decoding, consider the CQ channel from some codeword bit $X_r$ to the output quantum systems $Y_1,\dots,Y_n$  obtained by passing a randomly-chosen codeword with the prescribed value of $X_r$ through the channel $W[\theta_1]\otimes W[\theta_2]\otimes \cdots\otimes W[\theta_n]$. 
Optimality is the statement that the decoding quantum circuit realizes the Helstrom measurement for the two output states of the channel. 

A convenient means to show this is provided by the Forney-style factor graph representation of the CQ channel and the decoding operations. 
Such a factor graph of the channel can be constructed by taking the binary tree of the MPG and adding a half-edge to the root equality node and double edges to the channel nodes at the leaves. 
For example, a factor graph representing this CQ channel for the 5-bit code and $r=1$ is depicted in~\cref{fig:5bitcode_fg}. 
To this factor graph we may append the factor graph representation of the gates involved in the decoding operation. 
The resulting graph can be successively simplified by making use of two contraction identities, each of which reduces the number of channel nodes by one. 
The final simplified graph will be a CQ channel with a single-qubit output, for which the Helstrom measurement is a simple projection operation onto the eigenbasis of $\sigma_x$.

\paragraph{Equality node contraction}
Let $\alpha,\beta\in(0,\pi)$.
Consider the CQ channel described by the factor graph
  \begin{equation}
  \label{eq:equalitynodeFFG}
  \tikzsetnextfilename{EqualityNodeFFG}
  \begin{tikzpicture}[baseline]
    \pgfdeclarelayer{bg}
    \pgfsetlayers{bg,main}
    \node[squarenode] (eq) {\cequal};
    \node[quadnode] (W1) [above right=0.0cm and 0.4cm of eq] {$W[\alpha]$};
    \node[quadnode] (W2) [below right=0.0cm and 0.4cm of eq] {$W[\beta]$};
    \begin{pgfonlayer}{bg}
      \draw (eq.west) to node[above] {$X$} ++(-0.5cm,0);
      \draw (eq.north) |- (W1.west);
      \draw (eq.south) |- (W2.west);
      \draw[style=double] (W1.east) to ++(+0.3cm,0);
      \draw[style=double] (W2.east) to ++(+0.3cm,0);
    \end{pgfonlayer}
  \end{tikzpicture}
  \end{equation}
that produces two-qubit pure quantum states of the form 
  \begin{equation}
  \label{eq:equalityconvolutionoutput}
    \left[ W[\alpha]\ostar W[\beta] \right](x) :=\qstate{x}{\alpha}\qstateconj{x}{\alpha} \otimes \qstate{x}{\beta}\qstateconj{x}{\beta}
  \end{equation}
 given the input bit $x\in\{0,1\}$.

Applying the unitary $U_\ostar(\alpha,\beta)$ from \eqref{eq:equalitynodeU} as per BPQM yields 
\begin{align}\label{eq:Ustar_identity}
   U_{\ostar}(\alpha,\beta)\cdot \qstate{x}{\alpha}\otimes\qstate{x}{\beta} = \qstate{x}{\alpha\ostar\beta}\otimes\ket{0}\,,
\end{align}
using the $\ostar$ angle function from \eqref{eq:ostarangle}. 
Importantly, the first output qubit, the data qubit, is the output of a pure state channel $W$ for appropriate choice of angle, namely $\alpha\ostar\beta$. 
That is, the information about the input $x$ has been compressed into a single qubit, in accordance with the fact that the two possible pure output states in \eqref{eq:equalityconvolutionoutput} only span a two-dimensional space.
This is made especially clear in the factor graph contraction identity depicted in \cref{fig:contraction_identity_equality}. 
  By linearity, the contraction identity can be extended to the case that the channels $W[\alpha]$ and $W[\beta]$ depend on additional classical information which determines their angle parameters. 
  The extended contraction identity is depicted in \cref{fig:contraction_identity_equality_generalized}. 
  \begin{figure}[h]
   \begin{subfigure}[b]{0.49\textwidth}
   \tikzsetnextfilename{contraction_identity_equality}
     \begin{tikzpicture}
    \pgfdeclarelayer{bg}
    \pgfsetlayers{bg,main}

    \node[squarenode] (eq) {\cequal};
    \node[quadnode] (W1) [above right=0.2cm and 0.2cm of eq] {$W[\alpha]$};
    \node[quadnode] (W2) [below right=0.2cm and 0.2cm of eq] {$W[\beta]$};
    \node[quadnode, minimum height=2cm] (U) [right=1.5cm of eq] {\scalebox{0.7}{$U_{\ostar}(\alpha,\beta)$}};

    \draw (eq.west) -- node[above] {$X$} ++(-0.5cm,0);
    \draw (eq.north) to[out=90, in=180] (W1.west);
    \draw (eq.south) to[out=-90, in=180] (W2.west);
    % We have to put the double dashed lines in a background layer, otherwise they look ugly
    \begin{pgfonlayer}{bg}
      \draw [style=double] (W1.east) to ++(1.9cm,0) node[above] {$Q_1$};
      \draw [style=double] (W2.east) to ++(1.9cm,0) node[below] {$Q_2$};
    \end{pgfonlayer}

    % equivalence sign
    \node (equiv) [right=0.4cm of U] {$\cong$};
    
    % right-hand side:
    \node[quadnode] (Wrhs) [above right=0.15cm and 0.5cm of equiv] {$W[\alpha\ostar\beta]$};
    \node[quadnode] (zerorhs) [below=0.8cm of Wrhs] {$\ket{0}\bra{0}$};
    \draw (Wrhs.west) -- node[above] {$X$} ++(-0.5cm,0);
    \begin{pgfonlayer}{bg}
      \draw [style=double] (Wrhs) to node[above] {$Q_1$} ++(1.4cm,0);
      \draw [style=double] (zerorhs) to node[above] {$Q_2$} ++(1.4cm,0);
    \end{pgfonlayer}
  \end{tikzpicture}
\caption{}
  \label{fig:contraction_identity_equality}
  \end{subfigure}
  \begin{subfigure}[b]{0.49\textwidth}
  \tikzsetnextfilename{contraction_identity_equality_generalized}
\begin{tikzpicture}
    \pgfdeclarelayer{bg}
    \pgfsetlayers{bg,main}

    \node[squarenode] (eq) {\cequal};
    \node[quadnode] (W1) [above right=0.2cm and 0.2cm of eq] {$W[\alpha_i]$};
    \node[quadnode] (W2) [below right=0.2cm and 0.2cm of eq] {$W[\beta_j]$};
    \node[quadnode, minimum height=2cm] (U) [right=1.5cm of eq] {\scalebox{0.7}{$U_{\ostar}(\alpha_{\tilde{i}},\beta_{\tilde{j}})$}};
    \node[squarenode] (eqi) [above=0.5cm of W1] {\cequal};
    \node[squarenode] (eqj) [below=0.5cm of W2] {\cequal};
    \node[squarenode] (phasei) [right=1.0cm of eqi] {$\bullet$};
    \node[squarenode] (phasej) [right=1.0cm of eqj] {$\bullet$};

    \draw (eq.west) -- node[above] {$X$} ++(-0.5cm,0);
    \draw (eq.north) to[out=90, in=180] (W1.west);
    \draw (eq.south) to[out=-90, in=180] (W2.west);
    % We have to put the double dashed lines in a background layer, otherwise they look ugly
    \begin{pgfonlayer}{bg}
      \draw [style=double] (W1.east) to ++(1.9cm,0) node[above] {$Q_1$};
      \draw [style=double] (W2.east) to ++(1.9cm,0) node[below] {$Q_2$};
      \draw (W1.north) to (eqi.south);
      \draw (W2.south) to (eqj.north);
      \draw (eqi.west) to node[above] {$I$} ++(-0.5cm,0);
      \draw (eqj.west) to node[above] {$J$} ++(-0.5cm,0);
      \draw [double] (phasei.south) to ++(0,-1.0cm);
      \draw [double] (phasej.north) to ++(0,1.0cm);
      \draw [double] (eqi.east) to ++(2.0cm,0) node[above] {$\tilde{I}$};
      \draw [double] (eqj.east) to ++(2.0cm,0) node[below] {$\tilde{J}$};
    \end{pgfonlayer}

    % equivalence sign
    \node (equiv) [right=0.4cm of U] {$\cong$};
    
    % right-hand side:
    \node[quadnode] (Wrhs) [below right=0.0cm and 0.5cm of equiv] {$W[\alpha_i\ostar\beta_j]$};
    \node[quadnode] (zerorhs) [below=0.3cm of Wrhs] {$\ket{0}\bra{0}$};
    \node[squarenode] (eqirhs) [above left=1.0cm and 0.2cm of Wrhs.north] {\cequal};
    \node[squarenode] (eqjrhs) [above right=0.4cm and 0.2cm of Wrhs.north] {\cequal};
    \draw (Wrhs.west) -- node[above] {$X$} ++(-0.5cm,0);
    \begin{pgfonlayer}{bg}
      \draw (eqirhs) to ++(-1.0cm,0) node[above] {$I$};
      \draw (eqjrhs) to ++(-1.8cm,0) node[above] {$J$};
      \draw (eqirhs) to ++(0,-1.2cm);
      \draw (eqjrhs) to ++(0,-0.7cm);
      \draw [style=double](eqirhs) to ++(1.8cm,0) node[above] {$\tilde{I}$};
      \draw [style=double](eqjrhs) to ++(1.0cm,0) node[above] {$\tilde{J}$};
      \draw [style=double] (Wrhs) to node[above] {$Q_1$} ++(1.4cm,0);
      \draw [style=double] (zerorhs) to node[above right=0cm and -0.1cm] {$Q_2$} ++(1.4cm,0);
    \end{pgfonlayer}
  \end{tikzpicture}
  \caption{}
  \label{fig:contraction_identity_equality_generalized}
  \end{subfigure}
  \caption{(a) Graphical depiction of equality node contraction identity using Forney-style factor graphs. The equivalence $\cong$ is to be understood as equality of the exterior functions of the two factor graphs. (b) represents a generalizations of this identity, where the channel parameters are classically conditioned on some $n_I$-bit and $n_J$-bit systems $I$ and $J$ for some $n_I,n_J\geq 0$. $\{\alpha_i\}_{i=0,\dots,2^{n_I}-1}$ and $\{\beta_j\}_{j=0,\dots,2^{n_J}-1}$ are some sets of angles between $0$ and $\pi$.}
  \label{fig:contraction_identities_equality}
  \end{figure}
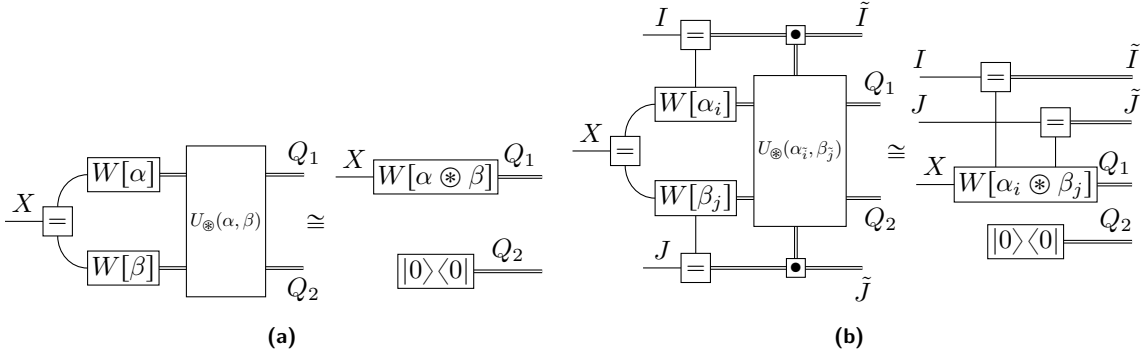

\paragraph{Check node contraction}
Let $\alpha,\beta\in(0,\pi)$.
Consider the CQ channel described by the factor graph
  \begin{equation}
  \label{eq:checknodeffg}
  \tikzsetnextfilename{CheckNodeFFG}
  \begin{tikzpicture}[baseline]
    \pgfdeclarelayer{bg}
    \pgfsetlayers{bg,main}
    \node[squarenode] (check) {$+$};
    \node[quadnode] (W1) [above right=0.0cm and 0.4cm of check] {$W[\alpha]$};
    \node[quadnode] (W2) [below right=0.0cm and 0.4cm of check] {$W[\beta]$};
    \begin{pgfonlayer}{bg}
      \draw (check.west) to node[above] {$X$} ++(-0.5cm,0);
      \draw (check.north) |- (W1.west);
      \draw (check.south) |- (W2.west);
      \draw[style=double] (W1.east) to ++(+0.3cm,0);
      \draw[style=double] (W2.east) to ++(+0.3cm,0);
    \end{pgfonlayer}
  \end{tikzpicture}
  \end{equation}
which, given the input $x\in\{0,1\}$, produces mixed quantum states of the form
  \begin{align}
    \left[ W[\alpha]\boxstar W[\beta] \right](x) :=\, &
    \frac{1}{2}\qstate{x}{\alpha}\qstateconj{x}{\alpha} \otimes \qstate{0}{\beta}\qstateconj{0}{\beta} \nonumber\\&+ 
    \frac{1}{2}\qstate{1-x}{\alpha}\qstateconj{1-x}{\alpha} \otimes \qstate{1}{\beta}\qstateconj{1}{\beta}\,.
  \end{align}
Applying a \cnot{} gate on the two involved qubits results in the state 
  \begin{equation}
  \label{eq:checknodeactionmixed}
    \cnot{}\cdot [W[\alpha]\boxstar W[\beta]](x) \cdot \cnot{}^{\dagger} = \sum\limits_{l=0,1}{p_l\qstate{x}{\alpha\boxstar_l\beta}\qstateconj{x}{\alpha\boxstar_l\beta}\otimes \ket{l}\bra{l}}\,,
  \end{equation}
  where $p_{0/1} := \frac{1}{2}(1\pm \cos\alpha\cos\beta)$ and using $\alpha\boxstar_{l}\beta$ from \eqref{eq:boxstarangle}. 
Again the data qubit of the output is the output of a pure state channel. 
Now, however, there are two possible values for the angle parameter of the channel, and the information about the particular case is stored (classically) in the ancilla qubit. 
In this case the information about the input is not compressed to a single qubit, but rather to a CQ state of two qubits. 
The corresponding factor graph contraction identity is depicted in \cref{fig:contraction_identity_check}. 
It can also be extended by linearity to the case that the parameters of the channels $W[\alpha]$ and $W[\beta]$ are determined by further random variables, as depicted in \cref{fig:contraction_identity_check_generalized}.
\begin{figure}
  \begin{subfigure}[b]{0.49\textwidth}
  \tikzsetnextfilename{contraction_identity_check}
\begin{tikzpicture}
    \pgfdeclarelayer{bg}
    \pgfsetlayers{bg,main}

    \node[squarenode] (plus) {$+$};
    \node[quadnode] (W1) [above right=0.2cm and 0.2cm of plus] {$W[\alpha]$};
    \node[quadnode] (W2) [below right=0.2cm and 0.2cm of plus] {$W[\beta]$};
    \node[quadnode, minimum height=1.8cm, minimum width=0.5cm] (CNOT) [right=1.5cm of plus] {};
    \node[phase] (CNOTcontrol) [right=0.52cm of W1] {};
    \node[circlewc] (CNOTtarget) [below=0.98cm of CNOTcontrol] {};

    \draw (plus.west) -- node[above] {$X$} ++(-0.5cm,0);
    \draw (plus.north) to[out=90, in=180] (W1.west);
    \draw (plus.south) to[out=-90, in=180] (W2.west);
    \draw [thick] (CNOTcontrol.south) to (CNOTtarget.north);
    % We have to put the double dashed lines in a background layer, otherwise they look ugly
    \begin{pgfonlayer}{bg}
      \draw [style=double] (W1.east) to ++(1.2cm,0) node[above] {$Q_1$};
      \draw [style=double] (W2.east) to ++(1.2cm,0) node[below] {$Q_2$};
    \end{pgfonlayer}

    % equivalence sign
    \node (equiv) [right=2.3cm of plus] {$\cong$};
    
    % right-hand side:
    \node[quadnode] (Wrhs) [above right=0.5cm and 0.5cm of equiv] {$W[\alpha\boxstar_l\beta]$};
    \node[quadnode] (Pkrhs) [below right=1.0cm and 0.0cm of Wrhs.south] {$P_L$};
    \node[squarenode] (eq3rhs) [above=0.45cm of Pkrhs] {\cequal};

    \draw (Wrhs.west) -- node[above] {$X$} ++(-0.5cm,0);
    \begin{pgfonlayer}{bg}
      \draw (Pkrhs) to node[right] {$L$} (eq3rhs);
      \draw (eq3rhs) to ++(0,0.5cm);
      \draw [style=double] (eq3rhs) to ++(1.3cm,0) node[below] {$Q_2$};
      \draw [style=double] (Wrhs.east) to ++(0.7cm,0) node[above] {$Q_1$};
    \end{pgfonlayer}
  \end{tikzpicture}
  \caption{}
  \label{fig:contraction_identity_check}
  \end{subfigure}
  \begin{subfigure}[b]{0.49\textwidth}
  \tikzsetnextfilename{contraction_identity_check_generalized}
  \begin{tikzpicture}
    \pgfdeclarelayer{bg}
    \pgfsetlayers{bg,main}

    \node[squarenode] (plus) {$+$};
    \node[quadnode] (W1) [above right=0.2cm and 0.2cm of plus] {$W[\alpha_i]$};
    \node[quadnode] (W2) [below right=0.2cm and 0.2cm of plus] {$W[\beta_j]$};
    \node[quadnode, minimum height=1.8cm, minimum width=0.5cm] (CNOT) [right=1.5cm of plus] {};
    \node[phase] (CNOTcontrol) [right=0.41cm of W1] {};
    \node[circlewc] (CNOTtarget) [below=0.98cm of CNOTcontrol] {};

    \draw (plus.west) -- node[above] {$X$} ++(-0.5cm,0);
    \draw (plus.north) to[out=90, in=180] (W1.west);
    \draw (plus.south) to[out=-90, in=180] (W2.west);
    \draw [thick] (CNOTcontrol.south) to (CNOTtarget.north);
    % We have to put the double dashed lines in a background layer, otherwise they look ugly
    \begin{pgfonlayer}{bg}
      \draw [style=double] (W1.east) to ++(1.2cm,0) node[above] {$Q_1$};
      \draw [style=double] (W2.east) to ++(1.2cm,0) node[below] {$Q_2$};
      \draw (W1.north) to node[above left=-0.1cm and 0] {$I$} ++(0,0.6cm);
      \draw (W2.south) to node[below left=-0.1cm and 0] {$J$} ++(0,-0.6cm);
    \end{pgfonlayer}

    % equivalence sign
    \node (equiv) [right=2.3cm of plus] {$\cong$};
    
    % right-hand side:
    \node[quadnode] (Wrhs) [above right=0.5cm and 0.5cm of equiv] {$W[\alpha_i\boxstar_l\beta_j]$};
    \node[quadnode] (Pkrhs) [below right=1.0cm and 0.4cm of Wrhs.south] {$P_L$};
    \node[squarenode] (eq1rhs) [left=0.9cm of Pkrhs] {\cequal};
    \node[squarenode] (eq2rhs) [above left=0.15cm and 0.3cm of Pkrhs] {\cequal};
    \node[squarenode] (eq3rhs) [above=0.45cm of Pkrhs] {\cequal};

    \draw (Wrhs.west) -- node[above] {$X$} ++(-0.5cm,0);
    \begin{pgfonlayer}{bg}
      \draw (eq1rhs) to ++(0,1.2cm);
      \draw (eq1rhs) to node[below left=-0.16cm and 0] {$I$} ++(0,-0.65cm);
      \draw (eq2rhs) to ++(0,0.7cm);
      \draw (eq2rhs) to node[below left=0.1cm and 0] {$J$} ++(0,-1.2cm);
      \draw (eq1rhs) to (Pkrhs);
      \draw (eq2rhs) to (Pkrhs);
      \draw (Pkrhs) to node[right] {$L$} (eq3rhs);
      \draw (eq3rhs) to ++(0,0.5cm);
      \draw [style=double] (eq3rhs) to ++(1.0cm,0) node[below] {$Q_2$};
      \draw [style=double] (Wrhs.east) to ++(0.7cm,0) node[above] {$Q_1$};
    \end{pgfonlayer}
  \end{tikzpicture}
  \caption{}
  \label{fig:contraction_identity_check_generalized}
  \end{subfigure}
  \caption{(a) Graphical depiction of check node contraction identity using Forney-style factor graphs. The equivalence $\cong$ is to be understood as equality of the exterior functions of the two factor graphs. (b) represents a generalizations of this identity, where the channel parameters are classically conditioned on some $n_I$-bit and $n_J$-bit systems $I$ and $J$ for some $n_I,n_J\geq 0$. $\{\alpha_i\}_{i=0,\dots,2^{n_I}-1}$ and $\{\beta_j\}_{j=0,\dots,2^{n_J}-1}$ are some sets of angles between $0$ and $\pi$. In (a) $L$ denotes a 1-bit system which is distributed as $P_L(0)=(1+\cos\alpha\cos\beta)/2, P_L(1)=(1-\cos\alpha\cos\beta)/2$. The corresponding 1-bit system $L$ in (b) is distributed as $P_L(0)=(1+\cos\alpha_i\cos\beta_j)/2, P_L(1)=(1-\cos\alpha_i\cos\beta_j)/2$.}
  \label{fig:contraction_identities_check}
  \end{figure}

\bigskip

Return now to the factor graph describing the effective channel from $X_r$ to $Y_1,\dots,Y_n$ followed by the operations of $V_r$. 
Since the former graph is a binary tree, the contraction identities depicted in \Cref{fig:contraction_identity_equality,fig:contraction_identity_check} can be used to simplify the full factor graph at the channel nodes. 
Simplifying a given channel pair reduces the number of channel nodes in the graph by one and converts the remaining channel node to a classically conditioned channel node. 
The extended identities of \Cref{fig:contraction_identity_equality_generalized,fig:contraction_identity_check_generalized} can then be used to further simplify these channel nodes in an iterative fashion.

Iteratively simplifying the factor graph in this manner $s$ times to remove $s$ channel nodes results in a factor graph of the form depicted in~\cref{fig:contraction_intermediate_state}. 
After $n-1$ steps, the final simplified factor graph describing the CQ channel from $X_r$ to the output of BPQM is therefore of the following form
\begin{equation}
\label{eq:simplifiedFFG}
\tikzsetnextfilename{SimplifiedFFG}
  \begin{tikzpicture}[baseline=-20pt]
    \pgfdeclarelayer{bg}
    \pgfsetlayers{bg,main}
    \node[quadnode] (W) {$W[\varphi_i]$};
    \node[squarenode] (eq) [below=0.4cm of W] {\cequal};
    \node[quadnode] (P) [left=0.4cm of eq] {$P_I$};
    \node[quadnode] (zero) [below=0.4cm of eq] {$(\ket{0}\bra{0})^{\otimes n-k}$};
    \begin{pgfonlayer}{bg}
      \draw (W.west) to node[above] {$X_r$} ++(-0.5cm,0);
      \draw[style=double] (W) to ++(+1.5cm,0) node[above] {$D$};
      \draw[style=double] (eq) to ++(+1.5cm,0) node[above] {$A$};
      \draw[style=double] (zero) to ++(+1.5cm,0) node[above] {$Z$};
      \draw (W.south) to node[right] {$I$} (eq.north);
      \draw (eq.west) to (P.east);
    \end{pgfonlayer}
  \end{tikzpicture}\,\,,
\end{equation}
where $I$ is a $(k-1)$-bit random variable, $P_I$ is some distribution defined on $I$, $(\varphi_0,\dots,\varphi_{2^{k-1}-1})$ is some set of angles $\in (0,\pi)$, $Z$ is the system of $n-k$ zero qubits produced by equality node operations, $A$ is the system of $k-1$ ancilla qubits produced by check node operations and $D$ is the data qubit produced by the root node.
We illustrate this contraction process on the 5-bit example in~\cref{fig:5bitcode_contraction}. 
A step-by-step version of this example which includes the intermediate steps of the contraction is provided in~\cref{app:step_by_step}.

\begin{figure}
  \centering
    \tikzsetnextfilename{IntermediateStateFFG}
  \begin{tikzpicture}
    \pgfdeclarelayer{bg}
    \pgfsetlayers{bg,main}
    \node[cloud, draw,cloud puffs=10,cloud puff arc=120, aspect=2, inner ysep=1em, text width=3cm] (cloud) {tree with $n-1-s$ `$=$' and `$+$' nodes};
    \node[quadnode] (W1) [above right=0.2cm and 1.5cm of cloud] {$W[\varphi_{i_1}]$};
    \node[squarenode] (c1) [below=0.5cm of W1] {\cequal};
    \node[quadnode] (P1) [left=0.3cm of c1] {$P_{I_1}$} ;
    \node[quadnode] (Wlast) [below right=0.0cm and 1.5cm of cloud] {$W[\varphi_{i_{n-s}}]$};
    \node[squarenode] (clast) [below=0.5cm of Wlast] {\cequal};
    \node[quadnode] (Plast) [left=0.3cm of clast] {$P_{I_{n-s}}$} ;
    \node[quadnode] (zero) [below right=1.8cm and 0.5cm of cloud] {$(\ket{0}\bra{0})^{\otimes e}$} ;
    \node (dots) [right=1.0cm of cloud] {$\vdots$};

    \draw (cloud.west) to node[above] {$X_r$} ++(-1.0cm,0);
    \draw (cloud) to (W1);
    \draw (cloud) to (Wlast);
    \draw (P1) to (c1);
    \draw (c1) to node[right] {$I_1$} (W1);
    \draw (Plast) to (clast);
    \draw (clast) to node[right] {$I_{n-s}$} (Wlast);
    \begin{pgfonlayer}{bg}
      \draw[double] (W1) to ++(1.5cm,0);
      \draw[double] (Wlast) to ++(1.5cm,0);
      \draw[double] (c1) to ++(1.5cm,0);
      \draw[double] (clast) to ++(1.5cm,0);
      \draw[double] (zero) to ++(2.5cm,0);
    \end{pgfonlayer}
  \end{tikzpicture}
  \caption{Intermediate state after $s$ contraction steps have been performed, where $s\in \{0,\dots,n-1\}$. Here $I_j$ is a $n_j$-bit system for $j=1,\dots,n-s$, for some $n_j\geq 0$. The $P_{I_j}$ are distributions defined over the system $I_j$. We denote by $e$ the number contraction steps that involved equality nodes, and therefore $s-e$ is the number of contraction steps that involved check nodes. It must hold that $\sum_j n_j = s-e$.}
  \label{fig:contraction_intermediate_state}
\end{figure}

\begin{figure}
  \centering
  \begin{adjustbox}{width=\textwidth}
  \tikzsetnextfilename{5bitFFGContraction}
  \begin{tikzpicture}
    \pgfdeclarelayer{bg}
    \pgfsetlayers{bg,main}

    \node[squarenode] (eq1) {\cequal};
    \node[squarenode] (eq2) [above right=0.2cm and 0.3cm of eq1] {\cequal};
    \node[squarenode] (ch1) [above right=0.05cm and 0.3cm of eq2] {$+$};
    \node[squarenode] (ch2) [below right=0.05cm and 0.3cm of eq2] {$+$};
    \node[quadnode] (W5) [above right=-0.1cm and 0.3cm of ch1] {$W[\theta_5]$};
    \node[quadnode] (W3) [below=0.15cm of W5] {$W[\theta_3]$};
    \node[quadnode] (W4) [below=0.15cm of W3] {$W[\theta_4]$};
    \node[quadnode] (W2) [below=0.15cm of W4] {$W[\theta_2]$};
    \node[quadnode] (W1) [below=0.15cm of W2] {$W[\theta_1]$};
    \node[quadnode, minimum width=0.4cm, minimum height=1.0cm] (CNOT1) [right=1.43cm of ch1] {};
    \node[phase] (CNOT1control) [right=0.19cm of W5] {};
    \node[circlewc] (CNOT1target) [below=0.33cm of CNOT1control] {};
    \node[quadnode, minimum width=0.4cm, minimum height=1.0cm] (CNOT2) [below=0.20cm of CNOT1] {};
    \node[phase] (CNOT2control) [right=0.19cm of W4] {};
    \node[circlewc] (CNOT2target) [below=0.33cm of CNOT2control] {};
    \node[quadnode] (U1) [right=2.1cm of ch1, minimum height=1.0cm] {\scalebox{0.7}{$U_{\ostar}(\theta_5\boxstar_i\theta_3,\theta_4\boxstar_j\theta_2)$}};
    \node[quadnode] (U2) [right=0.32cm of U1, minimum height=1.0cm] {\scalebox{0.7}{$U_{\ostar}((\theta_5\boxstar_i\theta_3)\ostar(\theta_4\boxstar_j\theta_2),\theta_1)$}};
    \node[squarenode] (control1i) [right=1.5cm of W4] {$\bullet$};
    \node[squarenode] (control1j) [right=2.5cm of W2] {$\bullet$};
    \node[squarenode] (control2i) [right=4.5cm of W4] {$\bullet$};
    \node[squarenode] (control2j) [right=5.5cm of W2] {$\bullet$};

    \draw[thick] (CNOT1control.south) to (CNOT1target.north);
    \draw[thick] (CNOT2control.south) to (CNOT2target.north);
    \draw (eq1.east) to[out=0, in=180] (eq2.west);
    \draw (eq2.east) to[out=0, in=180] (ch1.west);
    \draw (eq2.east) to[out=0, in=180] (ch2.west);
    \draw (ch1.east) to[out=0, in=180] (W5.west);
    \draw (ch1.east) to[out=0, in=180] (W3.west);
    \draw (ch2.east) to[out=0, in=180] (W4.west);
    \draw (ch2.east) to[out=0, in=180] (W2.west);
    \draw (eq1.east) to[out=0, in=180] (W1.west);
    \draw (eq1.west) -- node[above] {$X_1$} ++(-0.40cm,0);
    
    % We have to put the double dashed lines in a background layer, otherwise they look ugly
    \begin{pgfonlayer}{bg}
      % note: distance between lanes is 0.6cm
      \draw [style=double] (W5.east) to ++(7.2cm,0) node[above] {$Q_1$};
      \draw [style=double] (W3.east) to ++(0.6cm,0) to ++(0.1cm,-0.6cm) to ++(6.5cm,0) node[above] {$Q_3$};
      \draw [style=double] (W4.east) to ++(0.6cm,0) to ++(0.1cm,0.6cm) to ++(2.55cm,0) to ++(0.2cm,-1.8cm) to ++(3.75cm,0) node[above] {$Q_5$};
      \draw [style=double] (W2.east) to ++(7.2cm,0) node[above] {$Q_4$};
      \draw [style=double] (W1.east) to ++(3.25cm,0) to ++(0.2cm,1.8cm) to ++(3.75cm,0) node[above] {$Q_2$};
      \draw [style=double] (control1i.north) to node[below right=-0.15cm and 0.1cm] {$I$} ++(0,0.3cm);
      \draw [style=double] (control1j.north) to node[below right=-0.4cm and 0.0cm] {$J$} ++(0,1.0cm);
      \draw [style=double] (control2i.north) to node[below right=-0.15cm and 0.1cm] {$I$} ++(0,0.3cm);
      \draw [style=double] (control2j.north) to node[below right=-0.4cm and 0.0cm] {$J$} ++(0,1.0cm);
    \end{pgfonlayer}
    
    % equivalence sign
    \node (equiv) [above right=0cm and 10.1cm of eq1] {$\cong$};

    % right-hand side
    \node[quadnode] (Wrhs) [above right=0.5cm and 0.3cm of equiv] {\scalebox{0.7}{$W[(\theta_5\boxstar_i\theta_3)\ostar(\theta_4\boxstar_j\theta_2)\ostar\theta_1]$}};
    \node[squarenode] (eqirhs) [below left=0.3cm and 0.2cm of Wrhs.south] {\cequal};
    \node[squarenode] (eqjrhs) [below right=0.8cm and 0.2cm of Wrhs.south] {\cequal};
    \node[quadnode] (Pirhs) [left=0.4cm of eqirhs] {$P_{I}$};
    \node[quadnode] (Pjrhs) [left=1.15cm of eqjrhs] {$P_{J}$};
    \node[quadnode] (zero1rhs) [below=1.3cm of Wrhs.south] {$\ket{0}\bra{0}$};
    \node[quadnode] (zero2rhs) [below=1.8cm of Wrhs.south] {$\ket{0}\bra{0}$};

    \draw (Wrhs.west) to node[above] {$X_1$} ++(-0.5cm,0);
    \draw (Pirhs) to (eqirhs);
    \draw (eqirhs) to node[right] {$I$} +(0,0.5cm);
    \draw (Pjrhs) to (eqjrhs);
    \draw (eqjrhs) to node[above right] {$J$} ++(0,1.0cm);
    \begin{pgfonlayer}{bg}
      \draw [style=double] (Wrhs) to ++(2.0cm,0) node[above] {$Q_1$};
      \draw [style=double] (eqirhs) to ++(2.4cm,0) node[above] {$Q_3$};
      \draw [style=double] (eqjrhs) to ++(1.6cm,0) node[above] {$Q_4$};
      \draw [style=double] (zero1rhs) to ++(2.0cm,0) node[above] {$Q_2$};
      \draw [style=double] (zero2rhs) to ++(2.0cm,0) node[above] {$Q_5$};
    \end{pgfonlayer}

  \end{tikzpicture}\end{adjustbox}
  \caption{Contraction of the Forney-style factor graph representing the effect of the BPQM algorithm on the 5-bit code. Both factor graphs depict the quantum state of the system after executing $V_1$ on the channel output. By using the contraction identities from~\cref{fig:contraction_identities_equality,fig:contraction_identities_check}, this state can be characterized with the simpler factor graph on the right-hand side, which allows for a simple identification of the Helstrom measurement. $P_I$ and $P_J$ are identical distributions of a single-bit system defined by $P_I(0)=P_J(0)=(1+\cos^2\theta)/2$ and $P_I(1)=P_J(1)=(1-\cos^2\theta)/2$. A step-by-step version that illustrates the intermediate steps of this contraction is provided in~\cref{app:step_by_step}.}
  \label{fig:5bitcode_contraction}
\end{figure}
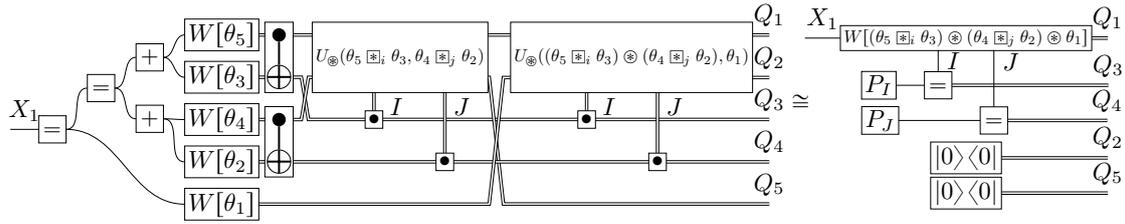

The final step in the single bit decoding procedure for the codeword bit $X_r$ is to measure the remaining data qubit in the $\ket{\pm}$ basis. 
It is easy to see that measurement in this basis is optimal for distinguishing the outputs of $W[\theta]$ for all values of $\theta$, as the two states are symmetrically arranged about the $z$-axis in the Bloch sphere no matter the angle between them. 
Therefore, a $\ket{\pm}$ measurement is also optimal for channels whose angle parameter is set by additional inputs, as in \eqref{eq:simplifiedFFG}. 
Since all operations in $V_r$ are unitary and thus reversible, BPQM decoding of $X_r$ must be optimal, i.e., it must realize the Helstrom measurement.

\section{Block optimality of BPQM}\label{sec:optimality} 
Now we turn to the blockwise performance of the BPQM decoder. 
Consider again some $(n,k)$ binary linear code $\mathcal{C}$ that exhibits a tree Tanner graph.
\begin{theorem}\label{thm:optimality}
The BPQM algorithm performs block optimal decoding.
\end{theorem}
The first step in the proof of this statement is the fact that optimal decoding is realized by the pretty-good measurement (PGM)~\cite{belavkin_1975,hausladen_pretty_1994}. 
Therefore it will suffice to show that BPQM realizes a measurement equivalent to the PGM.

Recall that the PGM for a set of pure states $\mathcal{S}=\{\ket{\phi_1},\dots,\ket{\phi_m}\}$ with associated probabilities $p_i,i=1,\dots,m$ is given by the POVM with elements rank-1 positive operators $E_i:=p_i\rho^{-1/2}\ket{\phi_i}\bra{\phi_i}\rho^{-1/2}$, where $\rho=\sum_{i=1}^m p_i \ketbra{\phi_i}$, and the matrix inverse is taken on its support. 
When $\sum_{i=1}^m E_i< \id$ the POVM formally requires an additional operator $E_{m+1}=\id-\sum_{i=1}^m E_i$, though note that the probability of observing $E_{m+1}$ is zero for all states in $\mathcal S$. 
Under certain conditions, the PGM is the optimal measurement for determining which state is actually present, averaged over the prior probabilities $p_i$. 
Relevant here is geometric uniformity of the states.
\begin{lemma}[Optimality of the PGM~\cite{ban_optimum_1997,sasaki_quantum_1998,eldar_2001}]\label{lem_optimality_pgm}
  If the $\{\ket{\phi_i}\}_i$ are equally probable, linearly independent and form a geometrically uniform state set,~i.e., $\{\ket{\phi_i}\}_i=\{U\ket{\phi} | U\in\mathcal{G}\}$ for some state $\ket{\phi}$ and some finite Abelian group $\mathcal{G}$, then the PGM elements are orthogonal and the PGM distinguishes the states $\ket{\phi_i}$ with the highest possible average probability of success. 
% Theorem 3
\end{lemma}
To apply the statement to our setting, denote by $\ket{\Psi_{\vec{x}}}$ the joint pure state output of the quantum channels when the codeword $\vec{x}\in\mathcal{C}$ is transmitted, i.e.,
\begin{equation}
  \ket{\Psi_{\vec{x}}} := \qstate{x_1}{\theta_1}\otimes\dots\otimes\qstate{x_n}{\theta_n} \, .
\end{equation}
Then $\mathcal S$ is the set of $\ket{\Psi_{\vec{x}}}$, of cardinality $m=2^k$, and the prior probability of each state is simply $1/m$.  
\begin{lemma}\label{lem_gus}
  $\mathcal{S}=\{\ket{\Psi_{\vec{x}}}\}_{\vec{x}}$ is a geometrically uniform state set.
\end{lemma}
\begin{proof}
  This follows rather directly from the realization that we can write $\ket{\Psi_{\vec{x}}}=\sigma_Z^{\vec{x}}\ket{\Psi_{\vec{0}}}$ for any codeword $\vec{x}\in\mathcal{C}$ where $\vec{0}\in\mathcal{C}$ is the all-zero codeword and $\sigma_Z^{\vec{x}}:=\sigma_Z^{x_1}\otimes \sigma_Z^{x_2}\otimes \dots \otimes \sigma_Z^{x_n}$. Since the codewords $\vec{x}$ form an Abelian group under addition, this implies that the unitaries $\{\sigma_Z^{\vec{x}} | \vec{x}\in\mathcal{C}\}$ form an Abelian group under multiplication.
\end{proof}
Since $\qstate{0}{\theta}$ and $\qstate{1}{\theta}$ are linearly independent, the set $\{\ket{\Psi_{\vec{x}}}\}_{\vec{x}}$ consisting of tensor products of these two vectors must also be linearly independent.
The PGM POVM elements therefore consist of the projectors onto the states $\ket{f_{\vec{x}}} := \frac{1}{\sqrt{2^k}}\rho^{-1/2}\ket{\Psi_{\vec{x}}}$ for $\rho:=\sum_{\vec x}\frac1{2^k} \ketbra{\Psi_{\vec x}}$. 

Before we show that BPQM realizes a measurement equivalent to the PGM, we need to get a better understanding of how the single codeword bit decoding operation acts on the channel outputs.
\begin{lemma}\label{lem:purestate_evolution}
Consider a codeword $\vec{x}\in\mathcal{C}$ and the unitary $V_r$ that describes the action of the BPQM node operations corresponding to some MPG $G$ of $\mathcal{C}$ with respect to $X_r$ for $r\in\{1,\dots,k\}$. 
Then 
\begin{equation}
\label{eq:compressed_state}
V_r\ket{\Psi_{\vec{x}}}=\sum\limits_{\vec{j}\in\mathbb Z_2^{k-1}} c_{\vec{j}}(x_1,\dots,x_k) \sqrt{p_{\vec{j}}}\qstate{x_r}{\theta_{\vec{j}}}_D\otimes \ket{\vec{j}}_A\otimes \ket{0^{n-k}}_Z\,,
\end{equation}
where the probabilities and angles $\{(p_{\vec{j}}, \theta_{\vec{j}}) : \vec{j}\in\{0,1\}^{k-1}\}$ correspond exactly to the angle and probability entries of the branch list of the root node of $G$.
The $c_{\vec{j}}$ are certain functions $\{0,1\}^k\rightarrow \{+1,-1\}$ which fulfill the ``orthogonality'' property
\begin{equation}
\label{eq:orthogonality}
    \sum\limits_{\tilde{x}_r} c_{\vec{j}_1}(x_1,\dots,x_k)c_{\vec{j}_2}(x_1,\dots,x_k) = 2^{k-1}\delta_{\vec{j}_1,\vec{j}_2} \quad \forall x_r\in\{0,1\},\forall \vec{j}_1,\vec{j}_2\in\{0,1\}^{k-1}
  \end{equation}
where $\tilde{x}_r$ denotes all free variables except $x_r$, i.e., $x_1,\dots,x_{r-1},x_{r+1},\dots,x_k$\end{lemma}
Here $D$ denotes the qubit to be measured after $V_r$, $A$ denotes the $k-1$ ancilla qubits produced by check node operations, and $Z$ denotes the $n-k$ zero qubits arising from equality node operations. 
The lemma uses the notation $\ket{\vec{l}}:=\ket{l_1}\otimes\dots\otimes\ket{l_m}$ for $\vec{l}\in\{0,1\}^m$ where each ket in the tensor product is a computational basis state of a qubit.
The symbol $0^m\in\{0,1\}^m$ denotes the all-zero vector.
Due to its length, we delegate the proof to~\cref{app:proof_optimality}.

The lemma can be thought of as a generalization of the contraction process described in~\cref{sec:contraction}. 
Both describe the state of the system after the BPQM node operations have been applied on the channel output, but from different perspectives. 
In~\cref{lem:purestate_evolution} we assume that we know the full codeword $\vec{x}$ that was transmitted, whereas in the contraction argument we assume all codeword bits except $X_r$ to be uniformly random. 
This difference makes it impossible to make the same contraction-style argument for~\cref{lem:purestate_evolution}, and instead a more careful approach has to be used in the proof. 
A central ingredient is the form of the output of a check node operation, the pure state version of \eqref{eq:checknodeactionmixed}: 
\begin{equation}
\label{eq:purestate_checknode}
  \cnot{}\cdot \qstate{y_1}{\alpha}\otimes\qstate{y_2}{\beta} = 
  %\sqrt{\frac{1+\cos(\alpha)\cos(\beta)}{2}}\qstate{y_1\oplus y_2}{\alpha\boxstar_0\beta}\ket{0} + (-1)^{y_2}\sqrt{\frac{1-\cos(\alpha)\cos(\beta)}{2}}\qstate{y_1\oplus y_2}{\alpha\boxstar_1\beta}\ket{1}\\
  \sum_{j=0}^1 (-1)^{jy_2}\sqrt{\frac{1+(-1)^j\cos(\alpha)\cos(\beta)}{2}}\qstate{y_1\oplus y_2}{\alpha\boxstar_j\beta}\otimes \ket{j}\,,
  \end{equation}
where $y_1,y_2\in\{0,1\}$, $\alpha,\beta\in (0,\pi)$ and $\oplus$ denotes addition modulo 2.
The proof proceeds by iteratively applying this identity (which is derived in \cref{app:proof_optimality}) and \eqref{eq:Ustar_identity} to $\ket{\Psi_{\vec{x}}}$. 

Observe that \Cref{lem:purestate_evolution} allows us to retrieve the factor graph contraction result by marginalizing over all inputs except $X_r$. 
Consider the state
\begin{equation}
 \tau_r:=V_r\left(\frac{1}{2^{k-1}}\sum\limits_{\tilde{x}_r} \ket{\Psi_{\vec{x}(x_1,\dots,x_k)}}\bra{\Psi_{\vec{x}(x_1,\dots,x_k)}} \right) V_r^{\dagger}\,,
\end{equation}
where $\vec{x}(x_1,\dots,x_k)$ is the (unique) extension of the tuple $(x_1,\dots,x_k)$ to a codeword in $\mathcal{C}$. 
Using~\cref{lem:purestate_evolution} we obtain 
\begin{equation}\resizebox{1.0 \textwidth}{!}{$
  \tau_r = \frac{1}{2^{k-1}} \sum\limits_{\vec{j}_1,\vec{j}_2\in\mathbb Z_2^{k-1}} \sum\limits_{\tilde{x}_r} c_{\vec{j}_1}(x_1,\dots,x_k) c_{\vec{j}_2}(x_1,\dots,x_k)  \sqrt{p_{\vec{j}_1}p_{\vec{j}_2}} \qstate{x_r}{\theta_{\vec{j}_1}}\qstateconj{x_r}{\theta_{\vec{j}_2}}_D 
  \otimes \ket{\vec{j}_1}\bra{\vec{j}_2}_A\otimes \ket{0^{n-k}}\bra{0^{n-k}}_Z\,.
$}\end{equation}
The off-diagonal $\vec{j}_1\neq\vec{j}_2$ terms cancel due to the orthogonality property of the $c_{\vec{j}}$ terms, leaving 
\begin{equation}
  \tau_r= \sum\limits_{\vec{j}\in\mathbb Z_2^{k-1}} p_{\vec{j}} \qstate{x_r}{\theta_{\vec{j}}}\qstateconj{x_r}{\theta_{\vec{j}}}_D\otimes \ket{\vec{j}}\bra{\vec{j}}_A\otimes \ket{0^{n-k}}\bra{0^{n-k}}_Z\,,
\end{equation}
which corresponds exactly to the fully-contracted factor graph encountered in~\cref{sec:contraction}. 

We now have all ingredients necessary to prove~\cref{thm:optimality}. 
The main idea is to consider the $\ket{f_{\vec{x}}}$ under the basis transformation defined by the unitary $V_r$. 
The above argument implies that for any $r\in\{1,\dots,k\}$ we can write
\begin{equation}\label{eq:basistrafo_rho}
  \rho = V_r^{\dagger} \cdot \sum\limits_{\vec{j}\in\mathbb Z_2^{k-1}} p_{\vec{j}} \left( \sum\limits_{z=0,1}\frac{1}{2}\qstate{z}{\theta_{\vec{j}}}\qstateconj{z}{\theta_{\vec{j}}}_D\right) \otimes \ket{\vec{j}}\bra{\vec{j}}_A \otimes \ket{0^{n-k}}\bra{0^{n-k}}_Z \cdot V_r \, .
\end{equation}
Furthermore, by~\cref{lem:purestate_evolution} we have
\begin{equation}\label{eq:basistrafo_psi}
    \ket{\Psi_{\vec{x}}} = V_r^{\dagger} \sum\limits_{\vec{j}'\in\mathbb Z_2^{k-1}} c_{\vec{j}'}(x_1,\dots,x_k) \sqrt{p_{\vec{j}'}} \qstate{x_r}{\theta_{\vec{j'}}}_D\otimes\ket{\vec{j}'}_A\otimes\ket{0^{n-k}}_Z \, .
\end{equation}
Then, by making use of 
\begin{equation}
  \sum\limits_{z}{\tfrac{1}{2}\qstate{z}{\theta_{\vec{j}}}\qstateconj{z}{\theta_{\vec{j}}}} = \cos^2\tfrac{\theta_{\vec{j}}}{2}\ket{0}\bra{0} + \sin^2\tfrac{\theta_{\vec{j}}}{2}\ket{1}\bra{1}\,,
\end{equation}
and inserting~\cref{eq:basistrafo_rho,eq:basistrafo_psi} into $\ket{f_{\vec{x}}}=\frac{1}{\sqrt{2^k}}\rho^{-1/2}\ket{\Psi_{\vec{x}}}$, we thus obtain 
\begin{equation}\label{eq_pgm_element}
  \ket{f_{\vec{x}}} 
  = V_r^{\dagger} \cdot \frac{1}{\sqrt{2}}(\ket{0}_D+(-1)^{x_r}\ket{1}_D)\otimes \Big( \frac{1}{\sqrt{2^{k-1}}}\sum\limits_{\vec{j}\in \mathbb Z_2^{k-1}} c_{\vec{j}}(x_1,\dots,x_k)\ket{\vec{j}}_A \Big)\otimes \ket{0^{n-k}}_Z \, .
\end{equation}

Now consider the space $\mathcal{H}$ spanned by the vectors $\{\ket{\Psi_{\vec{x}}}\}_{\vec{x}}$, which by the form of the PGM is exactly the space spanned by the orthonormal basis $\{\ket{f_{\vec{x}}}\}_{\vec{x}}$. 
By~\cref{eq_pgm_element} it holds that on $\mathcal{H}$ the sequence of $V_r$, the projection on the qubit $D$ described by $H\ket{m_r}\bra{m_r}H$ for $m_r\in\{0,1\}$, and $V_r^{\dagger}$ acts identically to the projector
\begin{equation}\label{eq:def_pgm_projector}
  \Pi_{x_r=m_r}:=\sum\limits_{\vec{x}\in\mathcal{C} : x_r=m_r}\ket{f_{\vec{x}}}\bra{f_{\vec{x}}} \, .
\end{equation}
Since
\begin{equation}
  \ket{f_{\vec{m}}}\bra{f_{\vec{m}}} = \prod\limits_{i=1}^k{\Pi_{x_i=m_i}}
\end{equation}
due to the orthonormality of the PGM, this implies that BPQM realizes a measurement equivalent to the PGM when acting any state $\ket{\Psi_{\vec{x}}}$ for  $\vec{x}\in\mathcal{C}$.

\bigskip
Let us also comment that while optimal decoding can be performed bitwise in the classical case, it  is not necessarily the same as optimal bitwise decoding. 
Given an optimal decoder, which is a deterministic function $\vec{y}\mapsto \vec{x}'(\vec{y})$ of the observed output, the blockwise-optimal bitwise output for the $r$th codeword bit is simply the $r$th bit of $\vec{x}'(\vec{y})$. 
However, bitwise optimal decoding does not always yield the same result. 
Take the BSC with crossover probability $p$ and the five-bit code from \cref{fig:5bitTanner} as an example.\footnote{This example is due to H.~D.~Pfister (personal communication).}
%Here the issue is somewhat clouded by the fact that the optimal decoder is not unique. 
For the output $\vec{y}=00011$, the optimal decoder uniquely returns the codeword $10011$ for all $p\in [0,\tfrac12]$, so the estimate $X_1'$ for the first bit is $X'_1=1$. 
The optimal bitwise decoder, however, is only concerned with the joint distribution of $X_1$ and $Y_1,\dots Y_5$, i.e., the probability of all codewords with $X_1=1$ versus those with $X_1=0$. 
In this case the ratio of these probabilities is $(1 - 2 p + 2 p^2)^2/4 (1 - p)^3 p$, which implies that for $p\gtrapprox 0.228$ the estimate will be $X'_1=0$. 
The reason that bitwise-optimal decoding and blockwise-optimal decoding coincide in the considered quantum setting can be traced back to the particular structure of the blockwise-optimal decoding (i.e.~it is a rank-1 projective measurement) implied by~\cref{lem_optimality_pgm,lem_gus}.

\section{Message-passing BPQM}\label{sec:message_passing}
The goal of this section is to introduce a message-passing algorithm called \emph{message-passing BPQM} which can approximate BPQM to arbitrary precision and which does not require a compilation step and whose quantum circuit implementation only scales polynomially with the blocklength.  
In BPQM the check nodes store ancilla qubits which indicate the angle arguments needed to perform subsequent equality node operations. 
The idea of the message-passing algorithm is to circumvent this issue by passing the angle information itself along with the data qubits.
Appendix~\ref{app:relation_to_bp} discusses how this same issue plays out in the case of classical BP.

\subsection{Decoding a single codeword bit}
\label{sec:message_passing_singlebit}
\subsubsection{Hybrid classical-quantum algorithm}
Consider first the simplified case of decoding only a single codeword bit $X_r$ for $r\in\{1,\dots,n\}$. 
It then suffices to measure the ancilla qubits as they are created at check nodes, as it will not be necessary to rewind the entire decoding operation. 
The measurement result can then be used to update the angle information. 

In more detail, the hybrid classical-quantum message-passing algorithm operates on the MPG for $X_r$ as follows.  
Messages are passed from the leaves inward into the tree until they reach the root. 
Each message consists of one qubit and one real number in the range $(0,\pi)$. 
The initial messages transmitted by the $i$th leaf nodes consist of the corresponding channel output qubit $Y_i$ as well as the corresponding channel parameter $\theta_i$. 
As in BPQM, the final qubit output by the root is measured in the $\ket{\pm}$ basis, yielding the decoder's estimate $\hat X_r$ of the value of $X_r$. 

The algorithm is fully specified by explaining how equality and check nodes process their two received messages into an output message to be sent up the tree. 
Denote the first and second data qubits of the two input messages of a node operation by $Q_1$ and $Q_2$, respectively, and denote the angle part of both input messages as $\varphi_1,\varphi_2$. 
The equality node operation consists of performing the unitary $U_{\ostar}(\varphi_1,\varphi_2)$ on the system $Q_1Q_2$. 
The output message then consists of the qubit $Q_1$ as well as the angle $\varphi_1\ostar\varphi_2$.
The check node operation consists of performing a $\cnot{}$ gate from $Q_1$ to $Q_2$ and then measuring the qubit $Q_2$ in the computational basis. 
Denote the measurement outcome by $l\in\{0,1\}$. 
The output message of the check node contains the qubit $Q_1$ as well as the angle $\varphi_1\boxstar_l\varphi_2$.
Note that in contrary to BPQM, here the node operations do not only consist of quantum gates, but also contain some classical computation.

\subsubsection{Optimality of bitwise decoding}
Note that in this hybrid quantum-classical algorithm, every message passed over some edge $e$ will be of the form $(\qstate{z}{\varphi},\varphi)$ form some $z\in\{0,1\}$ and $\varphi\in (0,\pi)$. 
Put differently, we know the angle of all data qubits passed around the MPG, but not the value of $z$. 
This can be verified with a simple inductive argument. 
It clearly holds for the initial messages sent by the leaves. 
Suppose $(\qstate{z_1}{\varphi_1},\varphi_1)$ and $(\qstate{z_2}{\varphi_2},\varphi_2)$ are the messages input to either a check or an equality node. 
At check nodes the required form is a consequence of \cref{eq:purestate_checknode}. 
At equality nodes it follows from the MPG that the state is such that $z_1=z_2$,\footnote{This follows directly from~\cref{lem:stronger44} in the proof of~\cref{lem:purestate_evolution}.} so the desired consequence follows directly from~\cref{eq:Ustar_identity}.
Thus, the hybrid classical-quantum message-passing algorithm implements the Helstrom measurement for determining the codeword bit.

\subsection{Decoding the complete codeword}
To construct a useful codeword decoder by sequentially decoding the individual codeword bits, it suffices to implement the above algorithm coherently. 
In particular, the angle information can be stored in a quantum register and the measurement of the final qubits (which produce the estimates for the decoded bits) can be replaced by a \cnot{} gate with a target ancilla qubit initialized in the state $\ket{0}$. 
Then the now-quantum angle register can be updated at check and equality nodes. 

Nominally this approach requires an infinite-dimensional quantum angle register to allow for a deterministic differentiation between any possible $\theta\in (0,\pi)$.
A natural solution to this problem is to use a discretized representation of the angle using $B$ qubits.
In order to simplify the analysis at a later point, we choose to represent an angle $\theta$ by its cosine $c:=\cos\theta\in (-1,1)$. 
Since $\theta$ lies in $(0,\pi)$ both representations contain the same information.  
More precisely, we restrict ourselves to only represent angles with their cosine being $c\in \mathcal{A}_B:=\{ -1 + 2\frac{1+k}{2^B+1} | k\in\{0,1,\dots,2^B-1\} \}$ with spacing 
\begin{equation}
\label{eq:deltadef}
\delta:=\frac{2}{2^B+1}
\end{equation}
between the individual representable values. 
Mathematically, this means that we choose the angle part of the message to be a quantum system living in a $2^B$-dimensional Hilbert space with some orthonormal basis denoted by $\{\anglestate{c} | c\in\mathcal{A}_B\}$. 

The goal of this section is to formalize the algorithm precisely and to show that the resulting decoder approaches the ideal performance of BPQM as $B$ goes to infinity. 
First we formally introduce the fully coherent message-passing BPQM algorithm that allows for decoding of the complete codeword. 
Then we will analyze and bound the effect of the discretization errors and determine their effects on the performance of the decoder.

\subsubsection{Quantum message-passing algorithm}
\label{sec:quantummessagepassing}

The decoding of a single message bit $X_r$ can be regarded as an algorithm on an MPG of the code with respect to $X_r$. 
Messages are passed from the leaves inward into the tree until they reach the root. 
Each message consists of a data qubit and a $B$-qubit register. 
The latter represents an angle cosine in $\mathcal{A}_B$. 
The data qubit of the final message generated by the root is measured in the $\ket{\pm}$ basis, producing the decoder estimate for the corresponding message bit. 
All message processing operations save the final measurement, are unitary and can therefore be undone after the value of $X_r$ has been estimated. 
This procedure is repeated for all codeword bits to be decoded, just as in BPQM. %in the same fashion as described in~\cref{sec:completedecoding}.

The message-passing algorithm is fully specified by describing how equality and check node operations take two messages, process them with some unitary operation, and produce an output message. 
This is visually depicted with the quantum circuits in~\cref{fig:message_passing_rules}. 
The two data qubits are denoted by $D_1,D_2$ and the angle registers are denoted by $C_1,C_2$ in the representation. The output message contains the data qubit $D'$ and the angle register $C'$. 
The gates 
\begin{equation}
  c_1,c_2 \mapsto \quant(c_1\cdot c_2)
  \text{ and }
  c_1,c_2\mapsto \quant\left(\frac{c_1 + (-1)^l c_2}{1 + (-1)^l c_1c_2}\right)
\end{equation}
represent quantum circuits that reversibly compute the corresponding classical functions $\mathcal{A}_B\times\mathcal{A}_B\rightarrow\mathcal{A}_B$ in the computational basis, where 
$\quant:(-1,1)\rightarrow\mathcal{A}_B$ denotes the rounding function 
\begin{equation}
  \quant(c) = \argmin_{\tilde{c}\in \mathcal{A}_B} |c-\tilde{c}| \, .
\end{equation}
Since it is understood how to implement these functions on a classical computer, we can simply translate these efficient classical circuits to an efficient quantum circuit. 
The quantum circuit will generally require $a_e$, respecitvely $a_c$, additional ``scratch'' qubits initialized to $\ket0$ for the realization of the computation as a reversible circuit. 
By uncomputing the classical circuit, we can ensure that these scratch qubits return to $\ket 0$  after execution of the node operation~\cite{bennett_1973}.
This is essential in allowing us to reuse the qubits at a later point without losing coherence of our state.

Notice that each node produces quantum information in $C_1,C_2$ and $D_2$ that it does not pass on to the next node. 
These registers are stored at each node for unwinding the decoding operations at a later point.

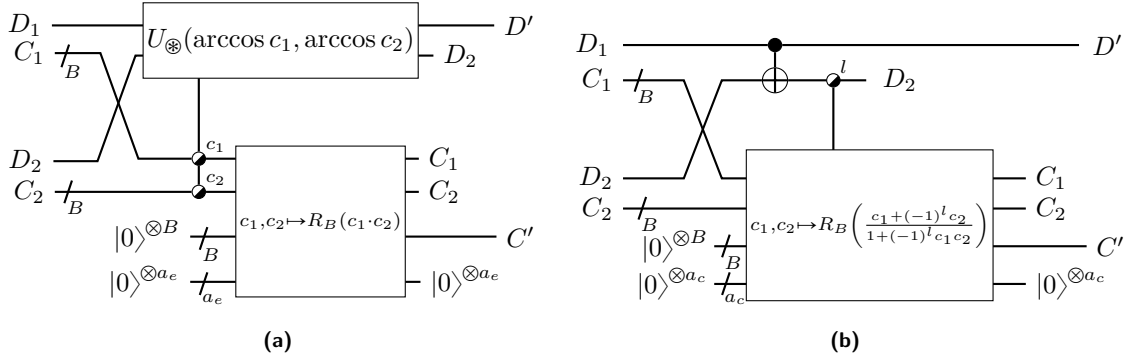
\begin{figure}
  \centering
  \begin{subfigure}[b]{0.49\textwidth}
    \tikzsetnextfilename{MPEqualityOperation}
  \begin{tikzpicture}
      \pgfdeclarelayer{bg}

    \pgfsetlayers{bg,main}
    \node[quadnode, minimum height=1.0cm] (U)  {$U_{\ostar}(\arccos c_1,\arccos c_2)$};
    \node[ucontrolled] (ctrl1) [below left=1cm and -0.8cm of U] {};
    \node[ucontrolled] (ctrl2) [below=0.25cm of ctrl1] {};
    \node (ctrllabel1) [above right=-0.1cm and -0.1cm of ctrl1] {$\scriptstyle c_1$};
    \node (ctrllabel1) [above right=-0.1cm and -0.1cm of ctrl2] {$\scriptstyle c_2$};
    \node[quadnode, minimum height=2.0cm] (classical) [below right=0.9cm and -2.4cm of U] {$\scriptstyle c_1,c_2\mapsto \quant(c_1\cdot c_2)$};

    \begin{pgfonlayer}{bg}
      \draw[thick] (U.west)+(0,0.2cm) to node[left=0.6cm] {$D_1$} ++(-1.2cm,0.2cm);
      \draw[thick] (U.west)+(0,-0.2cm) to ++(-0.1cm,-0.2cm) to ++(-0.5cm,-1.4cm) to node[left=0.3cm] {$D_2$} ++(-0.58cm,0);
      \draw[thick] (U.east)+(0,0.2cm) to node[right=0.5cm] {$D'$} ++(1.0cm, 0.2cm);
      \draw[thick] (U.east)+(0,-0.2cm) to node[right=0.1cm] {$D_2$} ++(0.2cm, -0.2cm);
      \draw[thick] (ctrl1.west) to ++(-0.7cm,0) to ++(-0.5cm,1.4cm) to node[left=0.3cm] {$C_1$} ++(-0.6cm,0);
      \draw[thick] (ctrl2.west) to node[left=0.9cm] {$C_2$} ++(-1.8cm,0);
      \draw[thick] (ctrl1.east) to node[right=1.4cm] {$C_1$} ++(2.8cm,0);
      \draw[thick] (ctrl2.east) to node[right=1.4cm] {$C_2$} ++(2.8cm,0);
      \draw[thick] (ctrl1.north) to ++(0,1.0cm);
      \draw[thick] (ctrl2.north) to (ctrl1.south);
      \draw[thick] (classical.west)+(0,-0.2cm) to node[left=0.3cm] {$\ket{0}^{\otimes B}$} ++(-0.6cm,-0.2cm);
      \draw[thick] (classical.west)+(0,-0.8cm) to node[left=0.3cm] {$\ket{0}^{\otimes a_e}$} ++(-0.6cm,-0.8cm);
      \draw[thick] (classical.east)+(0,-0.2cm) to node[right=0.6cm] {$C'$} ++(1.2cm,-0.2cm);
      \draw[thick] (classical.east)+(0,-0.8cm) to node[right=0.1cm] {$\ket{0}^{\otimes a_e}$} ++(0.2cm,-0.8cm);

      % lines indicating number of qubits
      \node (linemarker1) [below left=0.08cm and 0.9cm of U.west] {};
      \node (linemarker2) [left=1.5cm of ctrl2] {};
      \node (linemarker3) [below left=0.1cm and 0.3cm of classical.west] {};
      \node (linemarker4) [below left=0.7cm and 0.3cm of classical.west] {};
      \draw[thick] (linemarker1)+(-0.05cm,-0.15cm) to node[below right=0 and -0.15cm] {$\scriptstyle B$} ++(0.05cm,0.15cm);
      \draw[thick] (linemarker2)+(-0.05cm,-0.15cm) to node[below right=0 and -0.15cm] {$\scriptstyle B$} ++(0.05cm,0.15cm);
      \draw[thick] (linemarker3)+(-0.05cm,-0.15cm) to node[below right=0 and -0.15cm] {$\scriptstyle B$} ++(0.05cm,0.15cm);
      \draw[thick] (linemarker4)+(-0.05cm,-0.15cm) to node[below right=0 and -0.15cm] {$\scriptstyle a_e$} ++(0.05cm,0.15cm);
    \end{pgfonlayer}
  \end{tikzpicture}
  \caption{}
  \end{subfigure}
  \begin{subfigure}[b]{0.49\textwidth}
  \tikzsetnextfilename{MPCheckOperation}
  \begin{tikzpicture}
      \pgfdeclarelayer{bg}

    \pgfsetlayers{bg,main}
    \node[phase] (CNOTctrl) {};
    \node[circlewc] (CNOTtarg) [below=0.2cm of CNOTctrl] {};
    \node[ucontrolled] (ctrl) [right=0.5cm of CNOTtarg] {};
    \node (ctrllabel) [above right=-0.1cm and -0.1cm of ctrl] {$\scriptstyle l$};
    \node[quadnode, minimum height=2.0cm] (classical) [below right=0.8cm and -0.5cm of CNOTtarg] {$\scriptstyle c_1,c_2\mapsto \quant\left(\frac{c_1+(-1)^lc_2}{1+(-1)^lc_1c_2}\right)$};
    
    \begin{pgfonlayer}{bg}
      \draw[thick] (CNOTctrl) to node[left=1.0cm] {$D_1$} ++(-2.0cm,0);
      \draw[thick] (CNOTctrl) to node[right=2.0cm] {$D'$} ++(4.0cm,0);
      \draw[thick] (CNOTtarg) to ++(-0.7cm,0) to ++(-0.5cm,-1.3cm) to node[left=0.4cm] {$D_2$} ++(-0.8cm,0);
      \draw[thick] (CNOTtarg) to node[right=0.6cm] {$D_2$} ++(1.2cm,0);
      \draw[thick] (CNOTtarg)+(-2.0cm,0) to node[left=0.4cm] {$C_1$} ++(-1.2cm,0) to ++(0.5cm,-1.3cm) to node[right=2.0cm] {$C_1$} ++(4.0cm,0);
      \draw[thick] (CNOTtarg)+(-2.0cm,-1.7cm) to node[left=0.65cm] {$C_2$} ++(-0.7cm,-1.7cm) to node[right=2.0cm] {$C_2$} ++(4.0cm,0);
      \draw[thick] (CNOTtarg)+(-0.8cm,-2.2cm) to node[left=0.5cm] {$\ket{0}^{\otimes B}$} ++(0.3cm,-2.2cm) to node[right=1.9cm] {$C'$} ++(3.8cm,0);
      \draw[thick] (CNOTtarg)+(-0.8cm,-2.7cm) to node[left=0.5cm] {$\ket{0}^{\otimes a_c}$} ++(0.3cm,-2.7cm) to node[right=1.5cm] {$\ket{0}^{\otimes a_c}$} ++(3.0cm,0);
      \draw[thick] (CNOTctrl.south) to (CNOTtarg.north);
      \draw[thick] (ctrl.south) to ++(0,-1.5cm);

      % lines indicating number of qubits
      \node (linemarker1) [left=1.5cm of CNOTtarg] {};
      \node (linemarker2) [below left=1.45cm and 1.5cm of CNOTtarg] {};
      \node (linemarker3) [below left=1.95cm and 0.4cm of CNOTtarg] {};
      \node (linemarker4) [below left=2.45cm and 0.4cm of CNOTtarg] {};
      \draw[thick] (linemarker1)+(-0.05cm,-0.15cm) to node[below right=0 and -0.15cm] {$\scriptstyle B$} ++(0.05cm,0.15cm);
      \draw[thick] (linemarker2)+(-0.05cm,-0.15cm) to node[below right=0 and -0.15cm] {$\scriptstyle B$} ++(0.05cm,0.15cm);
      \draw[thick] (linemarker3)+(-0.05cm,-0.15cm) to node[below right=0 and -0.15cm] {$\scriptstyle B$} ++(0.05cm,0.15cm);
      \draw[thick] (linemarker4)+(-0.05cm,-0.15cm) to node[below right=0 and -0.15cm] {$\scriptstyle a_c$} ++(0.05cm,0.15cm);

    \end{pgfonlayer}
  \end{tikzpicture}
  \caption{}
  \end{subfigure}
  \tikzexternaldisable
  \caption{Quantum circuit representation of message-passing processing operations for equality and check node operations. Both types of node operations take two messages $(D_1,C_1)$, $(D_2,C_2)$ and produce an output message $(D',C')$. The symbol {\protect\tikz\protect \node [ucontrolled] (c) {}; } depicts on which qubits a gate is uniformly controlled.}
  \tikzexternalenable
  \label{fig:message_passing_rules}
\end{figure}

It remains to specify the implementation of the uniformly-controlled $U_{\ostar}$ gate depicted in~\cref{fig:message_passing_rules}, as the naive implementation would again require a quantum circuit size exponential in $B$. 
As shown in~\cref{fig:uc_ustar_identity}, the $U_{\ostar}(\varphi_1,\varphi_2)$ gate can be decomposed into \cnot{} gates and single-qubit $R_y$ rotations. 
The angles of these rotations are given by
\begin{align}\label{eq:alpha}
    \alpha = -\arccos(a_+) - \arccos(b_+) \pmod{2\pi}\,, \\ \beta = -\arccos(a_+) + \arccos(b_+) \pmod{2\pi} \, ,
\end{align}
where, as before,
\begin{align}\label{eq:beta}
    a_{\pm}:=\frac{1}{\sqrt{2}}\frac{\cos(\frac{\arccos c_1-\arccos c_2}{2})\pm \cos(\frac{\arccos c_1+\arccos c_2}{2})}{\sqrt{1+c_1c_2}}, \\ 
    b_{\pm}:=\frac{1}{\sqrt{2}}\frac{\sin(\frac{\arccos c_1+\arccos c_2}{2})\pm \sin(\frac{\arccos c_1-\arccos c_2}{2})}{\sqrt{1-c_1c_2}} \, 
\end{align}
and $c_1:=\cos\varphi_1,c_2:=\cos\varphi_2$. 
For convenience, we choose $\alpha,\beta$ to both lie in the range $[0,2\pi]$. 
We delegate the (short) proof of this identity to~\cref{app:ustar_identity}. 

Now only the single-qubit $R_y$ gates need to be controlled by the angle registers. 
This is easily done by first computing $\alpha$ and $\beta$, storing the results in two further $B$-qubit registers, and then sequentially performing $B$ controlled-$R_y$ rotations, as depicted in~\cref{fig:uc_ustar_simplification}. 
More specifically, first apply a quantum circuit implementing the classical functions $c_1,c_2\mapsto \quantalphabeta{B}{\alpha(c_1,c_2)}$ and $c_1,c_2\mapsto \quantalphabeta{B}{\beta(c_1,c_2)}$, where $\quantalphabeta{B}:[0,2\pi]\rightarrow \tilde{A}_B:=\{ 2\pi\frac{k}{2^B-1} | k=0,\dots,2^B-1\}$ is the rounding function $\quantalphabeta{B}{\varphi} := \argmin_{\tilde{\varphi}\in \tilde{\mathcal{A}}_B} |\varphi - \tilde{\varphi}|$, and store the result in a two ancilla $B$-qubit registers. 
Here again we require a certain number of scratch qubits to realize the classical computation, which we denote by $a_{\text{ang}}$. 
The $R_y$ rotations can now be uniformly controlled by these ancilla registers.

In contrary to a general uniformly-controlled gate, the two remaining unformly-controlled $R_y$ rotations can be realized efficiently, as described in the following.
Consider a $R_y$ rotation gate on the qubit $T$ that is uniformly controlled by the angle $\varphi$ stored digitally in the $B$ qubits $C_1,\dots,C_B$, i.e., the state $\ket{b_1}_{C_1}\otimes\dots\otimes\ket{b_B}_{C_B}$ for $b_i\in\{0,1\}$ represents the angle $\varphi=2\pi\sum_{i=1}^n b_i2^{-i}$. This uniformly-controlled gate can be realized with $B$ controlled $R_y$ gates (i.e.~two-qubit gates) where the $i$-th rotation is targeted on the qubit $T$ and controlled by the qubit $C_i$ and has the angle $2\pi2^{-i}$.

As we only use a finite number of qubits to represent the angles $\alpha$ and $\beta$, this implementation of a uniformly-controlled $U_{\ostar}$ gate exhibits discretization errors, which can in principle be made arbitrarily small by letting $B$ go to infinity.

\tikzexternaldisable
\begin{figure}
  \centering
  \begin{subfigure}[b]{1.0\textwidth}
    \begin{center}
    \resizebox{0.7\textwidth}{!}{
      \begin{quantikz}
        & \gate[wires=2]{U_{\ostar}(\arccos c_1,\arccos c_2)} & \qw \\
        & & \qw
      \end{quantikz}
      $\cong$
      \begin{quantikz}
        & \targ{}   &  \qw                          & \ctrl{1} & \qw         & \ctrl{1} & \qw \\
        & \ctrl{-1} &  \gate{R_y(\alpha)} & \targ{}  & \gate{R_y(\beta)} & \targ{}  & \qw
      \end{quantikz}
      }
    \end{center}
    \caption{}
    \label{fig:uc_ustar_identity}
  \end{subfigure}
  \newline
  \begin{subfigure}[b]{1.0\textwidth}
    \begin{center}
      \resizebox{0.9\textwidth}{!}{
      \begin{quantikz}
        & \gate[wires=2]{U_{\ostar}(\arccos c_1,\arccos c_2)} & \qw \\
        &                                                     & \qw \\
        & \uctrl{c_1}\vqw{-1}                                 & \qw \\
        & \uctrl{c_2}\vqw{-1}                                 & \qw
      \end{quantikz}
      $\approx$
\begin{quantikz}
& \qw                  & \targ{}                                              & \qw                    & \ctrl{1} & \qw                   & \ctrl{1}                                                  & \qw & \qw \\ [-0.4cm]
& \qw                  & \ctrl{-1}                                                     & \gate{R_y(\alpha)}      & \targ{}  & \gate{R_y(\beta)}      & \targ{}                                                   & \qw & \qw \\ [-0.4cm]
& \ket{0}^{\otimes B}       & \gate[wires=5]{\substack{\text{compute}\\\alpha,\beta}} & \uctrl{\alpha}\vqw{-1} & \qw      & \qw                   & \gate[wires=5]{\substack{\text{uncompute}\\\alpha,\beta}} & \qw & \ket{0}^{\otimes B}  \\ [-0.6cm]
& \ket{0}^{\otimes B}       &                                                         & \qw                    & \qw      & \uctrl{\beta}\vqw{-2} & \qw                                                       & \qw & \ket{0}^{\otimes B} \\ [-0.2cm]
& \qw                       &                                                         & \qw                    & \qw      & \qw                   & \qw                                                       & \qw & \qw \\
& \qw                       &                                                         & \qw                    & \qw      & \qw                   & \qw                                                       & \qw & \qw \\ [-0.3cm]
& \ket{0}^{\otimes a_{ang}} &                                                         & \qw                    & \qw      & \qw                   & \qw                                                       & \qw & \ket{0}^{\otimes a_{ang}}
\end{quantikz}
    }\end{center}
    \caption{}
    \label{fig:uc_ustar_simplification}
  \end{subfigure}
  \caption{(a) Decomposition of equality node unitary into single-qubit and \cnot{} gates. The angles of rotation $\alpha$ and $\beta$ are introduced in the main text in~\cref{eq:alpha,eq:beta}. (b) Efficient realization of a uniformly-controlled equality node unitary. The $\approx$ symbol illustrates that the right-hand side is not perfectly equal to the left-hand side, since the angles $\alpha,\beta$ are only computed to an accuracy of $B$ bits. However, this discretization error can be made arbitrarily small.}
  \label{fig:uc_ustar}
\end{figure}
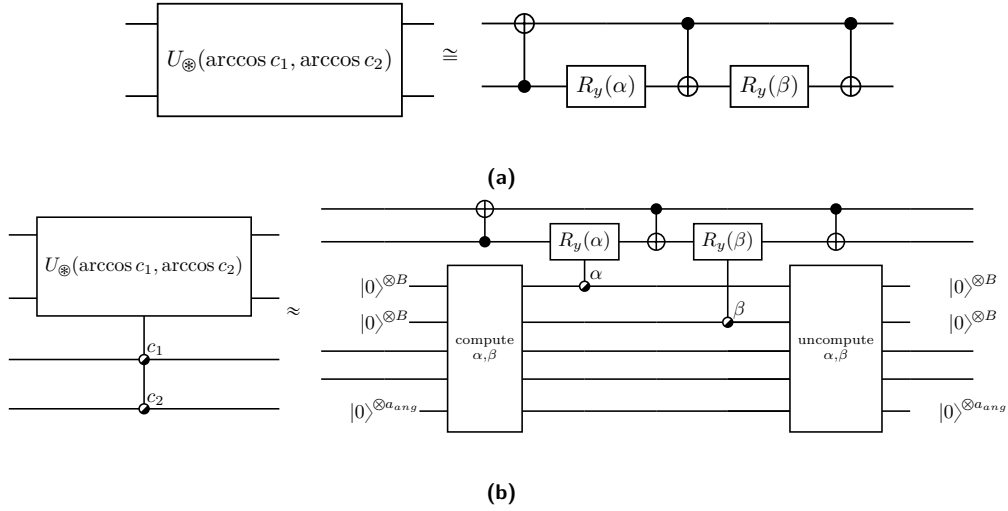
\tikzexternalenable

\subsubsection{Performance of the algorithm}
For the rest of this section, we focus on bounding the scaling of the discretization error in $B$ so as to determine how large $B$ has to be chosen to reach a desired decoding accuracy. 
By~\cref{lem:purestate_evolution}, the state after performing the BPQM unitary $V_r$ according to some MPG $G$ w.r.t $X_r$ on the state $\qstate{x_1}{\theta_1}\otimes\dots\otimes\qstate{x_n}{\theta_n}$ (but before performing the projective measurement in the $\ket{\pm}$ basis) for some $\vec{x}\in\mathcal{C}$ can be written in the form
\begin{align}\label{eq:original_state}
  \sum\limits_{\vec{j}\in\{0,1\}^{k-1}}{ \gamma_{\vec{j}} \qstate{x_r}{\theta_{\vec{j}}}_D \otimes \ket{0^{n-k}}_Z \otimes \ket{\vec{j}}_A }\,,
\end{align}
where the angles $\theta_{\vec{j}}\in (0,\pi)$ and squares of the amplitudes $\gamma_{\vec{j}}\in\mathbb{R}$ correspond precisely to the angles and probabilities in the branch list of the root node of $G$.
Here again, we denote by $D$ the output qubit on which the measurement is to be performed, $Z$ the system of ancilla qubits produced by equality node operations and $A$ the system of qubits produced by check node operations. Similarly, we now want to determine how the final state after the execution of message-passing BPQM on the identical MPG and input state looks like:
\begin{lemma}\label{lem:discretization_error}
  After executing the message-passing variant of BPQM according to the MPG $G$ w.r.t. $X_r$ on the state $\qstate{x_1}{\theta_1}\otimes\dots\otimes\qstate{x_n}{\theta_n}$ where $\vec{x}\in\mathcal{C}$, one obtains a state of the form
  \begin{align}\label{eq:messpass_state}
    \sum\limits_{\vec{j}\in\{0,1\}^{k-1}}{ \tilde{\gamma}_{\vec{j}} \ket{\tilde{\varphi}_{\vec{j}}}_{DZ} \otimes \ket{\vec{j}}_A \otimes \anglestate{\tilde{c}_{\vec{j}}}_C } \otimes \ket{\tilde{s}_{\vec{j}}}_S
  \end{align}
  for some $(n-k+1)$-qubit states $\ket{\tilde{\varphi}_{\vec{j}}}$, $\tilde{\gamma}_{\vec{j}}\in\mathbb{R}$, $\tilde{c}_{\vec{j}}\in\mathcal{A}_B$ and some $(n-1)2B$-qubit states $\ket{\tilde{s}_{\vec{j}}}$ such that
  \begin{enumerate}
    \item $\sum\limits_{\vec{j}\in\{0,1\}^{k-1}} \gamma_{\vec{j}}^2 | \cos\theta_{\vec{j}} - \tilde{c}_{\vec{j}} | \leq (2^{n+1}-3)\delta$
    \item \label{item:2}$\sum\limits_{\vec{j}\in\{0,1\}^{k-1}} | \gamma_{\vec{j}}^2 - \tilde{\gamma}_{\vec{j}}^2 | \leq 2^{n+\tfrac12}\pi\sqrt{\delta}\frac{1}{3}26^n$
    \item \label{item:3}$\sum\limits_{\vec{j}\in\{0,1\}^{k-1}} \gamma_{\vec{j}}^2 \norm{ \qstate{x_r}{\theta_{\vec{j}}}\otimes\ket{0^{n-k}} - \ket{\tilde{\varphi}_{\vec{j}}} } \leq 2^{n+\tfrac12}\pi\sqrt{\delta}26^n$
  \end{enumerate}
\end{lemma}
Note that the state is defined on a larger system than before: besides $D$, $Z$, and $A$, there are new systems $C$ and $S$ present in this expression. We denote by $C$ the $B$-qubit register that contains the angle-part of the message generated by the root node of $G$. Furthermore, $S$ denotes the system of all qubits produced in the node operations which are not passed on to the next node and are not part of $Z$ or $A$. By consulting~\cref{fig:message_passing_rules}, one can quickly see that there are exactly two such $B$-qubit registers (denoted $C_1$ and $C_2$ in the figure) per node operation. The precise states $\ket{\tilde{s}_{\vec{j}}}$ will not be of any importance, but we cannot discard these qubits, as they are required to undo the node operations at a later point.

Note that $\ket{\tilde{\varphi}_{\vec{j}}}$ describes the joint state of the output data qubit as well as the ancilla qubits produced by equality nodes operations. So, ideally, if discretization errors are small, then we should hope to find $\ket{\tilde{\varphi}_{\vec{j}}}\approx\qstate{x_r}{\theta_{\vec{j}}}\otimes \ket{0^{n-k}}$. Similarly, the $B$-qubit system $C$ should well represent the angle of the data qubits, i.e., $\tilde{c}_{\vec{j}}\approx\cos\theta_{\vec{j}}$. \Cref{lem:discretization_error} tells us how the discretization errors propagate through the message-passing algorithm. Due to the length and the technical nature, we delegate the proof to~\cref{app:proof_discretization_error}. We merely note here that the central idea of the lemma is that we do not directly make a statement about the errors $|\cos\theta_{\vec{j}}-\tilde{c}_{\vec{j}}|$ and $\norm{ \qstate{x_r}{\theta_{\vec{j}}}\ket{0^{n-k}} - \ket{\tilde{\varphi}_{\vec{j}}} }$ in each branch of the wave function, but rather make some sort of probability-averaged statement. This is necessary, because the classical computation of the function
\begin{equation}
  c_1,c_2\mapsto \frac{c_1+(-1)^lc_2}{1+(-1)^lc_1c_2}
\end{equation}
is numerically unstable in the regime where $(-1)^lc_1c_2$ approaches $-1$, as the right-hand side is not uniformly continuous in that neighborhood, so some branches of the wave function might experience large errors caused by the discretization. Fortunately, using some algebraic tricks we can show that in the probability-averaged\footnote{The probability average is taken over the possible ancilla states $\vec{j}\in\{0,1\}^{k-1}$ with probabilities $\gamma_{\vec{j}}^2$.} examination of the problem as formulated in~\cref{lem:discretization_error}, the errors remain small, so only branches with small probability amplitudes suffer from strong numerical errors.

\Cref{lem:discretization_error} can be applied to make a statement about the decoding performance of message-passing BPQM and how it compares to the ideal decoder, which is realized by BPQM. 
\begin{theorem}\label{thm:error_probab}
  Denote by $p^{(mp)}$ and $p^{(ideal)}$ the probabilities of sequentially decoding all the codeword bits $X_1,\dots,X_k$ correctly using the message-passing BPQM and BPQM algorithms. One has 
  \begin{equation}
    p^{(ideal)} - p^{(mp)} \leq \frac{2^{9/4}\sqrt{\pi}}{\sqrt{3}}n 2^{n\cdot \left(\frac{3}{2}+\frac{1}{2}\frac{\log 26}{\log 2}\right) - \frac{1}{4}B} \, .
  \end{equation}
\end{theorem}
This immediately implies that we can choose $B=\mathcal{O}(n+\log\frac{1}{\epsilon})$ to guarantee that the probability of successful decoding deviates by at most $\epsilon$ from the ideal value.
Due to its verbosity, we delegate the proof to~\cref{app:proof_complete_codeword}.
The main idea of the proof is to use~\cref{lem:discretization_error} iteratively to capture the error caused by the correponding message-passing BPQM operations, as well as the rewinding thereof.
In order for~\cref{lem:discretization_error} to be applicable, it is of central importance to argue that the intermediate states obtained after decoding a certain number of codeword bits are close to a state in the subspace spanned by the PGM basis elements.
This is required, since~\cref{lem:discretization_error} only makes a statement about how the action of BPQM and message-passing BPQM differ on states of the form $\qstate{x_1}{\theta_1}\otimes\dots\otimes\qstate{x_n}{\theta_n}$ and not on general states.

\subsection{Circuit complexity}\label{sec:complexity}
In this section we discuss the quantum circuit complexity of message-passing BPQM. In the case where we decode only a single codeword bit as described in~\cref{sec:message_passing_singlebit}, one requires exactly $n$ qubits, one for each channel output. 
The circuit depth scales with the number $n-1$ of node operations, and each node operation can be implemented in constant depth. 
The calculations done by the classical co-processor, which keeps track of the involved angles, are considered negligible.

The analysis is a bit more involved for the codeword decoder. 
The number of node operations now scales as $(n-1)k=\mathcal{O}(n^2)$ as we have to decode $k$ codeword bits. The complexity of every individual node operation depends on $B$ as well as the precise implementation of the involved classical functions. For classical computers the best known algorithm to evaluate these functions is based on the arithmetic-geometric mean operation~\cite{brent_1976} and when used in conjunction with the Harvey-Hoeven algorithm~\cite{harvey_2021} that allows for multiplication of $B$-bit numbers in $\mathcal{O}(B\log B)$ time, the time complexity of the classical operation would be $\mathcal{O}(B\log^2B)$.\footnote{It should be noted that the mentioned algorithms are only optimal when considering asymptotic runtimes for very large $B$. In practice, different algorithm would likely be more efficient when considering realistic values of $B$.}

It is not clear if and how these classical algorithm can be realized in a quantum circuit of the same asymptotic complexity. 
The current research on implementing transcendental functions, like trigonometric functions, on quantum computers is sparse~\cite{cao_2013,bhaskar_2016,haner_2018,wang_2020}. 
Wang \emph{et al.}\ recently proposed an algorithm that can realize our desired classical computation using $\mathcal{O}(B^3)$ operations and $\mathcal{O}(B^2)$ qubits~\cite{wang_2020}.
Using this result, message-passing BPQM exhibits a circuit depth that scales as $\mathcal{O}(nk B^3)$.

We require $(1+B)n$ qubits to store the $n$ initial messages and each of the $(n-1)$ node operations consumes a $B$-qubit register, as seen in~\cref{fig:message_passing_rules}. Next to these $(1+B)n+B(n-1)$ qubits, we also require $\max\{a_e,a_c,a_{ang}\}$ additional qubits as scratch register to realize the classical functions involved in every node operation, which scales as $\mathcal{O}(B^2)$ using the result from Wang \emph{et al.}. These qubits can be reused at every node. Therefore, the message-passing algorithm exhibits a circuit width that scales as $\mathcal{O}(B^2 + nB)=\mathcal{O}(n,B^2)$.

Again, we emphasize that the exact asympotic scaling in $B$ strongly depends on the realization of the classical computation, so no large importance should be attributed to the exact polynomial scaling factors used here. It should also be noted that the circuit depth could potentially scale more favorably for some certain families of codes. For instance, if the MPG is close to being a balanced binary tree, then the node operations can be performed in parallel. If we additionally assume a model of computation in which the classical scalar operations take constant time (or equivalently we just take $B$ to be fixed), then our circuit depth instead asymptotically scales as $\mathcal{O}(k\log n)$.

If instead we want to consider a model of computation that takes into account that $B$ has to be chosen larger to reach certain accuracies, we can directly apply~\cref{thm:error_probab}:
If we want to achieve a fixed decoding error $\epsilon$, we require $B=\mathcal{O}(n+\log\frac{1}{\epsilon})$.
Taking that into account, the resulting circuit complexities can be summarized as follows:
\begin{theorem}
\label{thm:depthwidth}
Consider the message-passing BPQM algorithm applied to a code of length $n$ such that the probability of successful decoding differs by at most $\epsilon$ from the ideal value.
The circuit depth and width satisfy
\begin{equation}
\mathrm{depth}=\mathcal{O}(n^5, \log^3\tfrac{1}{\epsilon}), \qquad \mathrm{width}=\mathcal{O}(n^2, \log^2\tfrac{1}{\epsilon}) \, .
\end{equation}
\end{theorem}

\subsection{Numerical results}\label{sec:mp_numerical_results}
We present some numerical results on how the number of qubits $B$ used to represent the angle-part of the messages impacts the decoder performance.
As an example, we consider a (17,11)  code for which the decoding of $X_1$ can be realized with the MPG depicted in~\cref{fig:17bitcode_mpg}.
This code was chosen because it exhibits multiple subsequent equality and check node operations. 
\Cref{fig:17bitcode_numerical_results} depicts the suboptimaility $\epsilon$ of the decoder, i.e., how much the probability of successfully decoding $X_1$ differs compared to the ideal Helstrom decoder, when applied on the all-zero codeword.
As we increase $B$, the suboptimality is seen to decrease exponentially quickly, which is in accord with the prediction from~\cref{thm:error_probab}.
However, the observed error is many orders of magnitude smaller than the values suggested by the bound in~\cref{thm:error_probab}.
The bound should mostly be considered as an analytic proof about the asymptotic scaling and does not accurately reflect the real number of qubits $B$ required to reach a certain decoding accuracy.

Note that message-passing BPQM involves more than $n(B+1)$ qubits, so a brute-force simulation of the circuit is infeasible for interesting values of $B$. Due to the structure of the algorithm, a more sophisticated simulation procedure can significantly reduce the memory requirements. For more details about the chosen linear code and the simulation procedure, we refer to~\cref{app:numerical_discretization_errors}.

\tikzset{dot/.style={fill,circle,inner sep=0.8pt}}
\tikzset{square/.style = {shape=regular polygon,regular polygon sides=4, text=black, draw=black, fill=white,inner sep=0.5mm, outer sep=0.0mm,text height=1.5ex,
    text depth=.25ex}}

\makeatletter
\tikzset{nomorepostaction/.code=\let\tikz@postactions\pgfutil@empty}
\makeatother
\tikzsetnextfilename{17bitMPG}
\begin{figure}
  \centering
  \begin{subfigure}[b]{0.60\textwidth}
  \resizebox{\textwidth}{!}{
  \begin{tikzpicture}
\def\xunit{0.875}
\def\yfirst{.875}
\def\ysecond{2*\yfirst-.1}
\def\ythird{3*\yfirst-.3}
\def\yfourth{4*\yfirst-.5}
\def\yfifth{5*\yfirst-.7}
\foreach \x in {0,1,2,...,16}{
  \node (n\x) at (\xunit*\x,0) {};
 % \node [below=0.1cm of n\x] {\x};
}
%\node[quadnode] (W17) at (n16) {$\theta_{17}$};
% \node[quadnode] (W1) at (n0) {$\theta_1$};
% \node[quadnode] (W2) at (n1) {$\theta_2$};
% \node[quadnode] (W3) at (n2) {$\theta_3$};
% \node[quadnode] (W4) at (n3) {$\theta_4$};
% \node[quadnode] (W5) at (n5) {$\theta_5$};
% \node[quadnode] (W6) at (n7) {$\theta_6$};
% \node[quadnode] (W7) at (n9) {$\theta_7$};
% \node[quadnode] (W8) at (n10) {$\theta_8$};
% \node[quadnode] (W9) at (n11) {$\theta_9$};
% \node[quadnode] (W10) at (n13) {$\theta_{10}$};
% \node[quadnode] (W11) at (n15) {$\theta_{11}$};
% \node[quadnode] (W12) at (n4) {$\theta_{12}$};
% \node[quadnode] (W13) at (n6) {$\theta_{13}$};
% \node[quadnode] (W14) at (n8) {$\theta_{14}$};
% \node[quadnode] (W15) at (n12) {$\theta_{15}$};
% \node[quadnode] (W16) at (n14) {$\theta_{16}$};
% 
\node[quadnode] (W1) at (n0) {$\theta_1$};
\node[quadnode] (W2) at (n1) {$\theta_2$};
\node[quadnode] (W3) at (n2) {$\theta_3$};
\node[quadnode] (W4) at (n3) {$\theta_4$};
\node[quadnode] (W5) at (n5) {$\theta_6$};
\node[quadnode] (W6) at (n7) {$\theta_8$};
\node[quadnode] (W7) at (n9) {$\theta_{10}$};
\node[quadnode] (W8) at (n10) {$\theta_{11}$};
\node[quadnode] (W9) at (n11) {$\theta_{12}$};
\node[quadnode] (W10) at (n13) {$\theta_{14}$};
\node[quadnode] (W11) at (n15) {$\theta_{16}$};
\node[quadnode] (W12) at (n4) {$\theta_{5}$};
\node[quadnode] (W13) at (n6) {$\theta_{7}$};
\node[quadnode] (W14) at (n8) {$\theta_{9}$};
\node[quadnode] (W15) at (n12) {$\theta_{13}$};
\node[quadnode] (W16) at (n14) {$\theta_{15}$};
\node[quadnode] (W17) at (n16) {$\theta_{17}$};

\foreach \y in {1,2,...,8}{
  \node[squarenode] (c\y) at (2*\y*\xunit-0.5*\xunit,\yfirst) {+};
  \node[dot] at (c\y.west) {};
}

\foreach \y in {1,2,3,4}{
  \node[squarenode] (e\y)  at (4*\y*\xunit-1.5*\xunit,\ysecond) {\cequal};
  \node[dot] at (e\y.west) {};
}

\foreach \y in {9,10}{
  \node[squarenode] (c\y)  at (8*\y*\xunit-67.5*\xunit,\ythird) {$+$};
  \node[dot] at (c\y.west) {};
}

\node[squarenode] (e5) at (8.5*\xunit,\yfourth) {\cequal};
\node[dot] at (e5.west) {};

\node[squarenode] (e6) at (4.25*\xunit,\yfifth) {\cequal};
\node[dot] at (e6.west) {};

% https://tex.stackexchange.com/a/5354
    \begin{scope}[
      every path/.style={postaction={nomorepostaction,decoration={markings,mark=at position 0.4 with {\arrow[>=stealth]{>}}},decorate}}
      ]
    \draw (W2) |- (c1);
     \draw (W3) |- (c1);
      \draw (W4) |- (c2);
      \draw (W12) |- (c2);
      \draw (W5) |- (c3);
      \draw (W13) |- (c3);
      \draw (W6) |- (c4);
      \draw (W14) |- (c4);
      \draw (W7) |- (c5);
      \draw (W8) |- (c5);
      \draw (W9) |- (c6);
      \draw (W15) |- (c6);
      \draw (W10) |- (c7);
      \draw (W16) |- (c7);
      \draw (W11) |- (c8);
      \draw (W17) |- (c8);
      
      \draw (c1) |- (e1);
      \draw (c2) |- (e1);
      \draw (c3) |- (e2);
      \draw (c4) |- (e2);
      \draw (c5) |- (e3);
      \draw (c6) |- (e3);
      \draw (c7) |- (e4);
      \draw (c8) |- (e4);
      
      \draw (e1) |- (c9);
      \draw (e2) |- (c9);
      \draw (e3) |- (c10);
      \draw (e4) |- (c10);
      
      \draw (c9) |- (e5);
      \draw (c10) |- (e5);
      
      \draw (W1) |- (e6);
      \draw (e5) |- (e6);
     
    \end{scope}
  \end{tikzpicture}
   }
  \caption{}
  \label{fig:17bitcode_mpg}
  \end{subfigure}
  \begin{subfigure}[b]{0.38\textwidth}
  \includegraphics[width=1.0\textwidth]{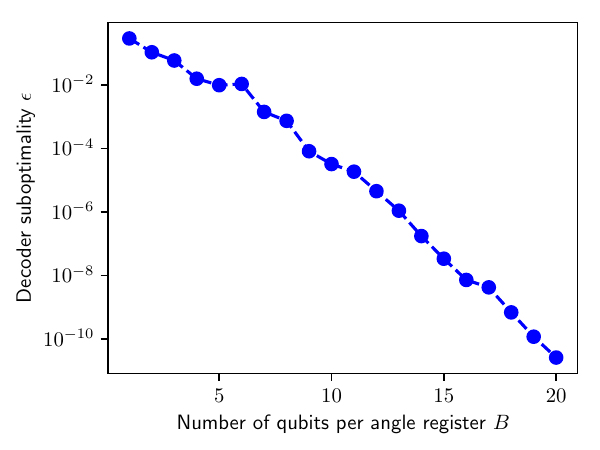}
  \caption{}
  \label{fig:17bitcode_numerical_results}
  \end{subfigure}
  \caption{(a) MPG of a (17,11) code w.r.t.\ $X_1$. The labels of the edges are not depicted. (b) Decoder suboptimality $\epsilon$ of the (17,11) code obtained through numerical simulations. The size $B$ of the angle registers is varied. The codeword to be decoded is the all-zero codeword.}
\end{figure}
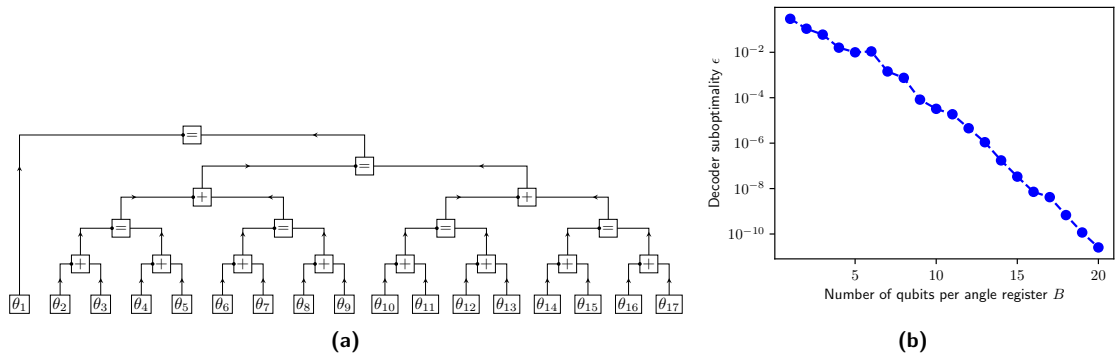

\section{BPQM on non-tree graphs}\label{sec:nontree_graphs}
Classical BP is commonly formulated as a message-passing algorithm on the Tanner graph, where each variable and check node receives messages, processes them and transmits the output to its neighbors. 
When the Tanner graph in question is a tree, there is a single natural way to define this algorithm: every node waits for all inputs to be received, and only then it generates the output message. 
This results into a deterministic and optimal marginalization. 

When the Tanner graph contains cycles, however, it is no longer possible to wait for all inputs to be obtained before generating an output message.  
For this reason, one has to carefully define a schedule according to which the message processing of the nodes operates. 
For the purposes of this work, we will follow the convention from~\cite{richardson_2008}: Message passing is performed in a sequence of rounds, where each round starts with the check nodes processing their inputs and sending their outputs to the variable nodes, upon which the variable nodes process their inputs and output the result among all their edges. 
Before the first round is executed, all variable nodes simply send the messages received from the channel output to their edges.

Consider a Tanner graph containing cycles on which classical BP is executed for $h$ rounds in order to decode the codeword bit $X_r$. 
The output of the algorithm is obtained from local node operations, so we can unroll the Tanner graph into a tree of depth $h$ that exactly represents this computation. 
We refer to this as the \emph{depth-$h$ computation tree of $X_r$}. 
Consider as an example the (8,4) linear code with Tanner graph depicted in \cref{fig:8bitcode_tanner_graph}. 
The associated computation trees of $X_1$ for $h=1,2,3$ are depicted in~\cref{fig:8bitcode_h1,fig:8bitcode_h2,fig:8bitcode_h3}.
The depth-$h$ computation tree exactly represents the computations required to obtain the output of BP after $h$ time steps. 
% 
% parity-check matrix
% \begin{equation}
%   H = 
%   \begin{pmatrix}
%     1 & 1 & 0 & 0 & 1 & 0 & 0 & 0 \\
%     0 & 1 & 1 & 0 & 0 & 1 & 0 & 0 \\
%     0 & 0 & 1 & 1 & 0 & 0 & 1 & 0 \\
%     1 & 0 & 0 & 1 & 0 & 0 & 0 & 1
%   \end{pmatrix}\,,
% \end{equation}
%whose Tanner graph is depicted in~\cref{fig:8bitcode_tanner_graph}. 

\begin{figure}
  \centering
  \begin{subfigure}[b]{0.22\textwidth}
  \begin{adjustbox}{width=\textwidth}
  \tikzsetnextfilename{8bitTannerGraph}
  \begin{tikzpicture}
    \node[squarenode] (CNW) {$+$};
    \node[roundnode] (X1) [right=0.5cm of CNW] {$X_1$};
    \node[squarenode] (CNE) [right=0.5cm of X1] {$+$};
    \node[roundnode] (X4) [below=0.5cm of CNE] {$X_4$};
    \node[squarenode] (CSE) [below=0.5cm of X4] {$+$};
    \node[roundnode] (X2) [below=0.5cm of CNW] {$X_2$};
    \node[squarenode] (CSW) [below=0.5cm of X2] {$+$};
    \node[roundnode] (X3) [right=0.5cm of CSW] {$X_3$};
    \node[roundnode] (X5) [above left=0.3cm and 0.3cm of CNW] {$X_5$};
    \node[roundnode] (X6) [below left=0.3cm and 0.3cm of CSW] {$X_6$};
    \node[roundnode] (X7) [below right=0.3cm and 0.3cm of CSE] {$X_7$};
    \node[roundnode] (X8) [above right=0.3cm and 0.3cm of CNE] {$X_8$};
    \draw (CNW) to (X1);
    \draw (X1) to (CNE);
    \draw (CNE) to (X4);
    \draw (X4) to (CSE);
    \draw (CSE) to (X3);
    \draw (X3) to (CSW);
    \draw (CSW) to (X2);
    \draw (X2) to (CNW);
    \draw (CNW) to (X5);
    \draw (CSW) to (X6);
    \draw (CSE) to (X7);
    \draw (CNE) to (X8);
  \end{tikzpicture}
  \end{adjustbox}
  \caption{}
  \label{fig:8bitcode_tanner_graph}
  \end{subfigure}
  \hspace{2mm}
  \begin{subfigure}[b]{0.22\textwidth}
  \begin{adjustbox}{width=\textwidth}
  \tikzsetnextfilename{8bitcode_h1}
  \begin{tikzpicture}
    \node[roundnode] (X1) {$X_1$};
    \node[squarenode] (C1) [below left=0.3cm and 0.3cm of X1] {$+$};
    \node[squarenode] (C2) [below right=0.3cm and 0.3cm of X1] {$+$};
    \node[roundnode] (X5) [left=0.5cm of C1] {$X_5$};
    \node[roundnode] (X2) [below=0.5cm of C1] {$X_2$};
    \node[roundnode] (X8) [right=0.5cm of C2] {$X_8$};
    \node[roundnode] (X4) [below=0.5cm of C2] {$X_4$};
    \draw (X1) to (C1);
    \draw (X1) to (C2);
    \draw (C1) to (X5);
    \draw (C1) to (X2);
    \draw (C2) to (X4);
    \draw (C2) to (X8);
  \end{tikzpicture}
  \end{adjustbox}
  \caption{}
  \label{fig:8bitcode_h1}
  \end{subfigure}
  \hspace{2mm}
  \begin{subfigure}[b]{0.22\textwidth}
  \begin{adjustbox}{width=\textwidth}
  \tikzsetnextfilename{8bitcode_h2}
  \begin{tikzpicture}
    \node[roundnode] (X1) {$X_1$};
    \node[squarenode] (C1) [below left=0.3cm and 0.3cm of X1] {$+$};
    \node[squarenode] (C2) [below right=0.3cm and 0.3cm of X1] {$+$};
    \node[roundnode] (X5) [left=0.5cm of C1] {$X_5$};
    \node[roundnode] (X2) [below=0.5cm of C1] {$X_2$};
    \node[roundnode] (X8) [right=0.5cm of C2] {$X_8$};
    \node[roundnode] (X4) [below=0.5cm of C2] {$X_4$};
    \node[squarenode] (C3) [below=0.5cm of X2] {$+$};
    \node[squarenode] (C4) [below=0.5cm of X4] {$+$};
    \node[roundnode] (X6) [left=0.5cm of C3] {$X_6$};
    \node[roundnode] (X3_1) [below=0.5cm of C3] {$X_3$};
    \node[roundnode] (X3_2) [below=0.5cm of C4] {$X_3$};
    \node[roundnode] (X7) [right=0.5cm of C4] {$X_7$};
    \draw (X1) to (C1);
    \draw (X1) to (C2);
    \draw (C1) to (X5);
    \draw (C1) to (X2);
    \draw (C2) to (X4);
    \draw (C2) to (X8);
    \draw (X2) to (C3);
    \draw (X4) to (C4);
    \draw (C3) to (X6);
    \draw (C3) to (X3_1);
    \draw (C4) to (X3_2);
    \draw (C4) to (X7);
  \end{tikzpicture}
  \end{adjustbox}
  \caption{}
  \label{fig:8bitcode_h2}
  \end{subfigure}
  \hspace{2mm}
  \begin{subfigure}[b]{0.22\textwidth}
  \begin{adjustbox}{width=\textwidth}
    \tikzsetnextfilename{8bitComputationalGraph}
  \begin{tikzpicture}
    \node[roundnode] (X1) {$X_1$};
    \node[squarenode] (C1) [below left=0.3cm and 0.3cm of X1] {$+$};
    \node[squarenode] (C2) [below right=0.3cm and 0.3cm of X1] {$+$};
    \node[roundnode] (X5) [left=0.5cm of C1] {$X_5$};
    \node[roundnode] (X2_1) [below=0.5cm of C1] {$X_2$};
    \node[roundnode] (X8) [right=0.5cm of C2] {$X_8$};
    \node[roundnode] (X4_1) [below=0.5cm of C2] {$X_4$};
    \node[squarenode] (C3) [below=0.5cm of X2_1] {$+$};
    \node[squarenode] (C4) [below=0.5cm of X4_1] {$+$};
    \node[roundnode] (X6_1) [left=0.5cm of C3] {$X_6$};
    \node[roundnode] (X3_1) [below=0.5cm of C3] {$X_3$};
    \node[roundnode] (X3_2) [below=0.5cm of C4] {$X_3$};
    \node[roundnode] (X7_1) [right=0.5cm of C4] {$X_7$};
    \node[squarenode] (C5) [below=0.5cm of X3_1] {$+$};
    \node[squarenode] (C6) [below=0.5cm of X3_2] {$+$};
    \node[roundnode] (X7_2) [left=0.5cm of C5] {$X_7$};
    \node[roundnode] (X4_2) [below=0.5cm of C5] {$X_4$};
    \node[roundnode] (X2_2) [below=0.5cm of C6] {$X_2$};
    \node[roundnode] (X6_2) [right=0.5cm of C6] {$X_6$};
    \draw (X1) to (C1);
    \draw (X1) to (C2);
    \draw (C1) to (X5);
    \draw (C1) to (X2_1);
    \draw (C2) to (X4_1);
    \draw (C2) to (X8);
    \draw (X2_1) to (C3);
    \draw (X4_1) to (C4);
    \draw (C3) to (X6_1);
    \draw (C3) to (X3_1);
    \draw (C4) to (X3_2);
    \draw (C4) to (X7_1);
    \draw (X3_1) to (C5);
    \draw (X3_2) to (C6);
    \draw (C5) to (X7_2);
    \draw (C5) to (X4_2);
    \draw (C6) to (X2_2);
    \draw (C6) to (X6_2);
  \end{tikzpicture}
  \end{adjustbox}
  \caption{}
  \label{fig:8bitcode_h3}
  \end{subfigure}

  \caption{Tanner graph of the (8,4) code $\mathcal C$ and associated $X_1$ computation trees for $h=1,2,3$.}
  \label{fig:8bitcode_computationalgraph}
\end{figure}
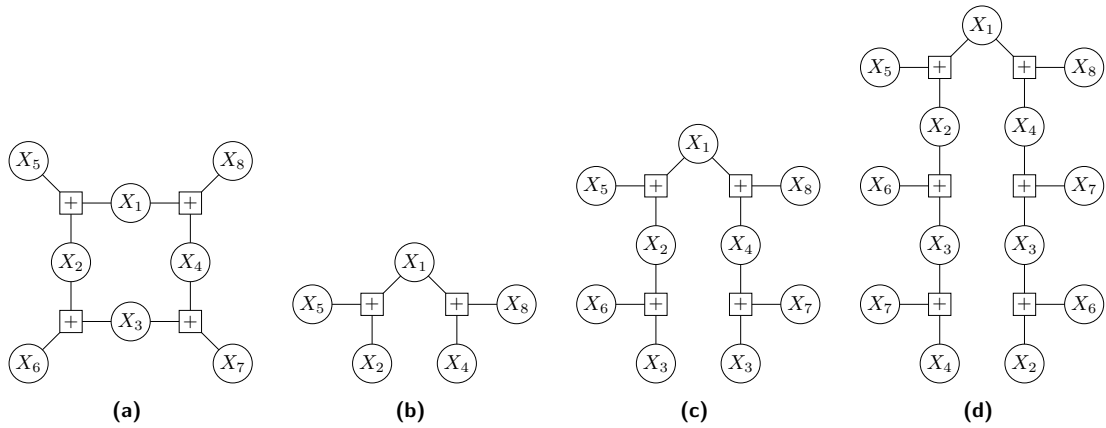

But it can also be interpreted in a different manner: As discussed in the introduction and depicted in \cref{fig:unrollingcomputational}, by associating a unique variable to each node we can consider the unrolled depth-$h$ graph to be the Tanner graph of an ($n',k'$) code, which we denote by $\mathcal{C}'$. 
In this sense, performing classical BP on the original non-tree factor graph of $\mathcal{C}$ for $h$ time steps is equivalent to performing classical BP on the tree factor graph of $\mathcal{C}'$. 
Note that for this purpose, certain bits of the noisy channel output must be duplicated, since they appear at multiple locations. 
For example, in the $h=3$ case of the 8-bit code, the variables $X_2,X_3,X_4,X_6,X_7$ all appear twice in the computation tree, so the corresponding channel outputs must be appropriately duplicated before decoding.

We translate this approach to the setting of BPQM by making use of approximate cloning both in the transformation of the actual channel output to something suitable for a $\mathcal C'$ decoder and in fixing the channel parameter inputs, i.e. angles, to the BPQM algorithm.
Let us now formalize the ideas outlined above. 
Consider the task of decoding the $n$-qubit channel output $\ket{\Psi_{\vec{x}}} := \qstate{x_1}{\theta_1}\otimes\dots\otimes\qstate{x_n}{\theta_n}$ for input $\vec x$ a codeword of an $(n,k)$ binary linear code $\mathcal{C}$. 
This can be done by choosing some new $(n',k')$ binary linear code $\mathcal{C}'$ as well as a map $\xi:\{1,\dots,n'\}\rightarrow \{1,\dots,n\}$ such that 
\begin{align}\begin{split}
  \lambda : \,\,&\mathbb{F}_2^n \rightarrow \mathbb{F}_2^{n'} \\
  &(x_1,\dots,x_n) \mapsto (x_{\xi (1)},\dots,x_{\xi (n')})
\end{split}\end{align}
maps codewords in $\mathcal C$ to codewords in $\mathcal C'$ by re-arranging, duplicating, and possibly even deleting codeword bits. 
%$\xi$ must be chosen such that $\lambda$ maps codewords in $\mathcal{C}$ to codewords in $\mathcal{C}'$. 
The map $\xi$ partially specifies how to transform the $n$-qubit channel output $\ket{\psi_{\vec{x}}}$ to an $n'$-qubit state $\rho'$. %, which we will treat as a channel output of the code $\mathcal{C}'$. 
Define $n_i:=|\{i'\in\{1,\dots,n'\} | \xi(i') = i\}|$ for $i\in\{1,\dots,n\}$. 
If $n_j=1$, then there exists exactly one $j'\in\{1,\dots,n'\}$ such that $\xi(j')=j$ and the $j'$th qubit of $\rho'$ is simply the $j$-th qubit of $\ket{\psi}$. 
If $n_j=m$ for some $m>1$, then there exist $m$ distinct indices $j_1',\dots,j_m'\in \{1,\dots,n'\}$ such that $\xi(j_l')=j$.
Then the $j$-th qubit of $\ket{\psi}$ has to be approximately cloned by an appropriate operation, resulting in $n_j$ approximate clones. 
The qubits $j_1',\dots,j_m'$ of $\rho'$ are then taken to be these clones. 
Choosing the cloning operation then fully specifies the transformation of $\ket{\psi_{\vec x}}$ to $\rho'_{\vec x}$. 
Finally, it remains to specify a choice of angle parameters $\theta'_j$ for $j=1\dots n'$, for the BPQM algorithm.\footnote{Recall that the input of the BPQM algorithm is not only a set of qubits $Y_1,\dots,Y_{n'}$ but also a set of angles $\theta'_1,\dots,\theta'_n$ that represent channel parameters.}
$\mathcal{C}'$ and $\xi$ are fully specified from the Tanner graph by the computational unrolling strategy for a given number of time steps $h$. 
There is however still freedom in the choice of approximate cloning operation as well as the channel parameters $\theta'_j$.

A very natural (but not necessarily optimal) choice of approximate cloner for $n_j=2$ is to take the adjoint equality node unitary $U_{\ostar}(\theta',\theta')^{\dagger}$ for $\theta' = \arccos\sqrt{\cos\theta}$, which maps the state $\ket{\pm \theta}\ket{0}$ to $\ket{\pm \theta'}\ket{\pm \theta'}$. 
We denote this strategy as \emph{ENU cloner} where ENU stands for equality node unitary. 
The main advantage of the ENU cloner is that it produces a product state with two qubits with known angle $\theta'$. 
It is therefore naturally clear how to choose the associated ${\theta'_j}$.
Note that the ENU cloner can also be easily generalized for $n_j\geq 3$ by mapping the state $\ket{\pm\theta}\ket{0}^{\otimes n}$ to $\ket{\pm \theta'}^{\otimes n}$ where $\theta' = \arccos(\cos(\theta)^{1/n_j})$.
One can quickly check, that this generalized cloner can be realized by a sequence of $n_j$ two-qubit unitaries $U_{\ostar}$ with appropriate angles. 

In~\cref{sec:optimal_cloner} we will consider a more refined choice for an approximate cloner.

Just as in standard BPQM, it is necessary to reverse all the decoding operations for a given codeword bit, including any cloning operations but excluding the final measurement, before proceeding to decode the next codeword bit. 
Note that this reversibility of the cloning does not imply that the cloning process must necessarily be unitary. 
By implementing a general approximate cloner through its Stinespring dilation~\cite{watrous_theory_2018}, we can revert its action at a later point by keeping track of the produced ancilla states.

\subsection{Numerical results}\label{sec:cloning_numerics}
We now present  numerical results on decoding the 8-bit code introduced above using BPQM. 
We generate the tree factor graph used for decoding by computationally unrolling around the target codeword bit for $h$ rounds. 
In cases where the unrolling is deep enough such that certain variables appear multiple times, we use the ENU to approximately clone the corresponding qubits. 

\Cref{fig:8bitcode_numerics_x1,fig:8bitcode_numerics_x5} depict the probability of successful decoding of the bits $X_1$ and $X_5$, respectively, as a function of the channel parameter $\theta$, identical for all channels. \
\Cref{fig:8bitcode_numerics_codeword} depicts the probability of successfully decoding the complete codeword. 
Included for comparison are the optimal classical bit-MAP and block-MAP (codeword) decoders, applied after performing the optimal Helstrom measurement for the CQ channel on each quantum channel output individually. 
The performance of the optimal quantum decoders can be obtained using the formulation as a semidefinite program (see e.g. \cite[\S3.1]{watrous_theory_2018}). 
In the bit-MAP case it is the Helstrom measurement of the corresponding channel from input codeword bit to quantum outputs, whereas in the codeword case it is realized by the PGM.

\begin{figure}
  \centering
  \begin{subfigure}[b]{0.32\textwidth}
     \includegraphics[width=1.0\textwidth]{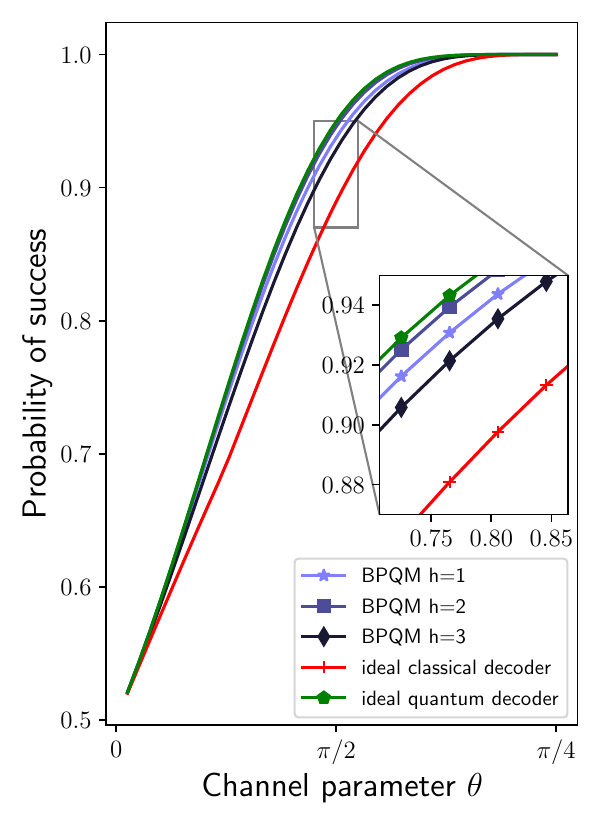}
     \caption{}
     \label{fig:8bitcode_numerics_x1}
  \end{subfigure}
  \begin{subfigure}[b]{0.32\textwidth}
     \includegraphics[width=1.0\textwidth]{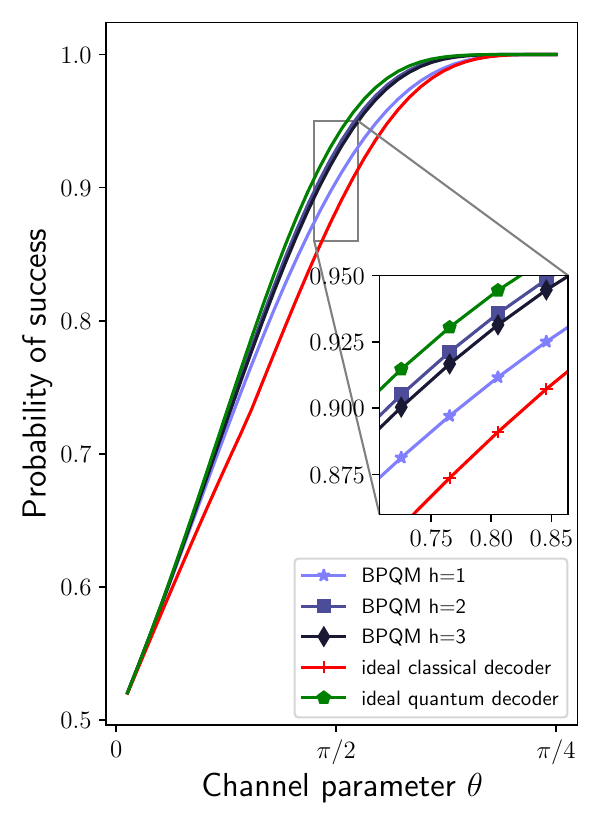}
     \caption{}
     \label{fig:8bitcode_numerics_x5}
  \end{subfigure}
  \begin{subfigure}[b]{0.32\textwidth}
     \includegraphics[width=1.0\textwidth]{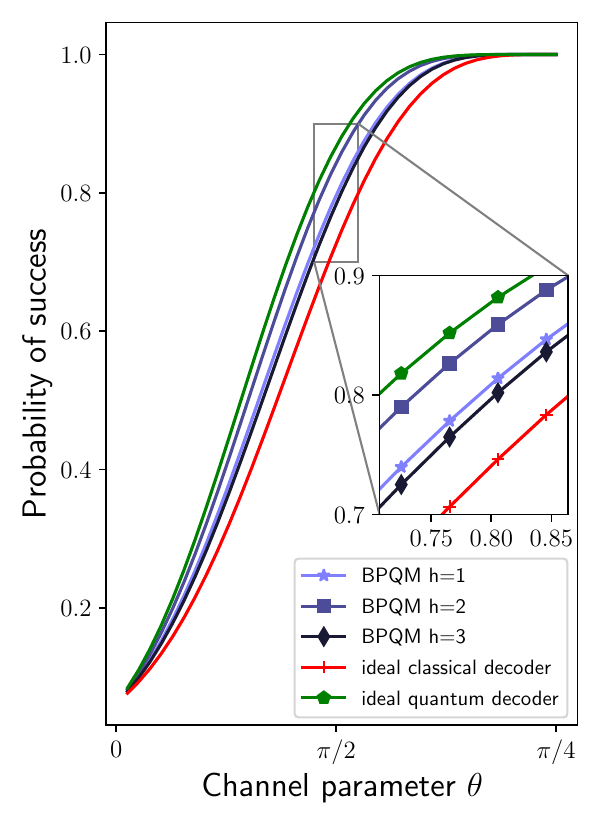}
     \caption{}
     \label{fig:8bitcode_numerics_codeword}
  \end{subfigure}
  \caption{Numerical results from decoding $X_1$ (a), $X_5$ (b) and the complete codeword (c) in the $8$-bit code depicted in~\cref{fig:8bitcode_tanner_graph}.}
  \label{fig:8bitcode_numerics}
\end{figure}

The numerics reveal that our approach can significantly outperform the best possible classical decoder in any of the analyzed cases. Furthermore we see that it is not beneficial to make the number of rounds $h$ as large as possible---at some point the benefits of considering more rounds is smaller than the drawback incurred by having to do more approximate cloning. 
This is a strongly different characteristic compared to the classical belief propagation, where  increasing the number of rounds generally improves the quality of the result.

One might wonder whether cloning is really necessary to achieve the best possible performance. 
For instance, it may be possible to simply employ BPQM defined for a spanning tree of the original Tanner graph. 
In~\cref{app:alternative_decoders} we investigate such alternative decoders and argue that it is unlikely that any such decoder can reach the same performance as the ENU+BPQM decoder.

\subsection{Optimal approximate cloner}\label{sec:optimal_cloner}
While the ENU cloner is conceptually very simple, there is a priori no reason why it should be the best choice of an approximate cloner. 
In this section we restrict ourselves to the simplest case in which the cloning operation produces just two approximate copies of the input. 
Bruß \emph{et al.}~\cite{bruss_1998} characterized the operation which maps $\qstate{0}{\theta}\ket{0}$ and $\qstate{1}{\theta}\ket{0}$ to some two-qubit states $\ket{\phi_0}$ and $ \ket{\phi_1}$, respectively, which maximizes the average global (squared) fidelity 
\begin{equation}
  \tfrac{1}{2}\left( |\qstateconj{0}{\theta}\qstateconj{0}{\theta}\cdot \ket{\phi_0}|^2 + |\qstateconj{1}{\theta}\qstateconj{1}{\theta}\cdot \ket{\phi_1}|^2  \right) \, .
\end{equation}
It turns out that the optimal cloner unitarily embeds the two input states into the span of the ideal cloned states. 
A particularly simple implementation is actually just given by $U_{\ostar}(\theta,\theta)^\dagger$, i.e.\ the ENU cloner for the ``wrong'' angle.  
% \begin{equation}
% \label{eq:optimalcloningU}
% U^{\text{opt}}=\frac1{\sqrt2}\begin{pmatrix}
% %a & 0 & 0 & -b\\ 0 & 1 & 1 & 0\\ 0 & -1 & 1 & 0\\ b & 0 & 0 & a
% a_+ & -a_- & 0 & 0\\ 0 & 0 & 1 & 1\\ 0 & 0 & 1 & -1\\ a_- & a_+ & 0 & 0
% \end{pmatrix}\,,
% \end{equation}
% for $a_\pm=\frac{1\pm f}{\sqrt{1+f^2}}$ and $f$ the fidelity (overlap) of the input states (compare with the transpose of $U_{\ostar}$ given in \cref{eq:equalitynodeU}). 
% For completeness, we give a quantum circuit decomposition of $U^{\text{opt}}$ in \Cref{app:ustar_identity}.

% In fact they show that this optimal cloner is realized by the unitary operation $(R_y(-\pi /2)\otimes R_y(-\pi /2))\cdot U\cdot (R_y(\pi/2)\otimes \id)$ where the unitary $U$ acts as
% \begin{align*}
%   U\ket{00} &= a\ket{00} + b\ket{01} + b\ket{10} + c\ket{11}\,,\\
%   U\ket{10} &= c\ket{00} + b\ket{01} + b\ket{10} + a\ket{11}\,,
% \end{align*}
% for
% \begin{align*}
%   a &= \frac{1}{2\cos 2\tilde{\theta}}\left[ \cos\tilde\theta \left( P + Q\cos 2\tilde{\theta} \right) - \sin\tilde{\theta}(P-Q\cos 2\tilde{\theta}) \right]\\
%   b &= \frac{1}{2\cos 2\tilde{\theta}} P \sin 2\tilde{\theta} (\cos\tilde{\theta} - \sin\tilde{\theta})\\
%   c &= \frac{1}{2\cos 2\tilde{\theta}}\left[ \cos\tilde\theta \left( P - Q\cos 2\tilde{\theta} \right) - \sin\tilde{\theta}(P+Q\cos 2\tilde{\theta}) \right]\\
%   P &= \frac{1}{2}\sqrt{\frac{1+\sin 2\tilde{\theta}}{1+\sin^2 2\tilde{\theta}}} \, ,\\
%   Q &= \frac{1}{2}\sqrt{\frac{1-\sin 2\tilde{\theta}}{1-\sin^2 2\tilde{\theta}}} \, ,\\
%   \tilde{\theta} &= \frac{\pi}{4} - \frac{1}{2}\theta\,.
% \end{align*}

When using this unitary to clone the channel outputs, it is not immediately clear what angle parameter to assign to the associated channels in the BPQM algorithm. 
We address this problem numerically. 
Consider decoding $X_1$ in the 8-bit code using $h=3$ for a fixed channel parameter $\theta$. 
Several of the qubits have to be cloned to produce two approximate clones. 
Numerical simulation of the decoding success probability is shown in~\cref{fig:optimal_cloner} for a range of channel parameters $\theta'$ used in the BPQM algorithm for the approximate clones. 
Interestingly, the optimal decoder can outperform the ENU cloner if $\theta'$ is chosen suitably. Furthermore, it would seem that the optimal value of $\theta'$ angle is very close to the value of $\arccos\sqrt{\cos \theta}$, the value used by the ENU cloner.

\begin{figure}
  \centering
  \includegraphics[width=0.8\textwidth]{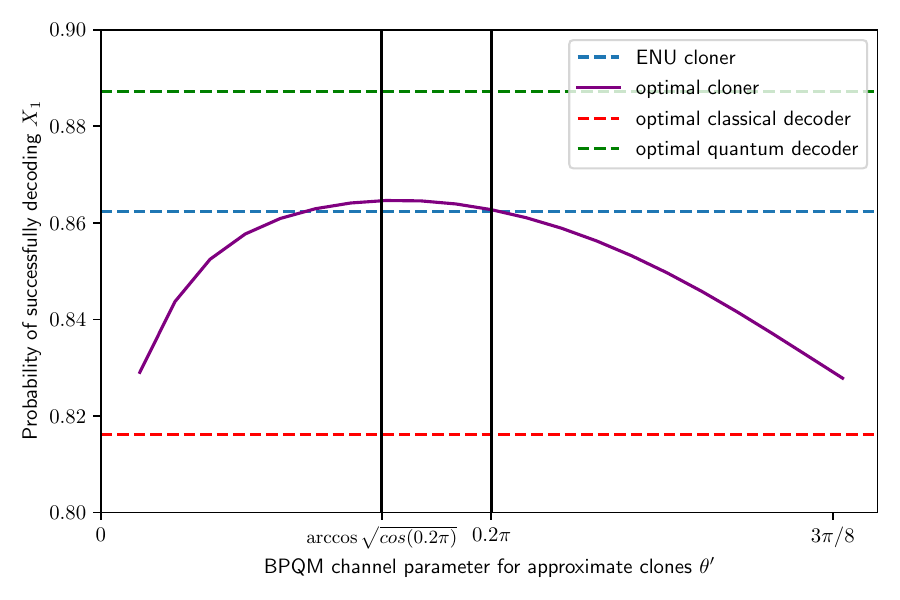}
  \caption{Decoding performance of optimal cloner on the 8-bit code with number of rounds $h=3$. The angle $\theta'$ is varied. The channel parameter is fixed at $\theta=0.2\pi$.}
  \label{fig:optimal_cloner}
\end{figure}

\section{Future directions}\label{sec:open_questions}
In this section we explore remaining open questions as well as possible avenues for future research directions related to the BPQM algorithm.
We should note that our focus here is on theoretical aspects of the algorithm, and we do not address the also interesting issues of hardware implementations. 
For more discussion of such issues, including the prospects for surpassing the performance of decoders which first measure the output qubits individually, see~\cite{rengaswamy_2020}.

\subsection{Potential improvements and extensions}
An immediate question is how well BPQM fares when decoding LDPC codes, the usual application of BP decoding in the classical case. 
It might even be possible to show that LDPC codes paired with BPQM decoding can achieve the Holevo capacity, particularly the ``spatially-coupled'' variant which achieves the classical capacity~\cite{kudekar_2013}. 
A central question to be answered for this is how much the approximate cloning deteriorates the quality of the decoder for these families of codes.

Another aspect to be analyzed when choosing appropriate codes is the speed of the decoder. 
Assuming a fixed accuracy $B$, we have shown in~\cref{sec:complexity} that the decoding circuit depth scales as $\mathcal{O}(kn)$ at worst for an arbitrary $(n,k)$ binary linear code. 
However, if the topology of the MPG allows us to parallelize some of the node operations, the circuit depth could be reduced up to $\mathcal{O}(k\log n)$ in the best case. 
It would therefore be interesting to study which families of codes would allow such decoding in quasilinear time.

There is also another potential optimization to reduce the quantum circuit size, which we haven't explored in this work.
In some cases it might be possible that certain operations in subsequent circuit parts cancel out.
For instance, the final node operations in \smash{$V_j^{\dagger}$} might cancel with the first node operations in \smash{$V_{j+1}$} for some $j\in\{1,\dots,k-1\}$.
To make use of this optimization, the order in which we decode the codeword bits suddenly becomes relevant.
It would be interesting to study in future work what families of codes could benefit most from this optimization.

It should also be noted that the analysis of the discretization errors in this work includes many crude bounds. It cannot be excluded that a more refined analysis and/or some tweaks to the message-passing algorithm itself might exhibit a more favorable scaling, for instance the required $B$ to reach a fixed accuracy $\epsilon$ might scale sublinearly in $n$.

Generally, given some code $\mathcal{C}$ with non-tree Tanner graph, there is considerable freedom in how exactly one decodes that code with BPQM. 
As explained in~\cref{sec:nontree_graphs}, one has to make some choice of $(n',k')$ code $\mathcal C'$ with tree Tanner graph, a map $\xi:\{1,\dots,n'\}\rightarrow\{1,\dots,n\}$ that determines how to map codewords from $\mathcal{C}$ to $\mathcal{C'}$, some choice of channel parameters $\theta'_j$ as well as an approximate cloning procedure. 
It is not clear what the best choice is, and in fact the optimal strategy is likely to vary strongly depending on the considered code. 
For instance, computational unrolling is unlikely to yield good results when the Tanner graph contains many small cycles. 

Furthermore, it should be noted that our framework is unlikely to be the most general approach to deal with Tanner graphs that contain cycles.
The approximate cloning procedure is executed purely before the message-passing on the MPG generated by $\mathcal{C}'$.
However, it is plausible that the approximate cloning could also be performed interleaved with the message-passing operations, e.g., the output of some node operation could be cloned as to be passed to two successors insted of just one.
In this case, our picture with the MPG would not be applicable anymore, since the data flow caould not be described by a tree.

Finally, it would be very interesting to extend BPQM to more general settings. 
One possibility would be to study general (non-pure) CQ channels. 
In that case, we expect BPQM to generally lose its block-optimality, since the noisy channel outputs generally do not form a geometrically uniform set (which was a crucial ingredient in~\cref{sec:optimality}). 
This would not be too surprising, as classical BP generally also does not perform block-optimal decoding, but only bitwise optimal decoding (assuming a code with tree Tanner graph).\footnote{Note that in the classical setting, the block-optimal solution could instead be obtained using the max-product or min-sum algorithm.}
As a further step, it would also be interesting to see if BPQM could even be generalized to decode quantum information transmitted over some form of a quantum channel.

\subsection{MPGs not from the Tanner graph}
The BPQM decoding algorithm is constructed using a tree MPG describing the codewords of the code, which we have shown can be constructed from a tree Tanner graph. 
However, it is not necessary to have a tree Tanner graph to have a tree MPG representation. 
For instance, the (8,5) code associated with the 8 channel leaves $\theta_2$ through $\theta_9$ in \cref{fig:17bitcode_mpg} has three parity-checks, leading to 168 possible parity-check matrices, none of which has a tree Tanner graph. 

Moreover, MPGs need not come from the Tanner graph. 
An important instance comes from successive cancellation decoding of Reed-Muller or polar codes~\cite{arikan_channel_2009}, as mentioned in \cite{renes_2017}. 
Let us describe how to apply BPQM in this case more precisely. 
Successive cancellation decoding proceeds by sequential decoding the individual bits of the message (not codeword) as input to the encoding circuit. 
For each message bit, the encoding circuit itself defines a ``synthesized'' channel from that bit to all outputs of the noisy channel, as well as all preceding message inputs. 
Again the goal is to implement the Helstrom measurement to distinguish the two possible channel outputs. 
In the case of Reed-Muller and polar codes, the synthetic channels have a perfect binary tree structure consisting of repeated channel convolution of two different types, which we might call ``better'' and ``worse''. 
These are very closely related to the check and variable node channel convolutions depicted in \cref{eq:checknodeffg,eq:equalitynodeFFG}. 

Indeed, the worse channel convolution is precisely the check convolution of \cref{eq:checknodeffg}. 
The better convolution is ultimately equivalent to the equality convolution for symmetric channels such as the pure state output channel under consideration here.
The better channel convolution of two channels $W_1$ and $W_2$ is $W'(x)=\sum_{y\in \mathbb Z_2} \ketbra{y}\otimes W_1(x)\otimes W_2(x\oplus y)$. 
When $W_2$ is symmetric in the sense that $W_2(x\oplus y)=U_y W_2(x) U_y^\dagger$ for an appropriate unitary $U_y$, then a control-$U_y$ operation from the classical output $\ketbra y$ transforms the output to $\tfrac12\id\otimes [W_1\ostar W_2](x)$, the equality node convolution. 

Due to the structure of the encoding circuit for Reed-Muller and polar codes, the additional classical outputs $\ketbra y$ arising from the worse convolution are precisely the preceding input message bits of the synthesized channels. 
The output of the synthesized channels corresponding to the preceding already-decoded input messages can therefore be decoupled from the outputs of the noisy channel, at which point the decoding task is precisely to perform the Helstrom measurement to distinguish the two outputs of a channel described by a perfect tree MPG. Hence BPQM can be used for this purpose.  

For a general CQ channel, an efficient implementation of the successive cancellation decoder is not known~\cite{wilde_polar_2013,renes_polar_2014}. 
Guha and Wilde examine the pure state case considered here in \cite{guha_polar_2012}, but also do not construct an efficient decoder. 
In \cite{renes_2017} it is claimed that BPQM leads to an efficient decoding algorithm for this case, but as we have seen above, this is more correctly a statement about the message-passing approximation of BPQM instead. 
Nevertheless, since the synthetic channels ``polarize'', i.e., become reliable in the large blocklength limit, and BPQM implements the optimal Helstrom measurement, the argument about the total error in decoding the entire message presented in \cite{renes_2017} goes through, ensuring that the resulting BPQM successive cancellation decoder achieves capacity for the pure state channel.

In light of the block optimality of BPQM, it would be interesting to investigate whether BPQM gives an optimal decoder for polar or Reed-Muller codes and pure state channels. 
However, one major difference to the MPGs generated from the Tanner graph is that each message bit requires a completely different MPG, not one obtained by moving the degree-two root node to the appropriate location on a fixed graph. 
Another case worth examination is if BPQM can be applied to decoding of turbo codes using the associated state FFG~\cite[Chapter 6]{richardson_2008}. 

\bigskip
\noindent{\bfseries Acknowledgments}\\
This work was supported by the Swiss National Science Foundation, through the National Center of Competence in Research ``Quantum Science and Technology'' (QSIT) and through grant number CRSII5\_186364.
The authors are grateful to Narayanan Rengaswamy, Henry D.\ Pfister, and Andru Gheorghiu for useful discussions.

\appendix

\section{Example of factor graph contraction on 5-bit code}\label{app:step_by_step}
In this section we illustrate step-by-step the contraction process described in~\cref{fig:5bitcode_contraction} on the 5-bit code to decode the codeword bit $X_1$.
The red dashed box indicates where the contraction takes place at each intermediate step.
\begin{center}
\tikzsetnextfilename{ForneyExample1}
\begin{tikzpicture}
\pgfdeclarelayer{bg}
\pgfsetlayers{bg,main}

\node[squarenode] (eq1) {\cequal};
\node[squarenode] (eq2) [above right=0.2cm and 0.3cm of eq1] {\cequal};
\node[squarenode] (ch1) [above right=0.05cm and 0.3cm of eq2] {$+$};
\node[squarenode] (ch2) [below right=0.05cm and 0.3cm of eq2] {$+$};
\node[quadnode] (W5) [above right=-0.1cm and 0.3cm of ch1] {$W[\theta_5]$};
\node[quadnode] (W3) [below=0.15cm of W5] {$W[\theta_3]$};
\node[quadnode] (W4) [below=0.15cm of W3] {$W[\theta_4]$};
\node[quadnode] (W2) [below=0.15cm of W4] {$W[\theta_2]$};
\node[quadnode] (W1) [below=0.15cm of W2] {$W[\theta_1]$};
\node[quadnode, minimum width=0.5cm, minimum height=1.0cm] (CNOT1) [right=1.45cm of ch1] {};
\node[phase] (CNOT1control) [right=0.29cm of W5] {};
\node[circlewc] (CNOT1target) [below=0.33cm of CNOT1control] {};
\node[quadnode, minimum width=0.5cm, minimum height=1.0cm] (CNOT2) [below=0.2cm of CNOT1] {};
\node[phase] (CNOT2control) [right=0.29cm of W4] {};
\node[circlewc] (CNOT2target) [below=0.33cm of CNOT2control] {};
\node[quadnode] (U1) [right=2.4cm of ch1, minimum height=1.0cm] {\scalebox{0.7}{$U_{\ostar}(\theta_5\boxstar_i\theta_3,\theta_4\boxstar_j\theta_2)$}};
\node[quadnode] (U2) [right=0.4cm of U1, minimum height=1.0cm] {\scalebox{0.7}{$U_{\ostar}((\theta_5\boxstar_i\theta_3)\ostar(\theta_4\boxstar_j\theta_2),\theta_1)$}};
\node[squarenode] (control1i) [right=1.5cm of W4] {$\bullet$};
\node[squarenode] (control1j) [right=2.5cm of W2] {$\bullet$};
\node[squarenode] (control2i) [right=4.5cm of W4] {$\bullet$};
\node[squarenode] (control2j) [right=5.5cm of W2] {$\bullet$};

\draw[thick] (CNOT1control.south) to (CNOT1target.north);
\draw[thick] (CNOT2control.south) to (CNOT2target.north);
\draw (eq1.east) to[out=0, in=180] (eq2.west);
\draw (eq2.east) to[out=0, in=180] (ch1.west);
\draw (eq2.east) to[out=0, in=180] (ch2.west);
\draw (ch1.east) to[out=0, in=180] (W5.west);
\draw (ch1.east) to[out=0, in=180] (W3.west);
\draw (ch2.east) to[out=0, in=180] (W4.west);
\draw (ch2.east) to[out=0, in=180] (W2.west);
\draw (eq1.east) to[out=0, in=180] (W1.west);
\draw (eq1.west) -- node[above] {$X_1$} ++(-0.5cm,0);

\draw[draw=red,dashed] (ch1.south west)+(-0.1,-0.4) rectangle ++(2.4,0.8);

% We have to put the double dashed lines in a background layer, otherwise they look ugly
\begin{pgfonlayer}{bg}
  % note: distance between lanes is 0.6cm
  \draw [style=double] (W5.east) to ++(7.8cm,0) node[above] {$Q_1$};
  \draw [style=double] (W3.east) to ++(0.8cm,0) to ++(0.2cm,-0.6cm) to ++(6.8cm,0) node[above] {$Q_3$};
  \draw [style=double] (W4.east) to ++(0.8cm,0) to ++(0.2cm,0.6cm) to ++(2.55cm,0) to ++(0.3cm,-1.8cm) to ++(3.95cm,0) node[above] {$Q_5$};
  \draw [style=double] (W2.east) to ++(7.8cm,0) node[above] {$Q_4$};
  \draw [style=double] (W1.east) to ++(3.55cm,0) to ++(0.3cm,1.8cm) to ++(3.95cm,0) node[above] {$Q_2$};
  \draw [style=double] (control1i.north) to node[below right=-0.15cm and 0.1cm] {$I$} ++(0,0.3cm);
  \draw [style=double] (control1j.north) to node[below right=-0.4cm and 0.0cm] {$J$} ++(0,1.0cm);
  \draw [style=double] (control2i.north) to node[below right=-0.15cm and 0.1cm] {$I$} ++(0,0.3cm);
  \draw [style=double] (control2j.north) to node[below right=-0.4cm and 0.0cm] {$J$} ++(0,1.0cm);
\end{pgfonlayer}
\end{tikzpicture}
\end{center}

\begin{center}
\tikzsetnextfilename{ForneyExample2}
\begin{tikzpicture}
\pgfdeclarelayer{bg}
\pgfsetlayers{bg,main}

\node[squarenode] (eq1) {\cequal};
\node (equiv) [left=1cm of eq1] {$\cong$};
\node[squarenode] (eq2) [above right=0.2cm and 0.3cm of eq1] {\cequal};
\node[squarenode,draw=none,text opacity=0] (ch1) [above right=0.05cm and 0.3cm of eq2] {$+$};
\node[squarenode] (ch2) [below right=0.05cm and 0.3cm of eq2] {$+$};
\node[quadnode] (W5) [above right=-0.1cm and 0.3cm of ch1] {$W[\theta_5\boxstar_i\theta_3]$};
\node[quadnode,draw=none,opacity=0,text opacity=0] (W3) [below=0.15cm of W5] {$W[\theta_3]$};
\node[quadnode] (W4) [below left=0.15cm and -0.5cm of W3] {$W[\theta_4]$};
\node[quadnode] (W2) [below=0.15cm of W4] {$W[\theta_2]$};
\node[quadnode] (W1) [below=0.15cm of W2] {$W[\theta_1]$};
\node[quadnode, minimum width=0.5cm, minimum height=1.0cm] (CNOT2) [below right=-0.4cm and 0.2cm of W4] {};
\node[phase] (CNOT2control) [right=0.35cm of W4] {};
\node[circlewc] (CNOT2target) [below=0.33cm of CNOT2control] {};
\node[quadnode] (U1) [right=2.4cm of ch1, minimum height=1.0cm] {\scalebox{0.7}{$U_{\ostar}(\theta_5\boxstar_i\theta_3,\theta_4\boxstar_j\theta_2)$}};
\node[quadnode] (U2) [right=0.4cm of U1, minimum height=1.0cm] {\scalebox{0.7}{$U_{\ostar}((\theta_5\boxstar_i\theta_3)\ostar(\theta_4\boxstar_j\theta_2),\theta_1)$}};
\node[squarenode] (control1i) [right=1.5cm of W4] {$\bullet$};
\node[squarenode] (control1j) [right=2.5cm of W2] {$\bullet$};
\node[squarenode] (control2i) [right=4.5cm of W4] {$\bullet$};
\node[squarenode] (control2j) [right=5.5cm of W2] {$\bullet$};
\node[squarenode] (eqi) [below=0.2cm of W5] {\cequal};
\node[quadnode] (Pi) [left=0.2cm of eqi] {$P_{I}$};

\draw[thick] (CNOT2control.south) to (CNOT2target.north);
\draw (eq1.east) to[out=0, in=180] (eq2.west);
\draw (eq2.east) to[out=0, in=180] (ch2.west);
\draw (eq2.east) to[out=0, in=180] (W5.west);
\draw (ch2.east) to[out=0, in=180] (W4.west);
\draw (ch2.east) to[out=0, in=180] (W2.west);
\draw (eq1.east) to[out=0, in=180] (W1.west);
\draw (eq1.west) -- node[above] {$X_1$} ++(-0.5cm,0);
\draw (Pi) to (eqi);
\draw (eqi) to (W5);

\draw[draw=red,dashed] (ch2.south west)+(-0.1,-0.7) rectangle ++(2.4,0.5);

% We have to put the double dashed lines in a background layer, otherwise they look ugly
\begin{pgfonlayer}{bg}
  % note: distance between lanes is 0.6cm
  \draw [style=double] (W5.east) to ++(7.0cm,0) node[above] {$Q_1$};
  \draw [style=double] (eqi.east) to ++(0.6cm,0) to ++(0.2cm,-0.6cm) to ++(7.0cm,0) node[above] {$Q_3$};
  \draw [style=double] (W4.east) to ++(0.8cm,0) to ++(0.2cm,0.6cm) to ++(2.65cm,0) to ++(0.3cm,-1.8cm) to ++(4.05cm,0) node[above] {$Q_5$};
  \draw [style=double] (W2.east) to ++(8.0cm,0) node[above] {$Q_4$};
  \draw [style=double] (W1.east) to ++(3.65cm,0) to ++(0.3cm,1.8cm) to ++(4.05cm,0) node[above] {$Q_2$};
  \draw [style=double] (control1i.north) to node[below right=-0.15cm and 0.1cm] {$I$} ++(0,0.3cm);
  \draw [style=double] (control1j.north) to node[below right=-0.4cm and 0.0cm] {$J$} ++(0,1.0cm);
  \draw [style=double] (control2i.north) to node[below right=-0.15cm and 0.1cm] {$I$} ++(0,0.3cm);
  \draw [style=double] (control2j.north) to node[below right=-0.4cm and 0.0cm] {$J$} ++(0,1.0cm);
\end{pgfonlayer}
\end{tikzpicture}
\end{center}

\begin{center}
\tikzsetnextfilename{ForneyExample3}
\begin{tikzpicture}
\pgfdeclarelayer{bg}
\pgfsetlayers{bg,main}

\node[squarenode] (eq1) {\cequal};
\node (equiv) [left=1cm of eq1] {$\cong$};
\node[squarenode] (eq2) [above right=0.2cm and 0.3cm of eq1] {\cequal};
\node[squarenode,draw=none,text opacity=0] (ch1) [above right=0.05cm and 0.3cm of eq2] {$+$};
\node[squarenode,draw=none,text opacity=0] (ch2) [below right=0.05cm and 0.3cm of eq2] {$+$};
\node[quadnode] (W5) [above right=-0.1cm and 0.3cm of ch1] {$W[\theta_5\boxstar_i\theta_3]$};
\node[quadnode,draw=none,opacity=0,text opacity=0] (W3) [below=0.15cm of W5] {$W[\theta_5]$};
\node[quadnode] (W4) [below=0.755cm of W5] {$W[\theta_4\boxstar_j\theta_2]$};
\node[quadnode,draw=none,opacity=0,text opacity=0] (W2) [below=0.15cm of W4] {$W[\theta_2]$};
\node[quadnode] (W1) [below=0.15cm of W2] {$W[\theta_1]$};
\node[quadnode] (U1) [right=2.7cm of ch1, minimum height=1.0cm] {\scalebox{0.7}{$U_{\ostar}(\theta_5\boxstar_i\theta_3,\theta_4\boxstar_j\theta_2)$}};
\node[quadnode] (U2) [right=0.4cm of U1, minimum height=1.0cm] {\scalebox{0.7}{$U_{\ostar}((\theta_5\boxstar_i\theta_3)\ostar(\theta_4\boxstar_j\theta_2),\theta_1)$}};
\node[squarenode] (control1i) [right=1.0cm of W4] {$\bullet$};
\node[squarenode] (control1j) [right=2.0cm of W2] {$\bullet$};
\node[squarenode] (control2i) [right=4.0cm of W4] {$\bullet$};
\node[squarenode] (control2j) [right=5.0cm of W2] {$\bullet$};
\node[squarenode] (eqi) [below=0.2cm of W5] {\cequal};
\node[quadnode] (Pi) [left=0.2cm of eqi] {$P_{I}$};
\node[squarenode] (eqj) [below=0.2cm of W4] {\cequal};
\node[quadnode] (Pj) [left=0.2cm of eqj] {$P_{J}$};

\draw (eq1.east) to[out=0, in=180] (eq2.west);
\draw (eq2.east) to[out=0, in=180] (W4.west);
\draw (eq2.east) to[out=0, in=180] (W5.west);
\draw (eq1.east) to[out=0, in=180] (W1.west);
\draw (eq1.west) -- node[above] {$X_1$} ++(-0.5cm,0);
\draw (Pi) to (eqi);
\draw (eqi) to (W5);
\draw (Pj) to (eqj);
\draw (eqj) to (W4);

\draw[draw=red,dashed] (eq2.south west)+(-0.1,-1.2) rectangle ++(6.3,1.25);

% We have to put the double dashed lines in a background layer, otherwise they look ugly
\begin{pgfonlayer}{bg}
  % note: distance between lanes is 0.6cm
  \draw [style=double] (W5.east) to ++(7.3cm,0) node[above] {$Q_1$};
  \draw [style=double] (eqi.east) to ++(0.9cm,0) to ++(0.2cm,-0.6cm) to ++(7.0cm,0) node[above] {$Q_3$};
  \draw [style=double] (W4.east) to ++(0.1cm,0) to ++(0.2cm,0.6cm) to ++(2.65cm,0) to ++(0.3cm,-1.8cm) to ++(4.05cm,0) node[above] {$Q_5$};
  \draw [style=double] (eqj.east) to ++(8.1cm,0) node[above] {$Q_4$};
  \draw [style=double] (W1.east) to ++(3.4cm,0) to ++(0.3cm,1.8cm) to ++(4.05cm,0) node[above] {$Q_2$};
  \draw [style=double] (control1i.north) to node[below right=-0.15cm and 0.1cm] {$I$} ++(0,0.3cm);
  \draw [style=double] (control1j.north) to node[below right=-0.4cm and 0.0cm] {$J$} ++(0,1.0cm);
  \draw [style=double] (control2i.north) to node[below right=-0.15cm and 0.1cm] {$I$} ++(0,0.3cm);
  \draw [style=double] (control2j.north) to node[below right=-0.4cm and 0.0cm] {$J$} ++(0,1.0cm);
\end{pgfonlayer}
\end{tikzpicture}
\end{center}

\begin{center}
\tikzsetnextfilename{ForneyExample4}
\begin{tikzpicture}
  \pgfdeclarelayer{bg}
  \pgfsetlayers{bg,main}
    \node[squarenode] (eq1) {\cequal};
    \node (equiv) [left=1cm of eq1] {$\cong$};
    \node[quadnode] (W5) [above right=0.5cm and 1.0cm of eq1] {$W[(\theta_5\boxstar_i\theta_3)\ostar(\theta_4\boxstar_j\theta_2)]$};
    \node[squarenode] (eqirhs) [below left=0.5cm and 0.2cm of W5.south] {\cequal};
    \node[squarenode] (eqjrhs) [below right=1.0cm and 0.2cm of W5.south] {\cequal};
    \node[quadnode] (Pirhs) [left=0.4cm of eqirhs] {$P_{I}$};
    \node[quadnode] (Pjrhs) [left=1.15cm of eqjrhs] {$P_{J}$};
    \node[quadnode] (W1) [below=1.4cm of W5.south] {$W[\theta_1]$};
    \node[quadnode] (zero) [below=1.9cm of W5.south] {$\ket{0}\bra{0}$};
    \node[quadnode] (U2) [below right=-0.45cm and 0.4cm of W5, minimum height=0.8cm] {\scalebox{0.7}{$U_{\ostar}((\theta_5\boxstar_i\theta_3)\ostar(\theta_4\boxstar_j\theta_2),\theta_1)$}};
    \node[squarenode] (control2i) [right=4.0cm of Pirhs] {$\bullet$};
    \node[squarenode] (control2j) [right=5.0cm of Pjrhs] {$\bullet$};

    \draw (eq1.west) to node[above] {$X_1$} ++(-0.5cm,0);
    \draw (eq1.east) to[out=0, in=180] (W5.west);
    \draw (eq1.east) to[out=0, in=180] (W1.west);
    \draw (Pirhs) to (eqirhs);
    \draw (Pjrhs) to (eqjrhs);
    
    \draw[draw=red,dashed] (eq1.south west)+(-0.05,-1.0) rectangle ++(9.5,1.45);
    
    \begin{pgfonlayer}{bg}
      \draw (eqirhs) to node[right] {$I$} +(0,0.7cm);
      \draw (eqjrhs) to node[above right] {$J$} ++(0,1.2cm);
      \draw [style=double] (W5) to ++(6.3cm,0) node[above] {$Q_1$};
      \draw [style=double] (eqirhs) to ++(6.7cm,0) node[above] {$Q_3$};
      \draw [style=double] (eqjrhs) to ++(5.9cm,0) node[above] {$Q_4$};
      \draw [style=double] (W1) to ++(1.8cm,0) to ++(0.2cm,1.4cm) to ++(4.3cm,0) node[above] {$Q_2$};
      \draw [style=double] (zero) to ++(6.3cm,0) node[above] {$Q_5$};
      \draw [style=double] (control2i.north) to node[below right=-0.15cm and 0.1cm] {$I$} ++(0,0.3cm);
      \draw [style=double] (control2j.north) to node[below right=-0.3cm and 0.0cm] {$J$} ++(0,1.0cm);
    \end{pgfonlayer}

\end{tikzpicture}
\end{center}

\begin{center}
\tikzsetnextfilename{ForneyExample5}
\begin{tikzpicture}
  \pgfdeclarelayer{bg}
  \pgfsetlayers{bg,main}
    \node[quadnode] (Wrhs) {$W[(\theta_5\boxstar_i\theta_3)\ostar(\theta_4\boxstar_j\theta_2)\ostar\theta_1]$};
    \node[squarenode] (eqirhs) [below left=0.3cm and 0.2cm of Wrhs.south] {\cequal};
    \node[squarenode] (eqjrhs) [below right=0.8cm and 0.2cm of Wrhs.south] {\cequal};
    \node[quadnode] (Pirhs) [left=0.4cm of eqirhs] {$P_{I}$};
    \node[quadnode] (Pjrhs) [left=1.15cm of eqjrhs] {$P_{J}$};
    \node[quadnode] (zero1rhs) [below=1.3cm of Wrhs.south] {$\ket{0}\bra{0}$};
    \node[quadnode] (zero2rhs) [below=1.8cm of Wrhs.south] {$\ket{0}\bra{0}$};
    \node (equiv) [left=1cm of Pjrhs] {$\cong$};

    \draw (Wrhs.west) to node[above] {$X_1$} ++(-0.5cm,0);
    \draw (Pirhs) to (eqirhs);
    \draw (eqirhs) to node[right] {$I$} +(0,0.5cm);
    \draw (Pjrhs) to (eqjrhs);
    \draw (eqjrhs) to node[above right] {$J$} ++(0,1.0cm);
    \begin{pgfonlayer}{bg}
      \draw [style=double] (Wrhs) to ++(3.0cm,0) node[above] {$Q_1$};
      \draw [style=double] (eqirhs) to ++(3.4cm,0) node[above] {$Q_3$};
      \draw [style=double] (eqjrhs) to ++(2.6cm,0) node[above] {$Q_4$};
      \draw [style=double] (zero1rhs) to ++(3.0cm,0) node[above] {$Q_2$};
      \draw [style=double] (zero2rhs) to ++(3.0cm,0) node[above] {$Q_5$};
    \end{pgfonlayer}

\end{tikzpicture}
\end{center}

\section{Proof of\texorpdfstring{~\cref{lem:purestate_evolution}}{ Lemma 4.4}}
\label{app:proof_optimality}
We first show~\cref{eq:compressed_state}, by giving the form of the state at all the intermediate steps in the algorithm.  
Then we show~\cref{eq:orthogonality}. 

For this proof, we have to again consider the factor graph from which the MPG was originally derived.
Every edge $e$ in the factor graph corresponds to some binary random variable, which is uniquely determined when given the values of $X_1,\dots,X_k$.
Since the factor graph only contains equality and check nodes, the random variable in question must be of the form $g_{e,1}X_1 + g_{e,2}X_2 + \dots + g_{e,k}X_k$ for some $\vec{g}_e\in\mathbb{F}_2^k$.
Since the value $\vec{d}:=(x_1,\dots,x_k)$ of $(X_1,\dots,X_k)$ is given in the setting of~\cref{lem:purestate_evolution}, the value $z_e$ of the random variable associated to $e$ is known to be $z_e=g_{e_1}x_1+\dots+g_{e_k}x_k$.

Furthremore, to slightly simplify the notation for this proof, we extend the MPG by adding a half-edge to the root, with the direction pointing away from the root.
The final data qubit produced by the root node is then said to be passed over that half-edge.

\begin{lemma}
\label{lem:stronger44}
Consider an MPG $G$ for the binary linear code $\mathcal C$ with respect to the codeword bit $X_r$, for $r\in 1,\dots,k$.
Then the joint state of the qubit passed by BPQM over any edge $e$ together with all ancilla qubits produced by check and variable nodes preceding $e$ can be written as
\begin{align}\label{eq:purestate_evol_intermediate}
  \ket{\Xi(e)}=\sum\limits_{\vec{j}\in\{0,1\}^{k_e-1}} c_{e,\vec{j}}(\vec{d}) \sqrt{p_{e,\vec{j}}} \qstate{z_e}{\theta_{e,\vec{j}}}_{D_e} \otimes \ket{\vec{j}}_{A_e} \otimes \ket{0^{n_e-k_e}}_{Z_e}\,,
\end{align}
where $n_e$ and $k_e-1$ denote the number of channel and check nodes preceding $e$, respectively, $c_{e,\vec{j}}(\vec{d})$ are some maps $\mathbb{F}_2^k\rightarrow\{+1,-1\}$, and the probabilities and angles $\{(\theta_{e,\vec{j}},p_{e,\vec{j}}) | \vec{j}\in\{0,1\}^{k_e-1}\}$ correspond exactly to the probability and angle entries of the branch list of the node directly preceding $e$ in $G$.
Furthermore,  $D_e$ denotes the qubit system transmitted over the edge $e$, $A_e$ denotes all ancilla qubits produced by check nodes preceding $e$, and $Z_e$ denotes all ancilla qubits produced by equality nodes preceding $e$.
\end{lemma}

By choosing $e$ to be the half-edge that we added to the MPG,~\cref{lem:stronger44} directly implies \cref{eq:compressed_state}. 
\begin{proof}
The proof is by induction.
The claim holds immediately for edges $e$ connected to leaf nodes in $G$, since the qubit in question is simply in the state $\qstate{z_e}{\theta}$. 
Therefore we only need to show that if the claim holds for two input edges $e,e'$ to a check or equality node, then it also holds for the output edge $\hat e$.
% \begin{center}
% \begin{tikzpicture}
%   \pgfdeclarelayer{bg}
%   \pgfsetlayers{bg,main}

%   \node[quadnode] (eq) {$=$ or $+$};
%   \begin{pgfonlayer}{bg}
%   \begin{scope}[decoration={markings,mark=at position 0.5 with {\arrow{>}}}]
%     \draw[postaction={decorate}] (eq)+(-0.5cm,-0.5cm) to node[left=0.1cm] {$e_1$} ++(-0.1cm,-0.1cm);
%     \draw[postaction={decorate}] (eq)+(0.5cm,-0.5cm) to node[right=0.1cm] {$e_2$} ++(0.1cm,-0.1cm);
%     \draw[postaction={decorate}] (eq) to ++(0,0.6cm) node[right] {$e_3$};
%   \end{scope}
%   \end{pgfonlayer}
% \end{tikzpicture}
% \end{center}
Now we assume the states of edges $e$ and $e'$ incident to a check or equality node are $\ket{\Xi(e)}$ and $\ket{\Xi(e')}$, respectively. 
We deal with the two cases separately. 
% By the induction hypothesis, for $i=1,2$ the joint state of the qubit passed over the edge $e_i$ together with all ancilla qubits produced by preceding equality/check nodes can be written as
% \begin{equation}
%   \sum\limits_{\vec{j}_i\in\{0,1\}^{k_{e_i}-1}} c_{i,\vec{j}_i}(\vec{d}) \sqrt{p_{i,\vec{j}_i}} \qstate{z_{e_i}}{\theta_{i,\vec{j}_i}}_{D_{e_i}} \otimes \ket{\vec{j}_i}_{A_{e_i}} \otimes \ket{0^{n_{e_i}-k_{e_i}}}_{Z_{e_i}}\,.
% \end{equation}
% \begin{equation}
%   \sum\limits_{\vec{j}\in\{0,1\}^{k_{e_i}-1}} c_{\vec{j}_i}(\vec{d}) \sqrt{p_{i,\vec{j}_i}} \qstate{z_{e_i}}{\theta_{i,\vec{j}_i}}_{D_{e_i}} \otimes \ket{\vec{j}_i}_{A_{e_i}} \otimes \ket{0^{n_{e_i}-k_{e_i}}}_{Z_{e_i}}\,.
% \end{equation}
% We differentiate between the equality node and check node case.
\paragraph{Equality node}
Due to the equality constraint, one has $z_{e}=z_{e'}=:z$. 
As described in \cref{equalitynodeunitary}, BPQM applies the unitary
  \begin{equation}
    \sum\limits_{\vec{j},\vec{j}'} U_{\ostar}(\theta_{e,\vec{j}},\theta_{e',\vec{j}'})_{D_{e},D_{e'}}\otimes \ket{\vec{j}}\bra{\vec{j}}_{A_{e}} \otimes \ket{\vec{j}'}\bra{\vec{j}'}_{A_{e'}}
  \end{equation}
to $\ket{\Xi(e)}\otimes \ket{\Xi(e')}$. 
By~\eqref{eq:Ustar_identity}, the resulting state is
\begin{equation}
\label{eq:equalityphasesmultiply}
\sum\limits_{\vec{j},\vec{j}'} c_{e,\vec{j}}(\vec{d})c_{e',\vec{j}'}(\vec{d}) \sqrt{p_{e,\vec{j}}p_{e',\vec{j}'}} \qstate{z}{\theta_{e,\vec{j}}\ostar\theta_{e',\vec{j}'}}_{D_{\hat e}} \otimes \left( \ket{\vec{j}}\ket{\vec{j}'} \right)_{A_{\hat e}} \otimes \ket{0^{n_{\hat e}-k_{\hat e}}}_{Z_{\hat e}}\,,
\end{equation}
which is again of the form of~\eqref{eq:purestate_evol_intermediate} (because $c_{e,\vec{j}}\cdot c_{e',\vec{j}'}$ is again also a map $\mathbb{F}_2^k\rightarrow\{+1,-1\}$ and because the probabilites and angles are combined according to the same rules as in the construction of the branch list entries in~\cref{sec:single_bit_decoding}). 
Note that here we have $n_{\hat e}=n_{e}+n_{e'}$ and $k_{\hat e}=k_{e}+k_{e'}-1$.
  
 \paragraph{Check node} 
 BPQM applies the gate $\cnot_{D_{e}D_{e'}}$ to $\ket{\Xi(e)}\otimes \ket{\Xi(e')}$. 
 By~\cref{eq:purestate_checknode}, the resulting state is
\begin{align}
\label{eq:checkphasessomething}
\sum\limits_{\vec{j},\vec{j}',l\in\{0,1\}} c_{e,\vec{j}}(\vec{d})c_{e',\vec{j}'}(\vec{d})(-1)^{l\cdot z_{e'}} \sqrt{p_{e,\vec{j}}p_{e',\vec{j}'}}\sqrt{\frac{1+(-1)^l\cos\theta_{e,\vec{j}} \cos\theta_{e',\vec{j}'}}{2}} \qstate{z}{\theta_{e,\vec{j}}\boxstar_l\theta_{e',\vec{j}'}}_{D_{\hat e}} \nonumber\\
\otimes \left( \ket{\vec{j}_i}\ket{\vec{j}'}\ket{l} \right)_{A_{\hat e}} \otimes \ket{0^{n_{\hat e}-k_{\hat e}}}_{Z_{\hat e}}
\end{align}
  which is again of the form of~\cref{eq:purestate_evol_intermediate}, by the same argument as the equality node above. Note that here we have $n_{\hat e}=n_{e}+n_{e'}$ and $k_{\hat e}=k_{e}+k_{e'}$.
It remains to show \cref{eq:purestate_checknode}; it is directly implied by the following identites, which can be derived from standard trigonometric identities:
\begin{align}
\qstate{x}{\alpha\boxstar_i\beta} &=
\cos\left(\frac{\alpha\boxstar_i\beta}{2}\right)\ket{0} + (-1)^x\sin\left(\frac{\alpha\boxstar_i\beta}{2}\right)\ket{1}\\
&=\frac1{\sqrt2}\frac{1}{\sqrt{1+(-1)^i\cos\alpha\cos\beta}}\sum_{j=0}^1(-1)^{xj}\sqrt{1+(-1)^j\cos\alpha}\sqrt{1+(-1)^{i+j}\cos\beta}\ket{j}\,,
\end{align}
\begin{align}
\label{eq:phasesource}
2\,\cnot{}\qstate{y_1}{\alpha}&\qstate{y_2}{\beta}=\nonumber\\
&\sqrt{1+\cos\alpha}\sqrt{1+\cos\beta}\ket{00} + (-1)^{y_1+y_2}\sqrt{1-\cos\alpha}\sqrt{1-\cos\beta}\ket{10}  \nonumber\\
    &+(-1)^{y_2}\left(\sqrt{1+\cos\alpha}\sqrt{1-\cos\beta}\ket{01} + (-1)^{y_1+y_2}\sqrt{1-\cos\alpha}\sqrt{1+\cos\beta}\ket{11}\right)\,.
\end{align}
\end{proof}

  % \begin{align*}\begin{split}
  %   \qstate{x}{\alpha\boxstar_0\beta} = \cos\left(\frac{\alpha\boxstar_0\beta}{2}\right)\ket{0} + (-1)^x\sin\left(\frac{\alpha\boxstar_0\beta}{2}\right)\ket{1} \\
  %   = \frac{1}{\sqrt{2}}\frac{1}{\sqrt{1+\cos\alpha\cos\beta}}\left( \sqrt{1+\cos\alpha}\sqrt{1+\cos\beta}\ket{0} + (-1)^x \sqrt{1-\cos\alpha}\sqrt{1-\cos\beta}\ket{1} \right)
  % \end{split}\end{align*}
  % \begin{align*}\begin{split}
  %   \qstate{x}{\alpha\boxstar_1\beta} = \cos\left(\frac{\alpha\boxstar_1\beta}{2}\right)\ket{0} + (-1)^x\sin\left(\frac{\alpha\boxstar_1\beta}{2}\right)\ket{1} \\
  %   = \frac{1}{\sqrt{2}}\frac{1}{\sqrt{1-\cos\alpha\cos\beta}}\left( \sqrt{1+\cos\alpha}\sqrt{1-\cos\beta}\ket{0} + (-1)^x \sqrt{1-\cos\alpha}\sqrt{1+\cos\beta}\ket{1} \right)
  % \end{split}\end{align*}
  % \begin{align*}\begin{split}
  %   CNOT\qstate{y_1}{\alpha}\otimes\qstate{y_2}{\beta} = \\
  %   \cos\frac{\alpha}{2}\cos\frac{\beta}{2}\ket{00} + (-1)^{y_2}\cos\frac{\alpha}{2}\sin\frac{\beta}{2}\ket{01} + (-1)^{y_1+y_2}\sin\frac{\alpha}{2}\sin\frac{\beta}{2}\ket{10} + (-1)^{y_1}\sin\frac{\alpha}{2}\cos\frac{\beta}{2}\ket{11} = \\
  %   \frac{1}{2}\sqrt{1+\cos\alpha}\sqrt{1+\cos\beta}\ket{00} + (-1)^{y_2}\frac{1}{2}\sqrt{1+\cos\alpha}\sqrt{1-\cos\beta}\ket{01} + \\
  %   (-1)^{y_1+y_2}\frac{1}{2}\sqrt{1-\cos\alpha}\sqrt{1-\cos\beta}\ket{10} + (-1)^{y_1}\frac{1}{2}\sqrt{1-\cos\alpha}\sqrt{1+\cos\beta}\ket{11}
  % \end{split}\end{align*}

Now we turn to proving \eqref{eq:orthogonality}. 
Let $E$ be the subset of MPG edges defined as follows: $E$ contains the half-edge at the root. Furthermore, for every check node in the MPG, $E$ contains the second edge leading into the node (i.e.~the edge that is not denoted by a dot in the MPG).
\begin{lemma}
  $\{\vec{g}_e | e\in E\}$ spans the space $\mathbb{F}_2^k$.
\end{lemma}
Since $E$ contains precisely $k$ elements, the $\{\vec{g}_e | e\in E\}$ thus form a basis of $\mathbb{F}_2^k$.
Put differently, this lemma states that the random variables associated to the edges in $E$ are independent and generate all other variables in the factor graph.
\begin{proof}
We use an inductive approach, where we visit all edges from the MPG, starting from the half-edge and working towards the leaves. During this traversal, we gradually build up a list $\tilde{E}$ of edges, which at the end of the traversal will end up being exactly $E$. At any point during the traversal, the following property will hold:

\textit{The subspace spanned by $\{\vec{g}_e | \text{The edge } e \text{ has already been visited}\}$ is exactly equal to the subspace spanned by $\{\vec{g}_e|e\in\tilde{E}\}$.}

When starting the traversal at the half-edge, we choose $\tilde{E}$ to simply contain the half-edge itself, such that this property trivially holds. Every time we come across an equality node
\begin{equation}
\tikzsetnextfilename{inoutequality}
\begin{tikzpicture}
[ 
    baseline,arrow/.style={postaction={decoration={markings,mark=at position 0.4 with {\arrow[>=stealth]{>}}},decorate}}
] 
  \pgfdeclarelayer{bg}
  \pgfsetlayers{bg,main}

  \node[squarenode] (eq) {\cequal};
  \node (SW) [below left=0.6cm and 0.2cm of eq] {};
  \node (SE) [below right=0.6cm and 0.2cm of eq] {};
  \begin{pgfonlayer}{bg}
    \draw[arrow] (SW) |- node[left] {$e_1$} (eq.west);
    \draw[arrow] (SE) |- node[right] {$e_2$} (eq.east);
    \draw[arrow] (eq.north) to ++(0,0.8cm) node[right] {$e_3$};
  \end{pgfonlayer}
\end{tikzpicture}
\end{equation}
one has by definition $\vec{g}_{e_3}=\vec{g}_{e_1}=\vec{g}_{e_2}$, so we do not need to extend $\tilde{E}$ in order for the desired property to hold. However, when we traverse a check node
\begin{equation}
\tikzsetnextfilename{inoutcheck}
\begin{tikzpicture}
[ 
    baseline, arrow/.style={postaction={decoration={markings,mark=at position 0.4 with {\arrow[>=stealth]{>}}},decorate}}
] 
  \pgfdeclarelayer{bg}
  \pgfsetlayers{bg,main}

  \node[squarenode] (eq) {$+$};
  \node (SW) [below left=0.6cm and 0.2cm of eq] {};
  \node (SE) [below right=0.6cm and 0.2cm of eq] {};
  \node[fill,circle,inner sep=0.8pt] at (eq.west) {};
  \begin{pgfonlayer}{bg}
    \draw[arrow] (SW) |- node[left] {$e_1$} (eq.west);
    \draw[arrow] (SE) |- node[right] {$e_2$} (eq.east);
    \draw[arrow] (eq.north) to ++(0,0.8cm) node[right] {$e_3$};
  \end{pgfonlayer}
\end{tikzpicture}
\end{equation}
we now have to insert $e_2$ to the list $\tilde{E}$.
Since $\vec{g}_{e_1}=\vec{g}_{e_2} + \vec{g}_{e_3}$, the desired property now holds again.
By construction, at the end of the traversal the list $\tilde{E}$ will exactly be equal to $E$.
\end{proof}

Notice in~\cref{eq:purestate_checknode} that the phases in $\ket{\Xi(e)}$ arise from check node operations are determined by the edge variable of the target edge as well as the branch variable. 
These simply accumulate as the algorithm proceeds, and therefore $c_{e,\vec j}(\vec d)$ takes the form $c_{e,\vec{j}}(\vec{d})=(-1)^{\vec{d}\cdot\vec{h}_{e,\vec{j}}}$ for some vector $\vec{h}_{e,\vec{j}}\in\mathbb{F}_2^k$.
Define $H_e:=\{\vec{h}_{e,\vec{j}} | \vec{j}\in\mathbb Z_2^{k_e-1}\}$, with $H_e=\{\vec 0\}$ for $k_e=1$. 
Then we have
\begin{lemma}
For each edge $e$, $H_e$ is the subspace of $\mathbb{F}_2^k$ spanned by $\{\vec{g}_{e'} | e'\in E$ and $e'$ precedes $e\}$.
\end{lemma}
Since the $\vec g_e$ are a basis, the span of $\{\vec{g}_{e'} | e'\in E$ and $e'$ precedes $e\}$ has dimension $k_e-1$, and therefore $H_e$ must contain $2^{k_e-1}$ distinct elements.
For the half-edge at the root, call it $e_r$, $\vec g_{e_r}$ is the standard basis vector with a single 1 in the $r$th entry. 
Notice that $H_{e_r}$ does not contain $\vec g_{e_r}$. 
\begin{proof}
As usual, the proof is by induction. 
At leaf edges $H_e=\{\vec 0\}$. 

Suppose $e$ and $e'$ are the first and second input edges to an equality node with output edge $\hat e$. 
By \cref{eq:equalityphasesmultiply} the phases multiply and therefore the associated $\vec{h}$ vectors add:
\begin{equation}
  H_{\hat e} = \{\vec{h} + \vec{h}' | \vec{h}\in H_{e}, \vec{h}'\in H_{e'}\} \, .
\end{equation}
This is the subspace spanned by $\{\vec{g}_{\bar e} | \bar e\in E$ and $\bar e$ precedes $e$ or $e'\}$, which is exactly $\{\vec{g}_{e} | e\in E$ and $e$ precedes $\hat e\}$.

Similarly, for $e$ and $e'$ incident on a check node, \cref{eq:checkphasessomething} implies
\begin{equation}
  H_{\hat e} = \{\vec{h} + \vec{h}' + l\cdot \vec{g}_{e'} | \vec{h}\in H_{e}, \vec{h}'\in H_{e'}, l\in \mathbb Z_2\} \, .
\end{equation}
Since $e'$ is an element of $E$, $H_{\hat e}$ is the space spanned by $\{\vec{g}_{e'}\}\cup \{\vec{g}_{\bar e} | \bar e\in E$ precedes $e \} \cup \{\vec{g}_{\bar e} | \bar e\in E$ precedes $e' \}$. 
Again this is precisely $\{\vec{g}_{e} | e\in E$ and $e$ precedes $\hat e\}$.
\end{proof}

For the root node $e_r$, just write $c_{\vec j}$ for $c_{e_r,\vec{j}}$, and similarly for $\vec h_{\vec j}$. 
Then 
\begin{equation}
  \sum\limits_{\tilde{x}_r} c_{\vec{j}}(x_1,\dots,x_k)c_{\vec{j}'}(x_1,\dots,x_k) = \sum\limits_{\tilde{x}_r} (-1)^{\vec{d}\cdot (\vec{h}_{\vec{j}} + \vec{h}_{\vec{j}'})} \, .
\end{equation}
The only way to avoid the summation resulting in zero is for $\vec h_{\vec j}+\vec h_{\vec j'}$ to be zero everywhere except possibly in the $r$th position, i.e., $\vec h_{\vec j}+\vec h_{\vec j'}=l\vec{g_{e_r}}$ for $l\in\{0,1\}$. 
Since the $\vec h_{\vec j}$ are all distinct, the $l=0$ option is only possible when $\vec j=\vec j'$. 
The $l=1$ option is impossible for any $\vec j$ and $\vec j'$ since $H_{e_r}$ does not contain $\vec g_{e_r}$. 
Therefore \eqref{eq:orthogonality} holds.

\section{Relation between BPQM and classical BP}\label{app:relation_to_bp}
While the BPQM algorithm and classical BP both allow for decoding classical information by passing around certain messages in a graph, at first glance this might seem to be the only similarity between the two algorithms.
In this section we show that this is not the case.
More precisely, we show that classical belief propagation decoding on the binary symmetric channel (BSC) can be expressed in a language very similar to BPQM, which clearly illustrates how BPQM can be considered a quantum analog of BP.

Consider the task of decoding a single codeword bit $X_r,r\in\{1,\dots,n\}$ given the noisy channel outputs.
Classical BP decoding can be described as a message-passing algorithm on the Tanner graph with added nodes for the channel output, as depicted in~\cref{fig:5bitcode_tanner_channelnodes}.
Alternatively, BP can also be expressed as a message-passing algorithm on the Forney-style representation of that same graph~\cite{loeliger_2004}.
One can expand this factor graph as described in~\cref{fig:degree_reduction} such that it only contains equality and check nodes of degree $3$.
This is possible as only the marginalization of the random variable $X_r$ is of interest.
When the resulting graph is a tree, we can regard BP as operating on the the same MPG as BPQM.
The messages that BP passes across the edges of the MPG are real numbers $l$ representing log-likelihood ratios.
The initial messages generated by the channel nodes are the log-likelihood ratios $l=\log\frac{P_{Y|X}(y|0)}{P_{Y|X}(y|1)}$ where $P_{Y|X}$ describes the BSC and $y$ is the observed channel output.
Every equality node operation consists of summing the two incoming messages $l_1,l_2$,
\begin{equation}
  l_1,l_2 \mapsto l_1+l_2\,.
\end{equation}
and every check node operation consists of the following operation:
\begin{equation}
  l_1,l_2 \mapsto 2\tanh^{-1}\left( \tanh\frac{l_1}{2} \tanh\frac{l_2}{2} \right)\,.
\end{equation}
If the final message $l$ generated at the root is greater than zero, then the decoder output is $0$. If $l$ is smaller than zero, then the decoder output is $1$.

To recognize the analogy between BP and BPQM, we separate the messages $l\in\mathbb{R}$ passed over the edge $e$ in BP into two parts $l=(-1)^bc$, where $b\in\{0,1\}$ and $c=|l|$.
Intuitively, one can think of $b$ being an estimate for the value of the random variable corresponding of the edge $e$, i.e., it represents the most likely value that this random variable has given the information from all preceding nodes in the MPG.
On the other hand, $c$ can be intuitively thought of as the reliability, i.e., how likely it is that the estimate $b$ is correct.
%The reason why this somewhat artificial reformulation of the messages is interesting will become clear further below.
In fact, $b$ turns out to be the classical analog of the data qubit passed over the edge $e$ in BPQM and $c$ plays the role analogous to the angle information in BPQM.
So in message-passing BPQM, the data and angle parts of the messages exactly correspond to $b$ and $c$ in the classical case.
In the language of $b$ and $c$ the equailty node operations of BP become
\begin{align}\label{eq:classical_bp_eqnode}
  (b_1,c_1), (b_2,c_2) \mapsto \Big(\Theta\big((-1)^{b_1}c_1 + (-1)^{b_2}c_2 \big), |(-1)^{b_1}c_1 + (-1)^{b_2}c_2| \Big)
\end{align}
where $\Theta(x)$ denotes the Heaviside function and and the check node operations become
\begin{align}\label{eq:classical_bp_checknode}
  (b_1,c_1), (b_2,c_2) \mapsto (b_1\oplus b_2, 2\tanh^{-1}\left( \tanh\frac{c_1}{2} \tanh\frac{c_2}{2} \right))
\end{align}
where $\oplus$ denotes addition modulo $2$. 
Note that the $\oplus$ operation is reminiscent to the \cnot{} operation used in BPQM.

The full similarity between BP and BPQM is unveiled by considering the effect of the BP node operations upon the state of the channel outputs.
A contraction argument analogous to the one introduced in~\cref{sec:contraction} then allows the resulting factor graph to be simplified to a form with a single channel node.
More precisely, the channel from $X_r$ to $Y_1,\dots,Y_n$ is now a classical channel, as the channel outputs $Y_1,\dots,Y_n$ are now classical systems.
The factor graph representing the channel from $X_1$ to $Y_1,\dots,Y_n$ for the 5-bit code is depicted in~\cref{fig:5bitcode_classical_fg}.
Here BSC$[c]$ denotes the BSC with crossover probability corresponding to the reliability $c$, i.e., such that the log-likelihood ratio corresponding to the two possible outputs are $l=\pm c$.

By consulting the relations in~\cref{eq:classical_bp_eqnode,eq:classical_bp_checknode}, one can devise a set of two contraction identities depicted in~\cref{fig:classical_contraction_identities}.
Then, by applying the corresponding node operations, one can iteratively simplify the factor graph as in BPQM.
Note that the equality node produces two output bits: The first bit $D$, obtained as in~\cref{eq:classical_bp_eqnode}, is the data bit.
The second is the ancilla bit $A$ which is simply defined as the parity $b_1\oplus b_2$ of the two input bits $b_1,b_2$.
Analogously to the equality node in BPQM which requires information about the qubit angles, the equality node operation in BP requires information about the reliabilities $c_1,c_2$ of the two input bits, and therefore it must be controlled by ancilla bits produced by previous node operations.
But there is an important difference here: These ancilla bits are not produced by preceding check nodes, as in BPQM, but rather by preceding equality nodes.
Of course, instead of conditioning equality nodes on all previous equality node ancillas, one can instead keep track of the reliability $c$ on-the-fly to obtain a true message-passing algorithm.
This is completely analogous to the on-the-fly bookkeeping of the angle in message-passing BPQM in~\cref{sec:message_passing}, and in fact this is how BP is actually run in practice.

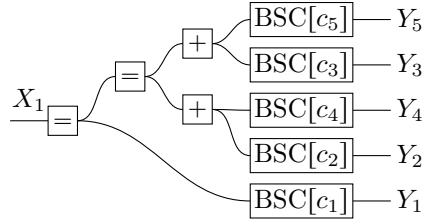
\begin{figure}
  \centering
  \tikzsetnextfilename{5bitForneyClassicalOutput}
  \begin{tikzpicture}
    \pgfdeclarelayer{bg}
    \pgfsetlayers{bg,main}

    \node[squarenode] (eq1) {\cequal};
    \node[squarenode] (eq2) [above right=0.2cm and 0.5cm of eq1] {\cequal};
    \node[squarenode] (ch1) [above right=0.05cm and 0.5cm of eq2] {$+$};
    \node[squarenode] (ch2) [below right=0.05cm and 0.5cm of eq2] {$+$};
	%\node[quadnode] (norm) [above left=0.6cm and 0cm of eq1] {$\frac{1}{4}$};
    \node[quadnode] (W5) [above right=-0.1cm and 0.5cm of ch1] {BSC$[c_5]$};
    \node[quadnode] (W3) [below=0.15cm of W5] {BSC$[c_3]$};
    \node[quadnode] (W4) [below=0.15cm of W3] {BSC$[c_4]$};
    \node[quadnode] (W2) [below=0.15cm of W4] {BSC$[c_2]$};
    \node[quadnode] (W1) [below=0.15cm of W2] {BSC$[c_1]$};

    \draw (eq1.east) to[out=0, in=180] (eq2.west);
    \draw (eq2.east) to[out=0, in=180] (ch1.west);
    \draw (eq2.east) to[out=0, in=180] (ch2.west);
    \draw (ch1.east) to[out=0, in=180] (W5.west);
    \draw (ch1.east) to[out=0, in=180] (W3.west);
    \draw (ch2.east) to[out=0, in=180] (W4.west);
    \draw (ch2.east) to[out=0, in=180] (W2.west);
    \draw (eq1.east) to[out=0, in=180] (W1.west);
    
    \draw (eq1.west) -- node[above] {$X_1$} ++(-0.5cm,0);
    % We have to put the double dashed lines in a background layer, otherwise they look ugly
    \begin{pgfonlayer}{bg}
      \draw (W5.east) to node[right=0.2cm] {$Y_5$} ++(0.5cm,0);
      \draw (W3.east) to node[right=0.2cm] {$Y_3$} ++(0.5cm,0);
      \draw (W4.east) to node[right=0.2cm] {$Y_4$} ++(0.5cm,0);
      \draw (W2.east) to node[right=0.2cm] {$Y_2$} ++(0.5cm,0);
      \draw (W1.east) to node[right=0.2cm] {$Y_1$} ++(0.5cm,0);
    \end{pgfonlayer}
  \end{tikzpicture}
  \caption{Factor graph describing the classical channel from $X_1$ to the channel outputs $Y_1,Y_2,Y_3,Y_4,Y_5$ for the 5-bit code.}
  \label{fig:5bitcode_classical_fg}
\end{figure}

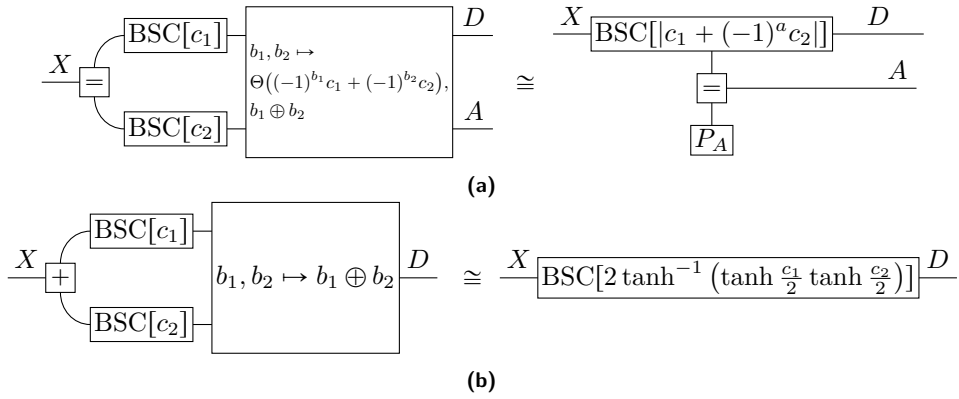
\begin{figure}
  \centering
  \begin{subfigure}[b]{\textwidth}
  \centering
      \tikzsetnextfilename{ClassicalEqualityNodeIdentity}
  \begin{tikzpicture}
    \pgfdeclarelayer{bg}
    \pgfsetlayers{bg,main}

    \node[squarenode] (eq) {\cequal};
    \node[quadnode] (W1) [above right=0.2cm and 0.2cm of eq] {BSC$[c_1]$};
    \node[quadnode] (W2) [below right=0.2cm and 0.2cm of eq] {BSC$[c_2]$};
    \node[quadnode, minimum height=2cm] (U) [right=1.8cm of eq] {\scalebox{0.7}{$\begin{aligned}& b_1,b_2 \mapsto \\ & \Theta\big((-1)^{b_1}c_1 + (-1)^{b_2}c_2 \big), \\ & b_1\oplus b_2 \end{aligned}$}};

    \draw (eq.west) -- node[above] {$X$} ++(-0.5cm,0);
    \draw (eq.north) to[out=90, in=180] (W1.west);
    \draw (eq.south) to[out=-90, in=180] (W2.west);
    % We have to put the double dashed lines in a background layer, otherwise they look ugly
    \begin{pgfonlayer}{bg}
      \draw (W1.east) to ++(3.0cm,0) to node[above] {$D$} ++(0.5cm,0);
      \draw (W2.east) to ++(3.0cm,0) to node[above] {$A$} ++(0.5cm,0);
    \end{pgfonlayer}

    % equivalence sign
    \node (equiv) [right=0.7cm of U] {$\cong$};
    
    % right-hand side:
    \node[quadnode] (Wrhs) [above right=0.2cm and 0.6cm of equiv] {BSC$[|c_1 + (-1)^{a}c_2|]$};
    \node[squarenode] (eqrhs) [below=0.3cm of Wrhs.south] {\cequal};
    \node[quadnode] (PA) [below=0.3cm of eqrhs.south] {$P_A$};
    \draw (Wrhs.west) -- node[above] {$X$} ++(-0.5cm,0);
    \draw (Wrhs.east) -- node[above] {$D$} ++(1.2cm,0);
    \draw (Wrhs.south) -- (eqrhs.north);
    \draw (eqrhs.south) -- (PA.north);
    \draw (eqrhs.east) to ++(2.0cm,0) to node[above] {$A$} ++(0.5cm,0);
  \end{tikzpicture}
  \caption{}
  \end{subfigure}
  \begin{subfigure}[b]{\textwidth}
    \tikzsetnextfilename{ClassicalCheckNodeIdentity}
    \centering
  \begin{tikzpicture}
    \pgfdeclarelayer{bg}
    \pgfsetlayers{bg,main}

    \node[squarenode] (eq) {$+$};
    \node[quadnode] (W1) [above right=0.2cm and 0.2cm of eq] {BSC$[c_1]$};
    \node[quadnode] (W2) [below right=0.2cm and 0.2cm of eq] {BSC$[c_2]$};
    \node[quadnode, minimum height=2cm] (U) [right=1.8cm of eq] {$b_1,b_2\mapsto b_1\oplus b_2$};

    \draw (eq.west) -- node[above] {$X$} ++(-0.5cm,0);
    \draw (eq.north) to[out=90, in=180] (W1.west);
    \draw (eq.south) to[out=-90, in=180] (W2.west);
    \draw (U.east) to node[above] {$D$} ++(0.5cm,0);
    % We have to put the double dashed lines in a background layer, otherwise they look ugly
    \begin{pgfonlayer}{bg}
      \draw (W1.east) to ++(0.5cm,0);
      \draw (W2.east) to ++(0.5cm,0);
    \end{pgfonlayer}

    % equivalence sign
    \node (equiv) [right=0.7cm of U] {$\cong$};
    
    % right-hand side:
    \node[quadnode] (Wrhs) [right=0.6cm of equiv] {BSC$[2\tanh^{-1}\left( \tanh\frac{c_1}{2} \tanh\frac{c_2}{2}\right) ]$};
    \draw (Wrhs.west) -- node[above] {$X$} ++(-0.5cm,0);
    \draw (Wrhs.east) -- node[above] {$D$} ++(0.5cm,0);
  \end{tikzpicture}
  \caption{}
  \end{subfigure}
  \caption{Graphical depiction of equality node (a) and check node (b) contraction identities. These identities are the classical counterpart to those depicted in~\cref{fig:contraction_identities_equality,fig:contraction_identities_check}. The second output bit $A$ of the large box in figure (a) is necessary in order to make the computation reversible and thus provide information about the reliability to subsequent node operations (analogous to the check node ancilla $Q_2$ in~\cref{fig:contraction_identities_check}). The probability distribution $P_A$ is defined as $P_A(0)=1 / (\exp(c_1) + \exp(c_2) + \exp(c_1+c_2))$ and $P_A(1)=1-P_A(0)$.}
  \label{fig:classical_contraction_identities}
\end{figure}

\section{Gate decomposition of equality node unitary}\label{app:ustar_identity}
In this appendix, we derive the quantum circuit identity depicted in~\cref{fig:uc_ustar_identity}. Recall, that the matrix representation of the variable node unitary in \cref{eq:equalitynodeU} is given by
\begin{equation}
  U_{\ostar}(\theta,\varphi) = 
  \begin{pmatrix}
    a_+ & 0   & 0   & a_- \\
    -a_- & 0   & 0   &a_+ \\
    0   & b_- & b_+ & 0   \\
    0   & b_+ &-b_- & 0 
  \end{pmatrix}\,,
\end{equation}
for $a_\pm$ and $b_\pm$ functions of $\theta$ and $\varphi$ as in \cref{eq:apmdef} and \cref{eq:bpmdef}. 
Applying a \cnot{} gate controlled by the second qubit (which we denote as \notc) produces the block diagonal form 
\begin{equation}
  U_{\ostar}(\theta,\varphi) = 
  \begin{pmatrix}
    a_+ & a_- & 0   & 0   \\
    -a_- &a_+ & 0   & 0   \\
    0   & 0   & b_+ & b_- \\
    0   & 0   &-b_- & b_+
  \end{pmatrix}
  \cdot \notc{}
  = \left( \ket{0}\bra{0}\otimes A + \ket{1}\bra{1}\otimes B \right) \cdot \notc{}
\end{equation}
where $A:=\begin{pmatrix}a_+&a_-\\-a_-&a_+\end{pmatrix}$ and $B:=\begin{pmatrix}b_+&b_-\\-b_-&b_+\end{pmatrix}$. 
The blocks $A$ and $B$ are both $R_y$ rotations. 
Denoting $R_y(\phi)=e^{-\tfrac12i\phi\sigma_y}$, we have $A=R_y(\phi_0)$ for $\phi_0 := -2\arccos{(a_+)}$ and $B=R_y(\phi_1)$ for $\phi_1 :=-2\arccos(b_+)$. 
Expressed as a circuit, this is 
\tikzexternaldisable
\begin{equation}
  \begin{quantikz}
    & \gate[wires=2]{U_{\ostar}(\theta,\varphi)} & \qw \\
    & & \qw
  \end{quantikz}
  \cong
  \begin{quantikz}
    & \targ{}   &    \uctrl{i}\vqw{1}  & \qw \\
    & \ctrl{-1} &  \gate{R_y(\phi_i)} & \qw
  \end{quantikz}\,.
\end{equation}
The desired identity then immediately follows from
\begin{equation}
  \begin{quantikz}
    & \uctrl{i}\vqw{1}  & \qw \\
    & \gate{R_y(\phi_i)} & \qw
  \end{quantikz}
  \cong
  \begin{quantikz}
    & \qw                                & \ctrl{1} & \qw                                & \ctrl{1} & \qw \\
    & \gate{R_y(\frac{\phi_0+\phi_1}{2})} & \targ{}  & \gate{R_y(\frac{\phi_0-\phi_1}{2})} & \targ{}  & \qw 
  \end{quantikz}\,.
\end{equation}
\tikzexternalenable

% A similar but actually simpler decomposition can be given for $U^{\text{opt}}$ in \cref{eq:optimalcloningU}. 
% Like $U_\ostar^\dagger$, it is a uniformly controlled $R_y$ gate followed by \notc. 
% However, the two angles are $\phi_0=2\arccos\left(\frac{1+f}{\sqrt2 \sqrt{1+f^2}}\right)$ and $\phi_1=\frac\pi 2$.

\section{Proof of \texorpdfstring{\cref{lem:discretization_error}}{Lemma 5.1}}
\label{app:proof_discretization_error}
Similarly to the proof of~\cref{lem:purestate_evolution}, we prove a more general statement about the intermediate states of the message-passing BPQM algorithm. This will directly imply the desired theorem.

Consider the BPQM algorithm executed on the MPG $G$ w.r.t.\ the codeword bit $X_r$. 
As shown in~\cref{lem:stronger44}, the joint state of the qubit passed over some edge $e$ together with all ancilla qubits produced by nodes preceding $e$ can be written as
\begin{equation}\label{eq:imediatestate_1}
  \sum\limits_{\vec{j}\in\{0,1\}^{k_e-1}}{ \gamma_{e,\vec{j}} \qstate{z_e}{\theta_{e,\vec{j}}}_{D_e} \otimes \ket{0^{n_e-k_e}}_{Z_e} \otimes \ket{\vec{j}}_{A_e} }\,,
\end{equation}
for some $\gamma_{e,\vec{j}}\in\mathbb{R}, z_e\in\{0,1\}, \theta_{e,\vec{j}}\in (0,\pi)$.
Here $n_e$ denotes the number of channel nodes preceding $e$ and $k_e-1$ denotes the number of check nodes preceding $e$.
The system $D_e$ describes the data qubit passed over the edge $e$, $Z_e$ contains all ancilla qubits produced by equality nodes preceding $e$ and $A_e$ contains all ancilla qubits produced by check nodes preceding $e$.

In the course of the proof, we will inductively show the following statement for message-passing BPQM: The joint state of the message sent over the edge $e$ together with all qubits stored in preceding check and variable nodes can be written in the form
\begin{equation}\label{eq:imediatestate_2}
  \sum\limits_{\vec{j}\in\{0,1\}^{k_e-1}}{ \tilde{\gamma}_{e,\vec{j}} \ket{\tilde{\varphi}_{e,\vec{j}}}_{D_eZ_e} \otimes \ket{\vec{j}}_{A_e} \otimes \anglestate{\tilde{c}_{e,\vec{j}}}_{C_e} } \otimes \ket{\tilde{s}_{e,\vec{j}}}_{S_e}
\end{equation}
for some $\tilde{\gamma}_{e,\vec{j}}\in\mathbb{R}$, $(n_e-k_e+1)$-qubit states $\ket{\tilde{\varphi}_{e,\vec{j}}}$, $\tilde{c}_{e,\vec{j}}\in\mathcal{A}_B$ and $(n_e-1)2B$-qubit states $\ket{\tilde{s}_{e,\vec{j}}}$. The message passed over the edge $e$ consists of the data part $D_e$ as well as the angle part $C_e$. 
We denote by $S_e$ the system of the remaining qubits produced by nodes preceeding $e$ and that were not passed to the subsequent node.

Note that the statevectors in~\cref{eq:imediatestate_1,eq:imediatestate_2} are properly normalized, implying that $\sum_{\vec{j}}\gamma_{e,\vec{j}}^2 = \sum_{\vec{j}}\tilde{\gamma}_{e,\vec{j}}^2=1$.

Furthermore, this intermediate state for any edge $e$ fulfills the following properties:
\begin{enumerate}
  \item $\sum\limits_{\vec{j}} \gamma_{e,\vec{j}}^2 | \cos\theta_{e,\vec{j}} - \tilde{c}_{e,\vec{j}} | \leq (2^{n_e+1}-3)\delta$,
  \item $\sum\limits_{\vec{j}} | \gamma_{e,\vec{j}}^2 - \tilde{\gamma}_{e,\vec{j}}^2 | \leq 2^{n+1/2}\pi\sqrt{\delta}\frac{1}{3}26^{n_e}$,
  \item $\sum\limits_{\vec{j}} \gamma_{e,\vec{j}}^2 \norm{ \qstate{z_e}{\theta_{e,\vec{j}}}\otimes\ket{0^{n_e-k_e}} - \ket{\tilde{\varphi}_{e,\vec{j}}} } \leq 2^{n+1/2}\pi\sqrt{\delta}26^{n_e}$.
\end{enumerate}

We prove our statement in an inductive manner. It is trivially true for the edges $e$ coming out of a leaf node. So it suffices to show the following statement: If two edges $e_1,e_2$ which lead into the same check or equality node fulfill the desired properties
\begin{equation}
\tikzsetnextfilename{equalityorcheck}
\begin{tikzpicture}[baseline]
  \pgfdeclarelayer{bg}
  \pgfsetlayers{bg,main}

  \node[quadnode] (eq) {$=$ or $+$};
  \begin{pgfonlayer}{bg}
  \begin{scope}[decoration={markings,mark=at position 0.5 with {\arrow{>}}}]
    \draw[postaction={decorate}] (eq)+(-0.5cm,-0.5cm) to node[left=0.1cm] {$e_1$} ++(-0.1cm,-0.1cm);
    \draw[postaction={decorate}] (eq)+(0.5cm,-0.5cm) to node[right=0.1cm] {$e_2$} ++(0.1cm,-0.1cm);
    \draw[postaction={decorate}] (eq) to ++(0,0.6cm) node[right] {$e_3$};
  \end{scope}
  \end{pgfonlayer}
\end{tikzpicture}
\end{equation}
then the output edge $e_3$ also fulfills the desired properties. 
For this purpose, we deal separately with the two cases where the node is an equality or check node. In each case we will prove that the intermediate state at $e_3$ is of the correct form and that it fulfills the three desired properties.

For simplicity of notation, let us write the joint state of the data sent over the edges $e_1,e_2$ and the corresponding ancillas produced by preceding nodes as
\begin{align}
  \left( \sum\limits_{\vec{j}_1\in\{0,1\}^{k_1-1}}{ \gamma_{1,\vec{j}_1} \qstate{z_1}{\theta_{1,\vec{j}_1}}_{D_1}  \ket{0^{n_1-k_1}}_{Z_1}  \ket{\vec{j}_1}_{A_1} } \right) \otimes \nonumber \\
  \left( \sum\limits_{\vec{j}_2\in\{0,1\}^{k_2-1}}{ \gamma_{2,\vec{j}_2} \qstate{z_2}{\theta_{2,\vec{j}_2}}_{D_2}  \ket{0^{n_2-k_2}}_{Z_2}  \ket{\vec{j}_2}_{A_2} } \right)
\end{align}
for BPQM and as
\begin{align}
  \left( \sum\limits_{\vec{j}_1\in\{0,1\}^{k_1-1}}{ \tilde{\gamma}_{1,\vec{j}_1} \ket{\tilde{\varphi}_{1,\vec{j}_1}}_{D_1Z_1} \otimes \ket{\vec{j}_1}_{A_1} \otimes \anglestate{\tilde{c}_{1,\vec{j}_1}}_{C_1} } \otimes \ket{\tilde{s}_{1,\vec{j}_1}}_{S_1} \right) \otimes \nonumber \\
  \left( \sum\limits_{\vec{j}_2\in\{0,1\}^{k_2-1}}{ \tilde{\gamma}_{2,\vec{j}_2} \ket{\tilde{\varphi}_{2,\vec{j}_2}}_{D_2Z_2} \otimes \ket{\vec{j}_2}_{A_2} \otimes \anglestate{\tilde{c}_{2,\vec{j}_2}}_{C_2} } \otimes \ket{\tilde{s}_{2,\vec{j}_2}}_{S_2} \right)
\end{align}
for message-passing BPQM.

\subsection{Check node} 
In the original BPQM description, the check node merely consists of a $\cnot_{D_1D_2}$ gate on the qubits $D_1$ and $D_2$, therefore by using~\cref{eq:purestate_checknode} the resulting state is given by
\begin{equation}\label{eq:check_node_after_state_original}
  \sum\limits_{\vec{j}_1,\vec{j}_2,l\in\{0,1\}}{ \gamma_{1,\vec{j}_1}\gamma_{2,\vec{j}_2}\kappa_{\vec{j}_1,\vec{j}_2,l} \qstate{z_1\oplus z_2}{\theta_{1,\vec{j}_1}\boxstar_l\theta_{2,\vec{j}_2}}_{D_3} \otimes \ket{0^{n_3-k_3}}_{Z_3} \otimes \left(\ket{\vec{j}_1}\ket{\vec{j}_2}\ket{l}\right)_{A_3} } 
\end{equation}
where $\kappa_{\vec{j}_1,\vec{j}_2,l}:=(-1)^{l\cdot z_2}\tfrac1{\sqrt2}\sqrt{1+(-1)^l\cos(\theta_{1,\vec{j}_1})\cos(\theta_{2,\vec{j}_2})}$. Here we use $n_3:=n_1+n_2$ and $k_3:=k_1+k_2$ and introduced the new systems $D_3=D_1$, $Z_3=Z_1Z_2$ and $A_3=A_1A_2D_2$.

In a similar fashion, we now describe how the input state evolves in the message-passing picture. 
First we apply the $\cnot_{D_1D_2}$ gate. We write the effect on the state $\ket{\varphi_{1,\vec{j}_1}}_{D_1Z_1}\otimes\ket{\varphi_{2,\vec{j}_2}}_{D_2Z_2}$ as
\begin{equation}
  \cnot_{D_1D_2}\ket{\tilde{\varphi}_{1,\vec{j}_1}}_{D_1Z_1}\otimes\ket{\tilde{\varphi}_{2,\vec{j}_2}}_{D_2Z_2} = \tilde{\kappa}_{\vec{j}_1,\vec{j}_2,0}\ket{\xi_{\vec{j}_1,\vec{j}_2,0}}_{D_1Z_1Z_2}\ket{0}_{D_2} + \tilde{\kappa}_{\vec{j}_1,\vec{j}_2,1}\ket{\xi_{\vec{j}_1,\vec{j}_2,1}}_{D_1Z_1Z_2}\ket{1}_{D_2}\,,
\end{equation}
for
\begin{equation}\label{eq:def_kappa}
  \tilde{\kappa}_{\vec{j}_1,\vec{j}_2,l}:=(-1)^{l\cdot z_2}\norm{P_l\cdot \cnot_{D_1D_2}\cdot \ket{\tilde{\varphi}_{1,\vec{j}_1}}_{D_1Z_1}\otimes\ket{\tilde{\varphi}_{2,\vec{j}_2}}_{D_2Z_2}} \in\mathbb{R}
\end{equation}
and
\begin{equation}
  \ket{\xi_{\vec{j}_1,\vec{j}_2,l}}:=\frac{1}{\tilde{\kappa}_{\vec{j}_1,\vec{j}_2,l}}P_l\cdot \cnot_{D_1D_2}\cdot \ket{\tilde{\varphi}_{1,\vec{j}_1}}_{D_1Z_1}\otimes\ket{\tilde{\varphi}_{2,\vec{j}_2}}_{D_2Z_2}\,,
\end{equation}
where $P_l$ is the isometry $\id_{D_1Z_1Z_2} \otimes\bra{l}_{D_2}$ giving the amplitude in the subspace where the target qubit of the \cnot{} gate is in the state $\ket{l}$ for $l\in\{0,1\}$. Considering the classical computation of the output angle, the resulting final state will thus be
\begin{align}
  \sum\limits_{\vec{j}_1,\vec{j}_2,l} \tilde{\gamma}_{1,\vec{j}_1}\tilde{\gamma}_{2,\vec{j}_2}\tilde{\kappa}_{\vec{j}_1,\vec{j}_2,l} \ket{\xi_{\vec{j}_1,\vec{j}_2,l}}_{D_3Z_3} \otimes (\ket{\vec{j}_1}\ket{\vec{j}_2}\ket{l})_{A_3} \otimes \anglestate{\quant\left(\frac{\tilde{c}_{1,\vec{j}_1}+(-1)^l\tilde{c}_{2,\vec{j}_2}}{1+(-1)^l\tilde{c}_{1,\vec{j}_1}\tilde{c}_{2,\vec{j}_2}}\right)}_{C_3} \nonumber \\
  \otimes \left( \ket{\tilde{s}_{1,\vec{j}_1}}\ket{\tilde{s}_{2,\vec{j}_2}}\anglestate{\tilde{c}_{1,\vec{j}_1}} \anglestate{\tilde{c}_{2,\vec{j}_2}} \right)_{S_3}
\end{align}
where $S_3$ contains $S_1$, $S_2$ and the two $B$-qubit registers consumed by the node operation (in accordance with~\cref{fig:message_passing_rules}).

Before we prove the three statements of the lemma, we show that the summed square difference between $\kappa_{\vec{j}_1,\vec{j}_2,l}$ and $\tilde{\kappa}_{\vec{j}_1,\vec{j}_2,l}$ can be bounded:
\begin{align}
  &\sum\limits_l{ \abs{\kappa_{\vec{j}_1,\vec{j}_2,l}^2 - \tilde{\kappa}_{\vec{j}_1,\vec{j}_2,l}^2} }\nonumber\\
  &\leq \sum\limits_l{ 2\abs{\abs{\kappa_{\vec{j}_1,\vec{j}_2,l}} - \abs{\tilde{\kappa}_{\vec{j}_1,\vec{j}_2,l}}} }\\
%\resizebox{\textwidth}{!}{$
  &= 2\sum\limits_l{\abs{ \norm{P_l\cdot \cnot{}\cdot \ket{\tilde{\varphi}_{1,\vec{j}_1}}\otimes\ket{\tilde{\varphi}_{2,\vec{j}_2}}} - \norm{P_l\cdot \cnot{}\cdot \qstate{z_1}{\theta_{1,\vec{j}_1}}  \ket{0^{n_1-k_1}}\otimes\qstate{z_2}{\theta_{2,\vec{j}_2}}  \ket{0^{n_2-k_2}} }}}\\
%$}
  &\leq 2\sum\limits_l{\norm{ P_l\cdot \cnot{}\cdot \left( \ket{\tilde{\varphi}_{1,\vec{j}_1}}\otimes\ket{\tilde{\varphi}_{2,\vec{j}_2}} - \qstate{z_1}{\theta_{1,\vec{j}_1}} \otimes \ket{0^{n_1-k_1}}\otimes\qstate{z_2}{\theta_{2,\vec{j}_2}} \otimes \ket{0^{n_2-k_2}} \right) }}\\
  &\leq 4 \norm{ \ket{\tilde{\varphi}_{1,\vec{j}_1}}\otimes\ket{\tilde{\varphi}_{2,\vec{j}_2}} - \qstate{z_1}{\theta_{1,\vec{j}_1}} \otimes \ket{0^{n_1-k_1}}\otimes\qstate{z_2}{\theta_{2,\vec{j}_2}} \otimes \ket{0^{n_2-k_2}} }\\
  &\leq 4 \norm{ \ket{\tilde{\varphi}_{1,\vec{j}_1}} - \qstate{z_1}{\theta_{1,\vec{j}_1}} \otimes \ket{0^{n_1-k_1}} } + 4 \norm{ \ket{\tilde{\varphi}_{2,\vec{j}_2}} - \qstate{z_2}{\theta_{2,\vec{j}_2}} \otimes \ket{0^{n_2-k_2}} }\,,
\end{align}
where the first inequality follows from $|\kappa_{\vec{j}_1,\vec{j}_2,l}|,|\tilde{\kappa}_{\vec{j}_1,\vec{j}_2,l}|\leq 1$ and the Lipschitz continuity of the function $x\mapsto x^2$ on the interval $[-1,1]$ as well as the fact that $\kappa_{\vec{j}_1,\vec{j}_2,l}$ and $\tilde{\kappa}_{\vec{j}_1,\vec{j}_2,l}$ have the identical sign.
The second inequality follows from the reverse triangle inequality, the third inequality follows because the operator norm of $P_l\cdot \cnot$ is $1$, and the last inequality follows from the triangle inequality and the fact that 
\begin{align}
  \lVert \ket{f_1} \otimes \ket{g_1} - \ket{f_2} \otimes \ket{g_2}\rVert &= \lVert \ket{f_1} \otimes \ket{g_1} - \ket{f_2} \otimes (\ket{g_1} + (\ket{g_2}-\ket{g_1})) \rVert \nonumber \\
  &= \lVert (\ket{f_1}-\ket{f_2}) \otimes \ket{g_1} + \ket{f_2} \otimes (\ket{g_2}-\ket{g_1}) \rVert \nonumber \\
  &\leq \lVert \ket{f_1}-\ket{f_2}\rVert \lVert \ket{g_1}\rVert + \lVert \ket{f_2}\rVert \lVert \ket{g_1}-\ket{g_2}\rVert \nonumber \\
  &= \lVert \ket{f_1}-\ket{f_2}\rVert + \lVert \ket{g_1}-\ket{g_2}\rVert
\end{align}
for any normalized vectors $\ket{f_1},\ket{f_2},\ket{g_1}\ket{g_2}$.
We now prove the three statements of the lemma.

\subsubsection{Property 1}

\begin{align}
  &\sum\limits_{\vec{j}_1,\vec{j}_2,l} \gamma_{1,\vec{j}_1}^2\gamma_{2,\vec{j}_2}^2\kappa_{\vec{j}_1,\vec{j}_2,l}^2 \cdot \left| \cos(\theta_{1,\vec{j}_1}\boxstar_l\theta_{2,\vec{j}_2}) - \quant\left(\frac{\tilde{c}_{1,\vec{j}_1} + (-1)^l\tilde{c}_{2,\vec{j}_2}}{1 + (-1)^l\tilde{c}_{1,\vec{j}_1}\tilde{c}_{2,\vec{j}_2}}\right) \right|\\
  &\leq \delta + \sum\limits_{\vec{j}_1,\vec{j}_2,l} \gamma_{1,\vec{j}_1}^2\gamma_{2,\vec{j}_2}^2\kappa_{\vec{j}_1,\vec{j}_2,l}^2 \cdot \left| \cos(\theta_{1,\vec{j}_1}\boxstar_l\theta_{2,\vec{j}_2}) - \frac{\tilde{c}_{1,\vec{j}_1} + (-1)^l\tilde{c}_{2,\vec{j}_2}}{1 + (-1)^l\tilde{c}_{1,\vec{j}_1}\tilde{c}_{2,\vec{j}_2}} \right|\\
  &= \delta + \sum\limits_{\vec{j}_1,\vec{j}_2,l} \gamma_{1,\vec{j}_1}^2\gamma_{2,\vec{j}_2}^2 \cdot \left| \frac{\left(\cos\theta_{1,\vec{j}_1} + (-1)^l\cos\theta_{2,\vec{j}_2} \right)}{2} - \kappa_{\vec{j}_1,\vec{j}_2,l}^2\frac{\tilde{c}_{1,\vec{j}_1} + (-1)^l\tilde{c}_{2,\vec{j}_2}}{1 + (-1)^l\tilde{c}_{1,\vec{j}_1}\tilde{c}_{2,\vec{j}_2}} \right|
\end{align}
    \begin{align}
      \leq \delta + \sum\limits_{\vec{j}_1,\vec{j}_2,l} \gamma_{1,\vec{j}_1}^2\gamma_{2,\vec{j}_2}^2 \cdot | \frac{\left(\cos\theta_{1,\vec{j}_1} + (-1)^l\cos\theta_{2,\vec{j}_2} \right)}{2} - \frac{1+(-1)^l\tilde{c}_{1,\vec{j}_1}\tilde{c}_{2,\vec{j}_2}}{2}\frac{\tilde{c}_{1,\vec{j}_1} + (-1)^l\tilde{c}_{2,\vec{j}_2}}{1 + (-1)^l\tilde{c}_{1,\vec{j}_1}\tilde{c}_{2,\vec{j}_2}} | + \nonumber \\
       \sum\limits_{\vec{j}_1,\vec{j}_2,l} \gamma_{1,\vec{j}_1}^2\gamma_{2,\vec{j}_2}^2 \cdot | \frac{1+(-1)^l\tilde{c}_{1,\vec{j}_1}\tilde{c}_{2,\vec{j}_2}}{2}\frac{\tilde{c}_{1,\vec{j}_1} + (-1)^l\tilde{c}_{2,\vec{j}_2}}{1 + (-1)^l\tilde{c}_{1,\vec{j}_1}\tilde{c}_{2,\vec{j}_2}} - \kappa_{\vec{j}_1,\vec{j}_2,l}^2\frac{\tilde{c}_{1,\vec{j}_1} + (-1)^l\tilde{c}_{2,\vec{j}_2}}{1 + (-1)^l\tilde{c}_{1,\vec{j}_1}\tilde{c}_{2,\vec{j}_2}} |
    \end{align}
    \begin{align}
      \leq \delta + \frac{1}{2} \sum\limits_{\vec{j}_1,\vec{j}_2,l} \gamma_{1,\vec{j}_1}^2\gamma_{2,\vec{j}_2}^2 | \left(\cos\theta_{1,\vec{j}_1} + (-1)^l\cos\theta_{2,\vec{j}_2} \right) - \left( \tilde{c}_{1,\vec{j}_1} + (-1)^l\tilde{c}_{2,\vec{j}_2}\right) | + \nonumber \\
      \frac{1}{2}\sum\limits_{\vec{j}_1,\vec{j}_2,l} \gamma_{1,\vec{j}_1}^2\gamma_{2,\vec{j}_2}^2 |\cos\theta_{1,\vec{j}_1}\cos\theta_{2,\vec{j}_2} - \tilde{c}_{1,\vec{j}_1}\tilde{c}_{2,\vec{j}_2}|
    \end{align}
    Here we used that $\abs{\frac{\tilde{c}_{1,\vec{j}_1}+(-1)^l\tilde{c}_{2,\vec{j}_2}}{1+(-1)^l\tilde{c}_{1,\vec{j}_1}\tilde{c}_{2,\vec{j}_2}}}\leq 1$.
    \begin{align}
      \leq \delta + 
      \sum\limits_{\vec{j}_1,\vec{j}_2} \gamma_{1,\vec{j}_1}^2\gamma_{2,\vec{j}_2}^2 | \cos\theta_{1,\vec{j}_1} - \tilde{c}_{1,\vec{j}_1} | +
      \sum\limits_{\vec{j}_1,\vec{j}_2} \gamma_{1,\vec{j}_1}^2\gamma_{2,\vec{j}_2}^2 | \cos\theta_{2,\vec{j}_2} - \tilde{c}_{2,\vec{j}_2} | + \nonumber \\
      \sum\limits_{\vec{j}_1,\vec{j}_2} \gamma_{1,\vec{j}_1}^2\gamma_{2,\vec{j}_2}^2 | \cos\theta_{1,\vec{j}_1} - \tilde{c}_{1,\vec{j}_1} | +
      \sum\limits_{\vec{j}_1,\vec{j}_2} \gamma_{1,\vec{j}_1}^2\gamma_{2,\vec{j}_2}^2 | \cos\theta_{2,\vec{j}_2} - \tilde{c}_{2,\vec{j}_2} |
    \end{align}
    \begin{align}
      \leq \delta + 2\left( (2^{n_1+1}-3) + (2^{n_2+1}-3) \right)\delta = \delta \left( 1 + 4(2^{n_1}+2^{n-n_1}) - 12 \right)
    \end{align}
    By using that $2^{n_1}+2^{n_3-n_1}$ for $n_1\in \{1,2,\dots,n_3-1\}$ is maximized by $n_1=n_3-1$ we obtain that this term is bounded above by 
    \begin{equation}
       \delta \left( 1 + 2^{n_3+1} + 8 - 12 \right) = \delta \left( 2^{n_3+1} - 3 \right)\,.
    \end{equation}
    Note that the above derivation of property 1 is precisely the reason why~\cref{lem:discretization_error} is formulated in a probability-averaged manner instead of bounding the discretisation errors in every branch of the wavefunction separately.

  \subsubsection{Property 2}
\begin{align}
  \sum\limits_{\vec{j}_1,\vec{j}_2,l}&{ \left| \gamma_{1,\vec{j}_1}^2\gamma_{2,\vec{j}_2}^2\kappa_{\vec{j}_1,\vec{j}_2,l}^2 - \tilde{\gamma}_{1,\vec{j}_1}^2\tilde{\gamma}_{2,\vec{j}_2}^2\tilde{\kappa}_{\vec{j}_1,\vec{j}_2,l}^2 \right| }\nonumber\\
  &= \sum\limits_{\vec{j}_1,\vec{j}_2,l}  \left| \left(\gamma_{1,\vec{j}_1}^2-\tilde{\gamma}_{1,\vec{j}_1}^2\right)\gamma_{2,\vec{j}_2}^2\tilde{\kappa}_{\vec{j}_1,\vec{j}_2,l}^2 + \left(\gamma_{2,\vec{j}_2}^2-\tilde{\gamma}_{2,\vec{j}_2}^2\right)\tilde{\gamma}_{1,\vec{j}_1}^2\tilde{\kappa}_{\vec{j}_1,\vec{j}_2,l}^2 + \left(\kappa_{\vec{j}_1,\vec{j}_2,l}^2-\tilde{\kappa}_{\vec{j}_1,\vec{j}_2,l}^2\right)\gamma_{1,\vec{j}_1}^2\gamma_{2,\vec{j}_2}^2 \right| \\
  &\leq \sum\limits_{\vec{j}_1,\vec{j}_2,l}{  \left|\gamma_{1,\vec{j}_1}^2-\tilde{\gamma}_{1,\vec{j}_1}^2\right|\gamma_{2,\vec{j}_2}^2\tilde{\kappa}_{\vec{j}_1,\vec{j}_2,l}^2 + \left|\gamma_{2,\vec{j}_2}^2-\tilde{\gamma}_{2,\vec{j}_2}^2\right|\tilde{\gamma}_{1,\vec{j}_1}^2\tilde{\kappa}_{\vec{j}_1,\vec{j}_2,l}^2 + \left|\kappa_{\vec{j}_1,\vec{j}_2,l}^2-\tilde{\kappa}_{\vec{j}_1,\vec{j}_2,l}^2\right|\gamma_{1,\vec{j}_1}^2\gamma_{2,\vec{j}_2}^2 }\\
  &= \sum\limits_{\vec{j}_1}{  \left|\gamma_{1,\vec{j}_1}^2-\tilde{\gamma}_{1,\vec{j}_1}^2\right| }
  + \sum\limits_{\vec{j}_2}{  \left|\gamma_{2,\vec{j}_2}^2-\tilde{\gamma}_{2,\vec{j}_2}^2\right| }
  + \sum\limits_{\vec{j}_1,\vec{j}_2,l}{  \left|\kappa_{\vec{j}_1,\vec{j}_2,l}^2-\tilde{\kappa}_{\vec{j}_1,\vec{j}_2,l}^2\right| \gamma_{1,\vec{j}_1}^2\gamma_{2,\vec{j}_2}^2 }\\
  &\leq 2^{n+1/2}\pi\sqrt{\delta}\frac{1}{3}26^{n_1}  +  2^{n+1/2}\pi\sqrt{\delta}\frac{1}{3}26^{n_2}  + 
  4\sum\limits_{\vec{j}_1}{\gamma_{1,\vec{j}_1}^2\norm{ \ket{\tilde{\varphi}_{1,\vec{j}_1}} - \qstate{z_1}{\theta_{1,\vec{j}_1}} \otimes \ket{\vec{j}_1} }} \nonumber\\
  &\phantom{\leq}+4\sum\limits_{\vec{j}_2}{\gamma_{2,\vec{j}_2}^2\norm{ \ket{\tilde{\varphi}_{2,\vec{j}_2}} - \qstate{z_2}{\theta_{2,\vec{j}_2}} \otimes \ket{\vec{j}_2} }}\,,
\end{align}
where we used the bound on the difference between $\kappa_{\vec{j}_1,\vec{j}_2,l}^2$ and $\tilde{\kappa}_{\vec{j}_1,\vec{j}_2,l}^2$ shown previously. 
We thus obtain the upper bound 
\begin{align}
  &2^{n+1/2}\pi\sqrt{\delta}\frac{1}{3} \left( 26^{n_1} + 26^{n_2} + 12\cdot 26^{n_1} + 12\cdot 26^{n_2} \right)\nonumber\\
  &= 2^{n+1/2}\pi\sqrt{\delta}\frac{1}{3} \left( 13\cdot 26^{n_1} + 13\cdot 26^{n_2} \right)\\
  &\leq 2^{n+1/2}\pi\sqrt{\delta}\frac{1}{3} \left( 13\cdot 26^{n-1} + 13\cdot 26^{n-1} \right)\\
  & = 2^{n+1/2}\pi\sqrt{\delta}\tfrac{1}{3} 26^n\,.
\end{align}

\subsubsection{Property 3}
\begin{align}
  &\sum\limits_{\vec{j}_1,\vec{j}_2,l}{ \gamma_{1,\vec{j}_1}^2\gamma_{2,\vec{j}_2}^2\kappa_{\vec{j}_1,\vec{j}_2,l}^2 \norm{ \qstate{z_1\oplus z_2}{\theta_{1,\vec{j}_1}\boxstar_l\theta_{2,\vec{j}_2}} \otimes \ket{0^{n_1-k_1}}\ket{0^{n_2k-_2}} - \ket{\xi_{\vec{j}_1,\vec{j}_2,l}}} } \nonumber\\
  &= \sum\limits_{\vec{j}_1,\vec{j}_2,l}{ \norm{
  \gamma_{1,\vec{j}_1}^2\gamma_{2,\vec{j}_2}^2\kappa_{\vec{j}_1,\vec{j}_2,l}^2  \qstate{z_1\oplus z_2}{\theta_{1,\vec{j}_1}\boxstar_l\theta_{2,\vec{j}_2}} \otimes \ket{0^{n_1-k_1}}\ket{0^{n_2-k_2}}
  - \gamma_{1,\vec{j}_1}^2\gamma_{2,\vec{j}_2}^2\kappa_{\vec{j}_1,\vec{j}_2,l}^2 \ket{\xi_{\vec{j}_1,\vec{j}_2,l}}} } \\
\label{eq:proof51_prop13_eq1}
  &\leq \sum\limits_{\vec{j}_1,\vec{j}_2,l} \norm{
  \gamma_{1,\vec{j}_1}^2\gamma_{2,\vec{j}_2}^2\kappa_{\vec{j}_1,\vec{j}_2,l}^2  \qstate{z_1\oplus z_2}{\theta_{1,\vec{j}_1}\boxstar_l\theta_{2,\vec{j}_2}} \otimes \ket{0^{n_1-k_1}}\ket{0^{n_2-k_2}}
  - \tilde{\gamma}_{1,\vec{j}_1}^2\tilde{\gamma}_{2,\vec{j}_2}^2\tilde{\kappa}_{\vec{j}_1,\vec{j}_2,l}^2 \ket{\xi_{\vec{j}_1,\vec{j}_2,l}} } \nonumber \\
  &\phantom{\leq}+ \sum\limits_{\vec{j}_1,\vec{j}_2,l} | \gamma_{1,\vec{j}_1}^2\gamma_{2,\vec{j}_2}^2\kappa_{\vec{j}_1,\vec{j}_2,l}^2 - \tilde{\gamma}_{1,\vec{j}_1}^2\tilde{\gamma}_{2,\vec{j}_2}^2\tilde{\kappa}_{\vec{j}_1,\vec{j}_2,l}^2 |\,.
\end{align}
The second sum in~\cref{eq:proof51_prop13_eq1} is upper bounded by $2^{n+1/2}\pi\sqrt{\delta}(\frac{13}{3}26^{n_1}+\frac{13}{3}26^{n_2})$ by point 2. 
For the first term we use
\begin{align}
  &\kappa_{\vec{j}_1,\vec{j}_2,l}
  \qstate{z_1\oplus z_2}{\theta_{1,\vec{j}_1}\boxstar_l\theta_{2,\vec{j}_2}} \otimes \ket{0^{n_1-k_1}}\ket{0^{n_2-k_2}} \nonumber \\
  &=P_l\cdot \cnot\cdot
  \qstate{z_1}{\theta_{1,\vec{j}_1}} \otimes \ket{0^{n_1-k_1}}\otimes\qstate{z_2}{\theta_{2,\vec{j}_2}} \otimes \ket{0^{n_2-k_2}}
\end{align}
and
\begin{equation}
  \tilde{\kappa}_{\vec{j}_1,\vec{j}_2,l}
  \ket{\xi_{\vec{j}_1,\vec{j}_2,l}} = P_l\cdot \cnot\cdot 
  \ket{\tilde{\varphi}_{1,\vec{j}_1}}\ket{\tilde{\varphi}_{2,\vec{j}_2}}
\end{equation}
to obtain the upper bound 
\begin{align}
  \sum\limits_{\vec{j}_1,\vec{j}_2,l} \norm{ 
  \gamma_{1,\vec{j}_1}^2\gamma_{2,\vec{j}_2}^2\kappa_{\vec{j}_1,\vec{j}_2,l} \qstate{z_1}{\theta_{1,\vec{j}_1}} \ket{0^{n_1-k_1}}\otimes\qstate{z_2}{\theta_{2,\vec{j}_2}} \ket{0^{n_2-k_2}}
  - \tilde{\gamma}_{1,\vec{j}_1}^2\tilde{\gamma}_{2,\vec{j}_2}^2\tilde{\kappa}_{\vec{j}_1,\vec{j}_2,l} \ket{\tilde{\varphi}_{1,\vec{j}_1}}\ket{\tilde{\varphi}_{2,\vec{j}_2}}
  } 
\end{align}
because the operator norm of $P_l\cnot$ is $1$. 
By the triangle inquality, this can be further upper bounded by 
\begin{align}
  &\sum\limits_{\vec{j}_1,\vec{j}_2,l} |\tilde{\kappa}_{\vec{j}_1,\vec{j}_2,l}| \norm{ 
  \gamma_{1,\vec{j}_1}^2\gamma_{2,\vec{j}_2}^2 \qstate{z_1}{\theta_{1,\vec{j}_1}} \otimes \ket{0^{n_1-k_1}}\otimes\qstate{z_2}{\theta_{2,\vec{j}_2}} \otimes \ket{0^{n_2-k_2}}
  - \tilde{\gamma}_{1,\vec{j}_1}^2\tilde{\gamma}_{2,\vec{j}_2}^2 \ket{\tilde{\varphi}_{1,\vec{j}_1}}\ket{\tilde{\varphi}_{2,\vec{j}_2}}
  } \nonumber \\
  &\phantom{\leq}+ \sum\limits_{\vec{j}_1,\vec{j}_2,l} \gamma_{1,\vec{j}_1}^2\gamma_{2,\vec{j}_2}^2 \cdot | \kappa_{\vec{j}_1,\vec{j}_2,l} - \tilde{\kappa}_{\vec{j}_1,\vec{j}_2,l} |\\
\label{eq:proof51_prop13_eq2}
  &\leq 2\sum\limits_{\vec{j}_1,\vec{j}_2} \norm{ 
  \gamma_{1,\vec{j}_1}^2\gamma_{2,\vec{j}_2}^2 \qstate{z_1}{\theta_{1,\vec{j}_1}} \otimes \ket{0^{n_1-k_1}}\otimes\qstate{z_2}{\theta_{2,\vec{j}_2}} \otimes \ket{0^{n_2-k_2}}
  - \tilde{\gamma}_{1,\vec{j}_1}^2\tilde{\gamma}_{2,\vec{j}_2}^2 \ket{\tilde{\varphi}_{1,\vec{j}_1}}\ket{\tilde{\varphi}_{2,\vec{j}_2}}
  } \nonumber \\
  &\phantom{\leq}+ 2\sum\limits_{\vec{j}_1,\vec{j}_2} \gamma_{1,\vec{j}_1}^2\gamma_{2,\vec{j}_2}^2 \cdot \left( \norm{ \ket{\tilde{\varphi}_{1,\vec{j}_1}} - \qstate{z_1}{\theta_{1,\vec{j}_1}} \otimes \ket{\vec{j}_1} } + \norm{ \ket{\tilde{\varphi}_{2,\vec{j}_2}} - \qstate{z_2}{\theta_{2,\vec{j}_2}} \otimes \ket{\vec{j}_2} }\,,
\right)
\end{align}
where we used the previously shown bound on the difference
\begin{align}
  ||\kappa_{\vec{j}_1,\vec{j}_2,l}|-|\tilde{\kappa}_{\vec{j}_1,\vec{j}_2,l}|| = |\kappa_{\vec{j}_1,\vec{j}_2,l}-\tilde{\kappa}_{\vec{j}_1,\vec{j}_2,l}| \, .
\end{align}
By the induction step, the second term in~\cref{eq:proof51_prop13_eq2} is less than $2^{n+1/2}\pi\sqrt{\delta}(2\cdot 26^{n_1} + 2\cdot 26^{n_2})$. The first term can be bounded by
\begin{align}
  \leq 2\sum\limits_{\vec{j}_1,\vec{j}_2} \gamma_{1,\vec{j}_1}^2\gamma_{2,\vec{j}_2}^2 
  \norm { \qstate{z_1}{\theta_{1,\vec{j}_1}} \otimes \ket{0^{n_1-k_1}}\otimes\qstate{z_2}{\theta_{2,\vec{j}_2}} \otimes \ket{0^{n_2-k_2}}
  - \ket{\tilde{\varphi}_{1,\vec{j}_1}}\ket{\tilde{\varphi}_{2,\vec{j}_2}}
  } \nonumber \\
  + 2\sum\limits_{\vec{j}_1,\vec{j}_2} | \gamma_{1,\vec{j}_1}^2\gamma_{2,\vec{j}_2}^2 - \tilde{\gamma}_{1,\vec{j}_1}^2\tilde{\gamma}_{2,\vec{j}_2}^2 |
\end{align}
\begin{align}
  \leq 2\sum\limits_{\vec{j}_1} \gamma_{1,\vec{j}_1}^2 \norm{ \qstate{z_1}{\theta_{1,\vec{j}_1}} \otimes \ket{0^{n_1-k_1}} - \ket{\tilde{\varphi}_{1,\vec{j}_1}} }
  +    2\sum\limits_{\vec{j}_2} \gamma_{2,\vec{j}_2}^2 \norm{ \qstate{z_2}{\theta_{2,\vec{j}_2}} \otimes \ket{0^{n_2-k_2}} - \ket{\tilde{\varphi}_{2,\vec{j}_2}} } \nonumber \\
  +    2\sum\limits_{\vec{j}_1} |\gamma_{1,\vec{j}_1}^2 - \tilde{\gamma}_{1,\vec{j}_1}^2|
  +    2\sum\limits_{\vec{j}_2} |\gamma_{2,\vec{j}_2}^2 - \tilde{\gamma}_{2,\vec{j}_2}^2|
\end{align}
\begin{align}
  \leq 2^{n+1/2}\pi\sqrt{\delta} \left( 2\cdot 26^{n_1} + 2\cdot 26^{n_2} + \frac{2}{3}26^{n_1} + \frac{2}{3}26^{n_2} \right)\,.
\end{align}
To summarize, we have found that
\begin{align}
  &\sum\limits_{\vec{j}_1,\vec{j}_2,l}{ \gamma_{1,\vec{j}_1}^2\gamma_{2,\vec{j}_2}^2\kappa_{\vec{j}_1,\vec{j}_2,l}^2 \norm{ \qstate{z_1\oplus z_2}{\theta_{1,\vec{j}_1}\boxstar_l\theta_{2,\vec{j}_2}} \otimes \ket{\vec{j}_1}\ket{\vec{j}_2}\otimes\ket{l} - \ket{\xi_{\vec{j}_1,\vec{j}_2,l}}\otimes\ket{l}} }\\
  &\leq 2^{n+1/2}\pi\sqrt{\delta} \left( \frac{13}{3}26^{n_1} + \frac{13}{3}26^{n_1} + 2\cdot 26^{n_1} + 2\cdot 26^{n_2} + 2\cdot 26^{n_1} + 2\cdot 26^{n_2} + \frac{2}{3}26^{n_1} + \frac{2}{3}26^{n_2}\right)\\
  &\leq 2^{n+1/2}\pi\sqrt{\delta}\left( 9\cdot 26^{n_1} + 9\cdot 26^{n_2} \right)\\
  &\leq 2^{n+1/2}\pi\sqrt{\delta}\left( 13\cdot 26^{n_3-1} + 13\cdot 26^{n_3-1} \right)\\
  &= 2^{n+1/2}\pi\sqrt{\delta}\cdot 26^{n_3} \,.
\end{align}

\subsection{Equality node}
  In the original BPQM description, the equality node consists of a uniformly-controlled $U_{\ostar}$ gate which transforms our state to
\begin{equation}
  \sum\limits_{\vec{j}_1\in\{0,1\}^{n_1-1},\vec{j}_2\in\{0,1\}^{n_2-1}}{ \gamma_{1,\vec{j}_1}\gamma_{2,\vec{j}_2} \qstate{z}{\theta_{1,\vec{j}_1}\ostar\theta_{2,\vec{j}_2}}_{D_3} \otimes \ket{0^{n_3-k_3}}_{Z_3} \otimes \left(\ket{\vec{j}_1}\ket{\vec{j}_2}\right)_{A_3} } \,,
\end{equation}
where $n_3:=n_1+n_2$, $k_3:=k_1+k_2-1$ and $D_3=D_1$, $Z_3=Z_1Z_2$, $A_3=A_1A_2D_2$. Here we used the fact that $z_1=z_2=:z$, which is a direct consequence of claim 1 in the proof of~\cref{lem:purestate_evolution}. In the message-passing variant, the resulting state is 
  \begin{align}
    \sum\limits_{\vec{j}_1\in\{0,1\}^{n_1-1},\vec{j}_2\in\{0,1\}^{n_2-1}}
    \tilde{\gamma}_{1,\vec{j}_1}\tilde{\gamma}_{2,\vec{j}_2}
    \left(V_{\vec{j}_1,\vec{j}_2} \ket{\tilde{\varphi}_{1,\vec{j}_1}}\ket{\tilde{\varphi}_{2,\vec{j}_2}} \right)_{D_3Z_3}
    \otimes \left( \ket{\vec{j}_1}\ket{\vec{j}_2} \right)_{A_3}
    \otimes \ket{Q(\tilde{c}_{1,\vec{j}_1}\cdot\tilde{c}_{2,\vec{j}_2})}_{C_3} \nonumber \\
    \otimes \left( \ket{\tilde{s}_{1,\vec{j}_1}}\ket{\tilde{s}_{2,\vec{j}_2}}\anglestate{\tilde{c}_{1,\vec{j}_1}} \anglestate{\tilde{c}_{2,\vec{j}_2}} \right)_{S_3}
  \end{align}
  where $V_{\vec{j}_1,\vec{j}_2}$ describes the quantum circuit from the right-hand-side of~\cref{fig:uc_ustar_identity} acting on the system $D_1D_2$ with the $R_y$ rotation angles $\quantalphabeta{B}{\alpha(\tilde{\theta}_{1,\vec{j}_1},\tilde{\theta}_{2,\vec{j}_2})}$ and $\quantalphabeta{B}{\beta(\tilde{\theta}_{1,\vec{j}_1},\tilde{\theta}_{2,\vec{j}_2})}$.
  We now show that the three statements of the lemma hold in this case as well.

\subsubsection{Property 1}
  In the first case we have
  \begin{align}
    &\sum\limits_{\vec{j}_1,\vec{j}_2,k} \gamma_{1,\vec{j}_1}^2\gamma_{2,\vec{j}_2}^2 \abs{\cos(\theta_{1,\vec{j}_1}\ostar\theta_{2,\vec{j}_2}) - \quant{\tilde{c}_{1,\vec{j}_1}\cdot \tilde{c}_{2,\vec{j}_2}}}\\
    &\leq \delta + \sum\limits_{\vec{j}_1,\vec{j}_2,k} \gamma_{1,\vec{j}_1}^2\gamma_{2,\vec{j}_2}^2 \abs{\cos\theta_{1,\vec{j}_1} \cos\theta_{2,\vec{j}_2} - \tilde{c}_{1,\vec{j}_1}\cdot \tilde{c}_{2,\vec{j}_2}}\\
    &\leq \delta + \sum\limits_{\vec{j}_1} \gamma_{1,\vec{j}_1}^2 \abs{\cos\theta_{1,\vec{j}_1} - \tilde{c}_{1,\vec{j}_1}} +
    \sum\limits_{\vec{j}_2} \gamma_{2,\vec{j}_2}^2 \abs{\cos\theta_{2,\vec{j}_2} - \tilde{c}_{2,\vec{j}_2} }\\
    &\leq \delta + (2^{n_1+1}-3)\delta + (2^{n_2+1}-3)\delta \\
    &= \delta (2^{n_1+1} + 2^{n_2+1} - 5) \\
    &\leq \delta (2^{n_3} + 2^{n_3} - 3 ) \\
    &= (2^{{n_3}+1} - 3)\delta\,.
  \end{align}

\subsubsection{Property 2}
The second case is straightforward:  
  \begin{align}
    \sum\limits_{\vec{j}_1,\vec{j}_2}{| \gamma_{1,\vec{j}_1}^2\gamma_{2,\vec{j}_2}^2 -  \tilde{\gamma}_{1,\vec{j}_1}^2\tilde{\gamma}_{2,\vec{j}_2}^2 |} 
    &\leq
    \sum\limits_{\vec{j}_1}{| \gamma_{1,\vec{j}_1}^2 -  \tilde{\gamma}_{1,\vec{j}_1}^2 |} +
    \sum\limits_{\vec{j}_2}{| \gamma_{2,\vec{j}_2}^2 -  \tilde{\gamma}_{2,\vec{j}_2}^2 |}\\
  &  \leq 2^{n+1/2}\pi\sqrt{\delta}\tfrac{1}{3}(26^{n_1} + 26^{n_2})\\
&    \leq 2^{n+1/2}\pi\sqrt{\delta}\tfrac{1}{3}(26^{n_3-1} + 26^{n_3-1})\\
  &    \leq 2^{n+1/2}\pi\sqrt{\delta}\tfrac{1}{3}26^{n_3}\,.
  \end{align}

\subsubsection{Property 3}
Finally, for the third statement we use the triangle inequality to bound the error by a sum of three terms, which we then address individually. 
   \begin{align}
     \sum\limits_{\vec{j}_1,\vec{j}_2,k} \gamma_{1,\vec{j}_1}^2\gamma_{2,\vec{j}_2}^2 \norm{ U_{\ostar}(\theta_{1,\vec{j}_1},\theta_{2,\vec{j}_2})\qstate{z}{\theta_{1,\vec{j}_1}}\ket{0^{n_1-k_1}}\otimes\qstate{z}{\theta_{2,\vec{j}_2}}\ket{0^{n_2-k_2}} - V_{\vec{j}_1,\vec{j}_2}\ket{\tilde{\varphi}_{1,\vec{j}_1}}\ket{\tilde{\varphi}_{2,\vec{j}_2}} }
   \end{align}
   \begin{align}\label{eq:proof51_prop23_3errors}
     \leq
     \sum\limits_{\vec{j}_1,\vec{j}_2} \gamma_{1,\vec{j}_1}^2\gamma_{2,\vec{j}_2}^2 \norm{V_{\vec{j}_1,\vec{j}_2}\qstate{z}{\theta_{1,\vec{j}_1}}\ket{0^{n_1-k_1}}\otimes\qstate{z}{\theta_{2,\vec{j}_2}}\ket{0^{n_2-k_2}} - V_{\vec{j}_1,\vec{j}_2}\ket{\tilde{\varphi}_{1,\vec{j}_1}}\ket{\tilde{\varphi}_{2,\vec{j}_2}}} \nonumber \\
     + \sum\limits_{\vec{j}_1,\vec{j}_2} \gamma_{1,\vec{j}_1}^2\gamma_{2,\vec{j}_2}^2 \Vert U_{\ostar}(\theta_{1,\vec{j}_1},\theta_{2,\vec{j}_2})\qstate{z}{\theta_{1,\vec{j}_1}}\ket{0^{n_1-k_1}}\otimes\qstate{z}{\theta_{2,\vec{j}_2}}\ket{0^{n_2-k_2}} \nonumber \\
     - U_{\ostar}(\tilde{\theta}_{1,\vec{j}_1},\tilde{\theta}_{2,\vec{j}_2})\qstate{z}{\theta_{1,\vec{j}_1}}\ket{0^{n_2-k_2}}\otimes\qstate{z}{\theta_{2,\vec{j}_2}}\ket{0^{n_2-k_2}} \Vert + \nonumber \\
     + \sum\limits_{\vec{j}_1,\vec{j}_2} \gamma_{1,\vec{j}_1}^2\gamma_{2,\vec{j}_2}^2 \Vert U_{\ostar}(\tilde{\theta}_{1,\vec{j}_1},\tilde{\theta}_{2,\vec{j}_2})\qstate{z}{\theta_{1,\vec{j}_1}}\ket{0^{n_1-k_1}}\otimes\qstate{z}{\theta_{2,\vec{j}_2}}\ket{0^{n_2-k_2}} \nonumber \\
     - V_{\vec{j}_1,\vec{j}_2}\qstate{z}{\theta_{1,\vec{j}_1}}\ket{0^{n_1-k_1}}\otimes\qstate{z}{\theta_{2,\vec{j}_2}}\ket{0^{n_2-k_2}} \Vert 
    \end{align}
    where we introduced $\tilde{\theta}_{i,\vec{j}_i} := \arccos \tilde{c}_{\theta_i,\vec{j}_i}$. The first term of~\cref{eq:proof51_prop23_3errors} is simply the error of the incoming state. By removing $V_{\vec{j}_1,\vec{j}_2}$ from the expression, we see that it can be bounded by $2^{n+1/2}\pi\sqrt{\delta}(26^{n_1}+26^{n_2})$. 

    The second term of~\cref{eq:proof51_prop23_3errors} is the error caused by the fact that we are executing the wrong equality node unitary, due to the errornous angle stored in the angle register. We make use of following property for any set of angles $\phi_1,\phi_2,\tilde{\phi}_1,\tilde{\phi}_2$ and $w\in\{0,1\}$:
    \begin{align}
      &\norm{U_{\ostar}(\phi_1,\phi_2)\qstate{w}{\phi_1}\qstate{w}{\phi_2} - U_{\ostar}(\tilde{\phi}_1,\tilde{\phi}_2)\qstate{w}{\phi_1}\qstate{w}{\phi_2}}\\
      &\leq \norm{U_{\ostar}(\phi_1,\phi_2)\qstate{w}{\phi_1}\qstate{w}{\phi_2} - U_{\ostar}(\tilde{\phi}_1,\tilde{\phi}_2)\qstate{w}{\tilde{\phi}_1}\qstate{w}{\tilde{\phi}_2}} \nonumber \\
      &+ \norm{U_{\ostar}(\tilde{\phi}_1,\tilde{\phi}_2)\qstate{w}{\tilde{\phi}_1}\qstate{w}{\tilde{\phi}_2} - U_{\ostar}(\tilde{\phi}_1,\tilde{\phi}_2)\qstate{w}{\phi_1}\qstate{w}{\phi_2}}\\
      & \leq \norm{\qstate{w}{\phi_1\ostar\phi_2} - \qstate{w}{\tilde{\phi}_1\ostar\tilde{\phi}_2}} + \norm{\qstate{w}{\phi_1}\qstate{w}{\phi_2} - \qstate{w}{\tilde{\phi_1}}\qstate{w}{\tilde{\phi_2}}}\\
      & \leq \norm{\qstate{w}{\phi_1\ostar\phi_2} - \qstate{w}{\tilde{\phi}_1\ostar\tilde{\phi}_2}} + \norm{\qstate{w}{\phi_1} - \qstate{w}{\tilde{\phi_1}}} + \norm{\qstate{w}{\phi_2} - \qstate{w}{\tilde{\phi_2}}}\\
      &\leq \tfrac{1}{2}\left( |\phi_1\ostar\phi_2 - \tilde{\phi}_1\ostar\tilde{\phi}_2| + |\phi_1-\tilde{\phi}_1| + |\phi_2-\tilde{\phi}_2| \right)\,,
    \end{align}
    where we used $\norm{\qstate{w}{\alpha}-\qstate{w}{\beta}} = 2|\sin\frac{\alpha-\beta}{4}| \leq \frac{|\alpha-\beta|}{2}$. Therefore, the second term can be bounded by
\begin{align}
  &\frac{1}{2}\sum\limits_{\vec{j}_1,\vec{j}_2} \gamma_{1,\vec{j}_1}^2\gamma_{2,\vec{j}_2}^2 \left(
  |\theta_{1,\vec{j}_1}\ostar\theta_{2,\vec{j}_2} - \arccos(\tilde{c}_{1,\vec{j}_1}\cdot\tilde{c}_{2,\vec{j}_2})| + 
  |\theta_{1,\vec{j}_1} - \arccos\tilde{c}_{1,\vec{j}_1}| + |\theta_{2,\vec{j}_2} - \arccos\tilde{c}_{2,\vec{j}_2}|
  \right)\\
  &\leq \frac{\pi}{2\sqrt{2}}\sum\limits_{\vec{j}_1,\vec{j}_2} \gamma_{1,\vec{j}_1}^2\gamma_{2,\vec{j}_2}^2 \left(
  \sqrt{|\cos\theta_{1,\vec{j}_1}\cdot\cos\theta_{2,\vec{j}_2} - \tilde{c}_{1,\vec{j}_1}\cdot\tilde{c}_{2,\vec{j}_2}|} +
  \sqrt{|\cos\theta_{1,\vec{j}_1} - \tilde{c}_{1,\vec{j}_1}|} + \sqrt{|\cos\theta_{2,\vec{j}_2} - \tilde{c}_{2,\vec{j}_2}|}
  \right)\\
  &\leq \frac{\pi}{\sqrt{2}}\sum\limits_{\vec{j}_1,\vec{j}_2} \gamma_{1,\vec{j}_1}^2\gamma_{2,\vec{j}_2}^2 \left(
  \sqrt{|\cos\theta_{1,\vec{j}_1} - \tilde{c}_{1,\vec{j}_1}|} + \sqrt{|\cos\theta_{2,\vec{j}_2} - \tilde{c}_{2,\vec{j}_2}|}
  \right)\\
  &\leq \frac{\pi}{\sqrt{2}}\sum\limits_{\vec{j}_1,\vec{j}_2} \sqrt{\gamma_{1,\vec{j}_1}^2\gamma_{2,\vec{j}_2}^2} \left(
  \sqrt{|\cos\theta_{1,\vec{j}_1} - \tilde{c}_{1,\vec{j}_1}|} + \sqrt{|\cos\theta_{2,\vec{j}_2} - \tilde{c}_{2,\vec{j}_2}|}
  \right)\,,
\end{align}
    which follows from the Hölder-continuity of the $\arccos$ function $|\arccos x-\arccos y|\leq \frac{\pi}{\sqrt{2}}\sqrt{|x-y|}$ as well as the subadditivity of the square root.
    A direct consequence of the Cauchy-Schwarz inequality is the statement $\norm{\vec{z}}_1 \leq \sqrt{M}\norm{\vec{z}}_2$ for $\vec{z}\in\mathbb{R}^M$, which allows us to upper bound the second term by
    \begin{align}
      &2^{n/2}\frac{\pi}{\sqrt{2}}\left (
      \sqrt{\sum\limits_{\vec{j}_1} \gamma_{1,\vec{j}_1}^2 |\cos\theta_{1,\vec{j}_1} - \tilde{c}_{1,\vec{j}_1}|} + 
      \sqrt{\sum\limits_{\vec{j}_2} \gamma_{2,\vec{j}_2}^2 |\cos\theta_{2,\vec{j}_2} - \tilde{c}_{2,\vec{j}_2}|} 
      \right)\\
      &\leq 2^{n/2}\frac{\pi}{\sqrt{2}} \sqrt{\delta} \left( \sqrt{2^{n_1+1}-3} + \sqrt{2^{n_2+1}-3} \right)\,.
    \end{align}

    Finally, the third term of~\cref{eq:proof51_prop23_3errors} entails the fact, that the unitary $V_{\vec{j}_1,\vec{j}_2}$ is an imperfect implementation of the unitary $U_{\ostar}(\tilde{\theta}_{1,\vec{j}_1},\tilde{\theta}_{2,\vec{j}_2})$ due to the discretization error in the computation of the angles $\alpha$ and $\beta$ in the circuit of~\cref{fig:uc_ustar}. 
    Since we know that the discretization error in the computation of $\alpha,\beta$ is bounded by $\pi\delta$, this implies that the $R_y$ gates implement a rotation by an angle that differs by at most $\pi\delta$ from the desired value. 
    By using
    \begin{equation}
      \norm{R_y(x+\Delta)-R_y(x)}_{\text{op}} \leq \frac{|\Delta|}{2}
    \end{equation}
    for the operator norm $\norm{\cdot}_{\text{op}}$ and any real numbers $x,\Delta$, as well as the fact that 
    \begin{equation}
      \norm{U_1U_2\dots U_s - \tilde{U}_1\tilde{U}_2\dots\tilde{U}_s}_{\text{op}} \leq \norm{U_1-\tilde{U}_1}_{\text{op}} + \norm{U_2-\tilde{U}_2}_{\text{op}} + \dots + \norm{U_s-\tilde{U}_s}_{\text{op}}
    \end{equation}
    for any sequences of unitaries $U_i,\tilde{U}_i$ of length $s$, we can upper bound the third term by
    \begin{equation}
       \sum\limits_{\vec{j}_1,\vec{j}_2} \gamma_{1,\vec{j}_1}^2\gamma_{2,\vec{j}_2}^2
      (\tfrac{1}{2}\pi\delta + \tfrac{1}{2}\pi\delta)
      = \pi\delta \, .
    \end{equation}

    To summarize, we have shown that
    \begin{align}
     \sum\limits_{\vec{j}_1,\vec{j}_2,k}&  \gamma_{1,\vec{j}_1}^2\gamma_{2,\vec{j}_2}^2 \norm{ U_{\ostar}(\theta_{1,\vec{j}_1},\theta_{2,\vec{j}_2})\qstate{z}{\theta_{1,\vec{j}_1}}\ket{0^{n_1-k_1}}\otimes\qstate{z}{\theta_{2,\vec{j}_2}}\ket{0^{n_2-k_2}} - V_{\vec{j}_1,\vec{j}_2}\ket{\tilde{\varphi}_{1,\vec{j}_1}}\ket{\tilde{\varphi}_{2,\vec{j}_2}} }\\
     & \leq 2^{n+1/2}\pi\sqrt{\delta}(26^{n_1}+26^{n_2}) + 2^{(n-1)/2}\pi\sqrt{\delta}\left( \sqrt{2^{n_1+1}-3} + \sqrt{2^{n_2+1}-3} \right) + \pi\delta\\
     & \leq 2^{n+1/2}\pi\sqrt{\delta}(26^{n_1}+26^{n_2}) + 2^{(n-1)/2}\pi\sqrt{\delta}\left( \sqrt{2^{n}} + \sqrt{2^{n}} \right) + \pi\sqrt{\delta}\\
     & \leq 2^{n+1/2}\pi\sqrt{\delta} \left(26^{n_3-1}+26^{n_3-1} + 1 + 1\right) \\
     &\leq 2^{n+1/2}\pi\sqrt{\delta} 26^{n_3} \, .
    \end{align}

\section{Proof of \texorpdfstring{\cref{thm:error_probab}}{Theorem 5.2}}
\label{app:proof_complete_codeword}
Directly comparing BPQM with message-passing BPQM is impeded by the somewhat technical difficulty  that the two corresponding unitaries operate on a different number of qubits, since the message-passing BPQM requires additional ancilla $B$-qubit registers. 
For reasons purely related to the analysis of message-passing BPQM we now introduce a third variant of BPQM, which we call \emph{extended BPQM}.
Extended BPQM keeps track of the angles of the involved qubits in $B$-qubit quantum registers as in message-passing BPQM, and therefore it operates on the same number of qubits as message-passing BPQM.
The only difference to the message-passing variant is that the equality node unitaries $U_{\ostar}$ do not make use of the angle registers, but rather condition on the check node ancilla qubits as in  BPQM. 
Since the angle registers never influence the data registers, the decoding is equivalent to BPQM and is therefore optimal. 
We denote the BPQM unitaries to decode the codeword bit $X_r$ by $\vmp{r}$ for the message-passing algorithm and $\vext{r}$ for the extended algorithm.

\Cref{lem:discretization_error} makes a statement about a weighted average of vector norms. We translate this into a statement about the distance between the final states of message-passing BPQM and extended BPQM applied on some state $\ket{\Psi_{\vec{x}}},\vec{x}\in\mathcal{C}$ by using the Cauchy-Schwarz inequality and the triangle inequality:
\begin{align}
	&\norm{ \vext{r}\ket{\Psi_{\vec{x}}}\ket{\vec{0}} - \vmp{r}\ket{\Psi_{\vec{x}}}\ket{\vec{0}} }
  = \norm{\sum\limits_{\vec{j}} \gamma_{\vec{j}}\qstate{z}{\theta_{\vec{j}}}\ket{0^{n-k}}\ket{\vec{j}}\ket{\tilde{\theta}_{\vec{j}}}\ket{\tilde{s}_{\vec{j}}} - \sum\limits_{\vec{j}} \tilde{\gamma}_{\vec{j}}\ket{\tilde{\varphi}_{\vec{j}}}\ket{\vec{j}}\ket{\tilde{\theta}_{\vec{j}}}\ket{\tilde{s}_{\vec{j}}} }\\
&
  \leq \norm{\sum\limits_{\vec{j}} \gamma_{\vec{j}}\qstate{z}{\theta_{\vec{j}}}\ket{0^{n-k}}\ket{\vec{j}}\ket{\tilde{\theta}_{\vec{j}}}\ket{\tilde{s}_{\vec{j}}} - \sum\limits_{\vec{j}} \gamma_{\vec{j}}\ket{\tilde{\varphi}_{\vec{j}}}\ket{\vec{j}}\ket{\tilde{\theta}_{\vec{j}}}\ket{\tilde{s}_{\vec{j}}} } \nonumber \\
  & + \norm{ \sum\limits_{\vec{j}} \gamma_{\vec{j}}\ket{\tilde{\varphi}_{\vec{j}}}\ket{\vec{j}}\ket{\tilde{\theta}_{\vec{j}}}\ket{\tilde{s}_{\vec{j}}} - \sum\limits_{\vec{j}} \tilde{\gamma}_{\vec{j}}\ket{\tilde{\varphi}_{\vec{j}}}\ket{\vec{j}}\ket{\tilde{\theta}_{\vec{j}}}\ket{\tilde{s}_{\vec{j}}} }\\
&  \leq \sum\limits_{\vec{j}} |\gamma_{\vec{j}}| \norm{ \qstate{z}{\theta_{\vec{j}}}\ket{0^{n-k}} - \ket{\tilde{\varphi}_{\vec{j}}} }
  + \sum\limits_{\vec{j}} |\gamma_{\vec{j}} - \tilde{\gamma}_{\vec{j}}|\\
  &  \leq \sqrt{2^n} \sqrt{ \sum\limits_{\vec{j}} \gamma_{\vec{j}}^2 \norm{ \qstate{z}{\theta_{\vec{j}}}\ket{0^{n-k}} - \ket{\tilde{\varphi}_{\vec{j}}} }^2 }
  + \sum\limits_{\vec{j}} \sqrt{|\gamma_{\vec{j}}^2 - \tilde{\gamma}_{\vec{j}}^2|}\\
&  \leq \sqrt{2^n} \sqrt{ \sum\limits_{\vec{j}} \gamma_{\vec{j}}^2 \norm{ \qstate{z}{\theta_{\vec{j}}}\ket{0^{n-k}} - \ket{\tilde{\varphi}_{\vec{j}}} } }
  + \sqrt{2^n}\sqrt{ \sum\limits_{\vec{j}} |\gamma_{\vec{j}}^2 - \tilde{\gamma}_{\vec{j}}^2| }\,,
\end{align}
where $\vec{0}$ captures all the ancilla qubits required to perform the message-passing/extended BPQM variant. 
The fact that $|\gamma_{\vec{j}}-\tilde{\gamma}_{\vec{j}}|\leq \sqrt{|\gamma_{\vec{j}}^2-\tilde{\gamma}_{\vec{j}}^2|}$ follows from the fact that $\gamma_{\vec{j}}$ and $\tilde{\gamma}_{\vec{j}}$ have the same sign, which in turn follows directly from the definition~\cref{eq:def_kappa} of $\tilde{\kappa}_{\vec{j}_1,\vec{j}_2,l}$ in~\cref{lem:discretization_error} that guarantees it to have the same sign as $\kappa_{\vec{j}_1,\vec{j}_2,l}$.
Appealing to \cref{item:2,item:3} in \cref{lem:discretization_error} and using $\delta\leq 2^{-B}$ gives
\begin{align}
 \norm{ \vext{r}\ket{\Psi_{\vec{x}}}\ket{\vec{0}} - \vmp{r}\ket{\Psi_{\vec{x}}}\ket{\vec{0}} }
 % &\leq \sqrt{2^n}2^{n+\tfrac12}\pi\sqrt{\delta}26^n + \sqrt{2^n}2^{n+\tfrac12}\pi\sqrt{\delta}\frac{1}{3}26^n\\ 
  &\leq \frac{2^{5/4}\sqrt{\pi}}{\sqrt{3}} 2^{n\cdot (1+\frac{1}{2}\frac{\log 26}{\log 2}) - \frac{1}{4}B} \nonumber\\
  &=\alpha 2^{\beta n - \gamma B}\label{eq:vextmpbound}\,,
\end{align}
 defining  $\alpha:=\frac{2^{5/4}\sqrt{\pi}}{\sqrt{3}}$, $\beta:=1+\frac{1}{2}\frac{\log 26}{\log 2}$, $\gamma:=1/4$.

Using this result, we can now analyze how message-passing BPQM performs when decoding the complete codeword.
More specifically, we consider that we decode in sequence the codeword bits $X_1,\dots,X_k$.
We want to determine how much lower the probability of decoding all codeword bits correctly is compared to the extended variant of BPQM. For that purpose, we consider the non-normalized post-measurement state conditioned on the $k$ measurement outcomes all being correct.
\footnote{The postulates of quantum mechanics dictate that a post-measurement state (conditioned on some outcome) is obtained by applying a projector on our state and then renormalizing it with the corresponding probability of the measurement outcome. By explicitly not performing this normalization, the terms involved in our analysis remain simpler.}
More precisely, for $P_{j}=(H\ketbra{x_{j}}H\otimes \id)$ the projection operator corresponding to correctly measuring the $j$th codeword bit, we inductively define
\begin{equation}
  \ketmp{\phi_0} := \ket{\Psi_{\vec{x}}} \otimes \ket{\vec{0}}, \, \ketext{\phi_0} := \ketmp{\phi_0}
\end{equation}
\begin{equation}
  \ketmp{\phi_{j}} := (\vmp{j})^{\dagger} P_{j} (\vmp{j}) \ketmp{\phi_{j-1}}
\end{equation}
\begin{equation}
  \ketext{\phi_{j}} := (\vext{j})^{\dagger} P_{j} (\vext{j}) \ketext{\phi_{j-1}}
\end{equation}
such that $\ketmp{\phi_j}$ and $\ketext{\phi_j}$ denote the state of the message-passing algorithm, respectively extended algorithm, after the $j$-th codeword bit has been correctly decoded, where $j=1,\dots,k$.
The probability of correctly decoding the complete codeword is therefore given by $\norm{\ketmp{\phi_k}}$ for the message-passing algorithm, respectively by $\norm{\ketext{\phi_k}}$ for the extended variant of BPQM.
Surely the latter is larger than the former, since the BPQM algorithm is itself optimal. 

Let $\mathcal{H}$ denote the space spanned by the states $\{\ket{\Psi_{\vec{x}}}\otimes\ket{\vec{0}} | \vec{x}\in\mathcal{C}\}$.
As seen in~\cref{sec:optimality}, this is also the space spanned by the orthonormal vectors of the PGM $\{\ket{f_{\vec{x}}}\otimes\ket{\vec{0}} | \vec{x}\in\mathcal{C}\}$.
The central part of the proof is an analysis how the discretization errors propagate throughout the decoding process. 
We capture the result in the following claim:
\begin{claim}
For all $j\in\{0,\dots,k\}$, \emph{$\ketext{\phi_j}\in\mathcal{H}$} and \emph{$\norm{\ketext{\phi_j} - \ketmp{\phi_j}}\leq j\cdot 2^{k/2+1}\alpha2^{\beta n-\gamma B}$}.
\end{claim}
The reverse triangle inequality immediately gives
 $ \norm{\ketext{\phi_j}}-\norm{\ketmp{\phi_j}}\leq \norm{\ketext{\phi_j} - \ketmp{\phi_j}}$,
  %\geq | \norm{\ket{\phi_j^e}} - \norm{\ket{\phi_j^h} - \ket{\phi_j^{mp}}} |
  %\geq \norm{\ket{\phi_j^e}} - \norm{\ket{\phi_j^h} - \ket{\phi_j^{mp}}}
which by the claim implies 
\begin{equation}
  p^{(\text{ext})}-p^{(\text{mp})}
  \leq k 2^{k/2+1}\alpha2^{\beta n-\gamma B}
  \leq n 2^{n/2+1}\alpha2^{\beta n-\gamma B}
  = 2n\alpha 2^{(\beta + 1/2)n-\gamma B}
  \, .
\end{equation}

It now only remains to show the validity of the claim. We make an inductive argument in $j$. Obviously for $j=0$ the statement holds, as $\ketext{\phi_0} = \ketmp{\phi_0}=\ket{\Psi_{\vec{x}}}\otimes\ket{\vec{0}}$.
We now show the induction step as follows. 
First,
\begin{align}
\norm{\ketext{\phi_{j+1}} - \ketmp{\phi_{j+1}}}
&=\norm{(\vext{j+1})^\dagger P_{j+1} \vext{j+1}\ketext{\phi_{j}} - (\vmp{j+1})^\dagger P_{j+1} \vmp{j+1}\ketmp{\phi_{j}}}\\
&\leq \norm{(\vext{j+1})^\dagger P_{j+1} \vext{j+1}\ketext{\phi_j} - (\vmp{j+1})^\dagger P_{j+1} \vext{j+1}\ketext{\phi_j}}\nonumber\\
&\quad+\norm{(\vmp{j+1})^\dagger P_{j+1} \vext{j+1}\ketext{\phi_j} - (\vmp{j+1})^\dagger P_{j+1} \vmp{j+1}\ketext{\phi_j}}\nonumber\\
&\quad+\norm{(\vmp{j+1})^\dagger P_{j+1} \vmp{j+1}\ketext{\phi_j} - (\vmp{j+1})^\dagger P_{j+1} \vmp{j+1}\ketmp{\phi_{j}} }\\
&\leq \norm{(\vext{j+1})^\dagger \ket{\varphi_{j+1}} - (\vmp{j+1})^\dagger \ket{\varphi_{j+1}}}\nonumber\\
&\quad+\norm{\vext{j+1}\ketext{\phi_j} - \vmp{j+1}\ketext{\phi_j}}+\norm{\ketext{\phi_j}-\ketmp{\phi_j}} \label{eq:claimf1_split}\,,
\end{align}
where the first inequality is the triangle inequality and the second follows from the fact that a projection cannot increase the norm. 
Additionally, we have defined the vectors
\begin{equation}
  \ket{\varphi_{j+1}} := P_{j+1} (\vext{j+1})\ketext{\phi_{j}}\,.
\end{equation} 

The third term of~\cref{eq:claimf1_split} can be bounded by $j\cdot 2\alpha2^{\beta n-\gamma B}(1+2^{k/2})$ by the induction hypothesis.

By the induction hypothesis, $\ketext{\phi_j}$ must lie in $\mathcal{H}$ and therefore one can find coefficients $a_{\vec{x}}$ such that
\begin{equation}
  \ketext{\phi_j} = \sum\limits_{\vec{x}\in\mathcal{C}} a_{\vec{x}} \ket{\Psi_{\vec{x}}}
\end{equation}
and $\sum_{\vec{x}}|a_{\vec{x}}|^2\leq 1$. By the triangle inequality we have 
\begin{align}
  & \norm{\vext{j+1}\ketext{\phi_j} - \vmp{j+1}\ketext{\phi_j}} \nonumber\\
  &\leq \sum\limits_{\vec{x}\in\mathcal{C}} |a_{\vec{x}}| \norm{\vext{j+1}\ket{\Psi_{\vec{x}}} - \vmp{j+1}\ket{\Psi_{\vec{x}}}} \,.
\end{align}
Bounding the norm in the summation by \cref{eq:vextmpbound} and invoking the Cauchy-Schwarz inequality to bound $\sum_{\vec{x}\in\mathcal{C}} |a_{\vec{x}}|$ by $\sqrt{2^k}$ gives
\begin{equation}
  \norm{\vext{j+1}\ketext{\phi_j} - \vmp{j+1}\ketext{\phi_j}} \leq 2^{k/2} \alpha 2^{\beta n -\gamma B}
\end{equation}

For the first term of~\cref{eq:claimf1_split} we can make a similar argument: $\ket{\varphi_{j+1}}$ can be expanded as
\begin{equation}
  \ket{\varphi_{j+1}} = \sum\limits_{\vec{x}\in\mathcal{C}} b_{\vec{x}} (\vext{j+1})\ket{\Psi_{\vec{x}}}
\end{equation}
for some coefficients $b_{\vec{x}}$. We thus get
\begin{align}
  &\norm{(\vext{j+1})^{\dagger}\ket{\varphi_{j+1}} - (\vmp{j+1})^{\dagger}\ket{\varphi_{j+1}} }\nonumber\\
  &\leq \sum\limits_{\vec{x}\in\mathcal{C}} |b_{\vec{x}}| \norm{(\vext{j+1})^{\dagger}(\vext{j+1}) \ket{\Psi_{\vec{x}}}\otimes\ket{\vec{0}} - (\vmp{j+1})^{\dagger}(\vext{j+1}) \ket{\Psi_{\vec{x}}}\otimes\ket{\vec{0}} }\\
  &= \sum\limits_{\vec{x}\in\mathcal{C}} |b_{\vec{x}}| \norm{(\vmp{j+1}) \ket{\Psi_{\vec{x}}}\otimes\ket{\vec{0}} - (\vext{j+1}) \ket{\Psi_{\vec{x}}}\otimes\ket{\vec{0}} } \\
  &\leq 2^{k/2} \alpha 2^{\beta n -\gamma B} \,.
\end{align}

Summarizing all three terms, we thus obtain
\begin{align}
  \norm{\ketext{\phi_{j+1}} - \ketmp{\phi_{j+1}}}
  & \leq j 2^{k/2 + 1}  \alpha 2^{\beta n-\gamma B} +  2 \cdot 2^{k/2} \alpha 2^{\beta n-\gamma B} \\
  & = (j+1) 2^{k/2 + 1}  \alpha 2^{\beta n-\gamma B} \, .
\end{align}
To complete the proof, we now argue that $\ketext{\phi_{j+1}}\in\mathcal{H}$.
This follows directly from $\ketext{\phi_{j}}\in\mathcal{H}$ and the observation made in~\cref{sec:optimality} that $(\vext{j})^{\dagger} P_{j} (\vext{j})$ realizes the rank-$2^{k-1}$ projector $\Pi_{x_r=m_r}$ introduced in~\cref{eq:def_pgm_projector}.

\section{Numerical simulation of discretization errors}\label{app:numerical_discretization_errors}
Here we consider a $(17,11)$ code, specifically chosen for its MPG structure, which contains multiple subsequent equality and check nodes.
The generator matrix and parity-check matrix in standard form are
\begin{equation}
G=\begin{pmatrix}
 1 & 0 & 0 & 0 & 0 & 0 & 1 & 0 & 1 & 0 & 0 & 0 & 0 & 0 & 1 & 0 & 1 \\
 0 & 1 & 0 & 0 & 1 & 0 & 1 & 0 & 1 & 0 & 0 & 0 & 0 & 0 & 0 & 0 & 0 \\
 0 & 0 & 1 & 0 & 1 & 0 & 1 & 0 & 1 & 0 & 0 & 0 & 0 & 0 & 0 & 0 & 0 \\
 0 & 0 & 0 & 1 & 1 & 0 & 0 & 0 & 0 & 0 & 0 & 0 & 0 & 0 & 0 & 0 & 0 \\
 0 & 0 & 0 & 0 & 0 & 1 & 1 & 0 & 0 & 0 & 0 & 0 & 0 & 0 & 0 & 0 & 0 \\
 0 & 0 & 0 & 0 & 0 & 0 & 0 & 1 & 1 & 0 & 0 & 0 & 0 & 0 & 0 & 0 & 0 \\
 0 & 0 & 0 & 0 & 0 & 0 & 0 & 0 & 0 & 1 & 0 & 0 & 1 & 0 & 1 & 0 & 1 \\
 0 & 0 & 0 & 0 & 0 & 0 & 0 & 0 & 0 & 0 & 1 & 0 & 1 & 0 & 1 & 0 & 1 \\
 0 & 0 & 0 & 0 & 0 & 0 & 0 & 0 & 0 & 0 & 0 & 1 & 1 & 0 & 0 & 0 & 0 \\
 0 & 0 & 0 & 0 & 0 & 0 & 0 & 0 & 0 & 0 & 0 & 0 & 0 & 1 & 1 & 0 & 0 \\
 0 & 0 & 0 & 0 & 0 & 0 & 0 & 0 & 0 & 0 & 0 & 0 & 0 & 0 & 0 & 1 & 1 
\end{pmatrix}\,,
\end{equation}
\begin{equation}
\label{eq:1711paritycheck}
  H = \begin{pmatrix}
 0 & 1 & 1 & 1 & 1 & 0 & 0 & 0 & 0 & 0 & 0 & 0 & 0 & 0 & 0 & 0 & 0 \\
 0 & 0 & 0 & 0 & 0 & 1 & 1 & 1 & 1 & 0 & 0 & 0 & 0 & 0 & 0 & 0 & 0 \\
 0 & 0 & 0 & 0 & 0 & 0 & 0 & 0 & 0 & 1 & 1 & 1 & 1 & 0 & 0 & 0 & 0 \\
 0 & 0 & 0 & 0 & 0 & 0 & 0 & 0 & 0 & 0 & 0 & 0 & 0 & 1 & 1 & 1 & 1 \\
 0 & 1 & 1 & 0 & 0 & 1 & 1 & 0 & 0 & 1 & 1 & 0 & 0 & 1 & 1 & 0 & 0 \\
 1 & 1 & 1 & 0 & 0 & 1 & 1 & 0 & 0 & 0 & 0 & 0 & 0 & 0 & 0 & 0 & 0 \\
  \end{pmatrix}\,.
\end{equation}

Simulating the full quantum circuit of message-passing BPQM would be unfeasible, due to the high number of qubits involved, which is larger than $n(B+1)$.
Thanks to the structure of the algorithm, it is possible to store the intermediate states of the circuit with a more refined approach.
In fact, following the argument in the proof of~\cref{lem:discretization_error}, one can see that the state of the message passed over some edge $e$ can be written as
\begin{equation}\label{eq:numerical_representation_state}
  \sum\limits_{\vec{j}\in\{0,1\}^{k_e-1}} p_{\vec{j}} \rho_{\vec{j},D_e} \otimes \anglestate{c_{\vec{j}}}\anglestateconj{c_{\vec{j}}}_{C_e}
\end{equation}
where $k_e$ is the number of check nodes preceding $e$, $\{p_{\vec{j}} | \vec{j}\}$ is a probability distribution, $\rho_{\vec{j}}$ is a single-qubit density matrix, $c_{\vec{j}}\in\mathcal{A}_B$. Here $D_e$ and $C_e$ denote the data part, respectively the angle part, of the message passed over $e$.
Note that this description is not pure, but rather in terms of a density matrix.
This is necessary, since we implicitly traced out the systems $A_e,Z_e$ and $S_e$.\footnote{$\rho_{\vec{j}}$ can be obtained by tracing out the system $Z$ in the state $\ket{\tilde{\varphi}_e}$ seen in the proof of~\cref{lem:discretization_error}.}

We can store a state of the form in~\cref{eq:numerical_representation_state} as a list of $2^{k_e-1}$ tuples $(p_i,\rho_i,c_i)_{i=0,\dots,2^{k_e-1}-1}$ where the $p_i$ are non-negative and sum up to one, $\rho_i$ are real $2\times 2$ matrices and $c_i\in\mathcal{A}_B$ are angle cosines.
For simplicity, we actually store the $c_i$ as floating-point numbers and always round them to the closest value in $\mathcal{A}_B$.
Since we do not store the full state and only the reduced state on $D_e$ and $C_e$, this representation makes it impossible to undo the action of BPQM, so this trick is only useful to determine the decoding performance of a single codeword bit.

An equality node operation takes two messages
\begin{equation*}
  (p_{1,i},\rho_{1,i},c_{1,i})_{i=0,\dots,2^{k_1-1}-1} \text{ and } (p_{2,j},\rho_{2,j},c_{2,j})_{j=0,\dots,2^{k_2-1}-1} \, .
\end{equation*}
Its output message is denoted by
\begin{equation*}
  (p_{3,ij},\rho_{3,ij},c_{3,ij})_{i=0,\dots,2^{k_1-1}-1,j=0,\dots,2^{k_2-1}-1} \, .
\end{equation*}
$\rho_{3,ij}$ is obtained by applying the unitary $U_{\ostar}$ characterized by the angles $\quantalphabeta{B}{\alpha(c_{1,i},c_{2,j})}$ and $\quantalphabeta{B}{\beta(c_{1,i},c_{2,j})}$ on the state $\rho_{1,i}\otimes\rho_{2,j}$ and then tracing out the ancilla output qubit.
The angle $c_{3,ij}$ is simply given by $\quant(c_{1,i}\cdot c_{2,j})$ and $p_{3,ij}$ by $p_{1,i}p_{2,j}$.

Similarly, the check node operation takes two messages
\begin{align*}
  (p_{1,i},\rho_{1,i},c_{1,i})_{i=0,\dots,2^{k_1-1}-1} \text{ and } (p_{2,j}, \rho_{2,j},c_{2,j})_{j=0,\dots,2^{k_2-1}-1} \, .
\end{align*}
Its output message is denoted by
\begin{align*}
  (p_{3,ijl},\rho_{3,ijl},c_{3,ijl})_{i=0,\dots,2^{k_1-1}-1,j=0,\dots,2^{k_2-1}-1,l=0,1} \, .
\end{align*}
and is defined by
\begin{equation}
  \rho_{3,ijl} = \frac{ (\id\otimes \bra{l})\cdot \cnot{}\cdot (\rho_{1,i}\otimes \rho_{2,j})\cdot \cnot{} \cdot (\id\otimes \ket{l})} {\Tr\left[(\id\otimes \bra{l})\cdot \cnot{}\cdot (\rho_{1,i}\otimes \rho_{2,j})\cdot \cnot{} \cdot (\id\otimes \ket{l}) \right]}\, 
\end{equation}
\begin{equation}
  p_{3,ijl} =  p_{1,i}p_{2,j} \cdot \Tr\left[(\id\otimes \bra{l})\cdot \cnot{}\cdot (\rho_{1,i}\otimes \rho_{2,j})\cdot \cnot{} \cdot (\id\otimes \ket{l}) \right] \, ,
\end{equation}
and
\begin{equation}
  c_{3,ijl} = \quant \left( \frac{c_{1,i} + (-1)^lc_{2,j}}{1 + (-1)^lc_{1,i}c_{2,j}} \right) \, .
\end{equation}

The statistics of measuring the data qubit of the final message $(p_i,\rho_i,c_i)_{i=0,\dots,2^{k-1}-1}$ in the $\ket{\pm}$ basis are given by the probabilities of measuring $\rho_i$ in the $\ket{\pm}$ basis, summed and weighted with the probabilities $p_i$.

\section{Alternative decoders for 8-bit code without cloning}\label{app:alternative_decoders}
In~\cref{sec:cloning_numerics} we investigated the performance of BPQM decoding of $X_1$ when we unroll the computation tree for different number of steps $h$. The result was that the decoder performed best under the choice of $h=2$, in which case the channel output of the bit $X_3$ is approximately cloned. However, this is not a conclusive proof that cloning certain channel outputs is necessary in order to obtain the best possible decoder. One might argue that there could exist some other choice of $\mathcal{C}'$ that does not require the use of approximate cloning and that achieves better result than the $h=2$ case. The goal of this appendix is to argue that this is not the case, at least for the arguably natural candidates for such cloning-free decoders.

We consider certain tree subgraphs of the Tanner graph of the 8-bit code, which each define a code, which in turn has its optimal decoder realized by BPQM. We take these induced decoders as candidates for cloning-free decoders which could potentially outperform the $h=2$ decoder. For example, if one were to choose the subtree consisting of the variable nodes $X_1,X_2,X_4,X_5,X_8$ and the in-between check nodes, then one would retrieve the identical decoder as in the $h=1$ case.

In~\cref{fig:alternative_decoders} we depict the subtrees for three candidate strategies which we deem plausible to possibly outperform the $h=2$ case. Each of them contains more variable nodes than the $h=1$ case---because the more variables the decoder takes into account, the better we expect the result to be. The challenge in choosing these strategies is that we somehow want to include $X_6$ and $X_7$ without having to clone $X_3$. The numerical results of the three strategies are depicted in~\cref{fig:alternative_decoders_performance}.

\begin{figure}[t]
  \centering
  \begin{adjustbox}{max width=0.33\textwidth}
  \tikzsetnextfilename{Strategy1}
  \begin{tikzpicture}
    \node[squarenode] (CNW) {$+$};
    \node[roundnode] (X1) [right=0.5cm of CNW] {$X_1$};
    \node[squarenode] (CNE) [right=0.5cm of X1] {$+$};
    \node[roundnode] (X4) [below=0.5cm of CNE] {$X_4$};
    \node[squarenode] (CSE) [below=0.5cm of X4] {$+$};
    \node[roundnode] (X2) [below=0.5cm of CNW] {$X_2$};
    \node[squarenode] (CSW) [below=0.5cm of X2] {$+$};
    \node[roundnode] (X5) [above left=0.3cm and 0.3cm of CNW] {$X_5$};
    \node[roundnode] (X6) [below left=0.3cm and 0.3cm of CSW] {$X_6$};
    \node[roundnode] (X7) [below right=0.3cm and 0.3cm of CSE] {$X_7$};
    \node[roundnode] (X8) [above right=0.3cm and 0.3cm of CNE] {$X_8$};
    \node (text) [above=0.5cm of X1] {Strategy 1:};
    \draw (CNW) to (X1);
    \draw (X1) to (CNE);
    \draw (CNE) to (X4);
    \draw (X4) to (CSE);
    \draw (CSW) to (X2);
    \draw (X2) to (CNW);
    \draw (CNW) to (X5);
    \draw (CSW) to (X6);
    \draw (CSE) to (X7);
    \draw (CNE) to (X8);
  \end{tikzpicture}
  \end{adjustbox}
  \hspace{5mm}
  \begin{adjustbox}{max width=0.33\textwidth}
    \tikzsetnextfilename{Strategy2}
  \begin{tikzpicture}
    \node[squarenode] (CNW) {$+$};
    \node[roundnode] (X1) [right=0.5cm of CNW] {$X_1$};
    \node[squarenode] (CNE) [right=0.5cm of X1] {$+$};
    \node[roundnode] (X4) [below=0.5cm of CNE] {$X_4$};
    \node[squarenode] (CSE) [below=0.5cm of X4] {$+$};
    \node[roundnode] (X2) [below=0.5cm of CNW] {$X_2$};
    \node[squarenode] (CSW) [below=0.5cm of X2] {$+$};
    \node[roundnode] (X3) [right=0.5cm of CSW] {$X_3$};
    \node[roundnode] (X5) [above left=0.3cm and 0.3cm of CNW] {$X_5$};
    \node[roundnode] (X6) [below left=0.3cm and 0.3cm of CSW] {$X_6$};
    \node[roundnode] (X7) [below right=0.3cm and 0.3cm of CSE] {$X_7$};
    \node[roundnode] (X8) [above right=0.3cm and 0.3cm of CNE] {$X_8$};
    \node (text) [above=0.5cm of X1] {Strategy 2:};
    \draw (CNW) to (X1);
    \draw (X1) to (CNE);
    \draw (CNE) to (X4);
    \draw (X4) to (CSE);
    \draw (X3) to (CSW);
    \draw (CSW) to (X2);
    \draw (X2) to (CNW);
    \draw (CNW) to (X5);
    \draw (CSW) to (X6);
    \draw (CSE) to (X7);
    \draw (CNE) to (X8);
  \end{tikzpicture}
  \end{adjustbox}
  \hspace{5mm}
  \begin{adjustbox}{max width=0.33\textwidth}
    \tikzsetnextfilename{Strategy3}
  \begin{tikzpicture}
    \node[squarenode] (CNW) {$+$};
    \node[roundnode] (X1) [right=0.5cm of CNW] {$X_1$};
    \node[squarenode] (CNE) [right=0.5cm of X1] {$+$};
    \node[roundnode] (X4) [below=0.5cm of CNE] {$X_4$};
    \node[squarenode] (CSE) [below=0.5cm of X4] {$+$};
    \node[roundnode] (X2) [below=0.5cm of CNW] {$X_2$};
    \node[roundnode] (X3) [left=0.5cm of CSE] {$X_3$};
    \node[roundnode] (X5) [above left=0.3cm and 0.3cm of CNW] {$X_5$};
    \node[roundnode] (X7) [below right=0.3cm and 0.3cm of CSE] {$X_7$};
    \node[roundnode] (X8) [above right=0.3cm and 0.3cm of CNE] {$X_8$};
    \node (text) [above=0.5cm of X1] {Strategy 3:};
    \draw (CNW) to (X1);
    \draw (X1) to (CNE);
    \draw (CNE) to (X4);
    \draw (X4) to (CSE);
    \draw (CSE) to (X3);
    \draw (X2) to (CNW);
    \draw (CNW) to (X5);
    \draw (CSE) to (X7);
    \draw (CNE) to (X8);
  \end{tikzpicture}
  \end{adjustbox}
  \caption{Three candidate strategies for cloning-free decoders.}
  \label{fig:alternative_decoders}
\end{figure}
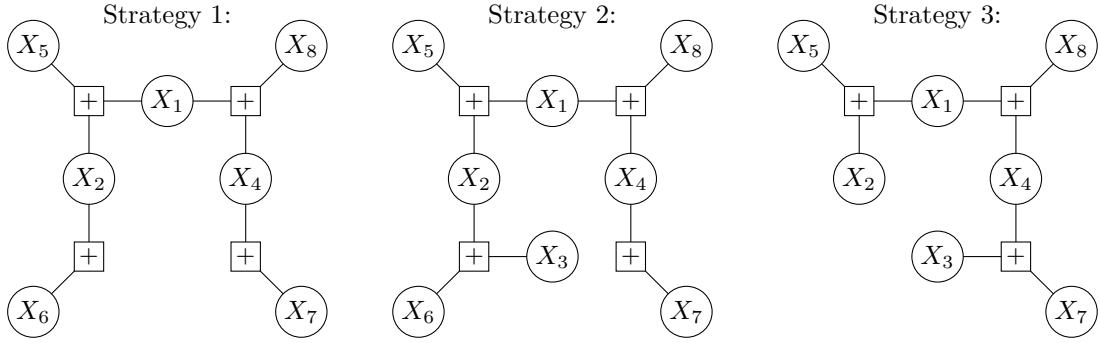

\begin{figure}
  \centering
   \includegraphics[width=0.75\textwidth]{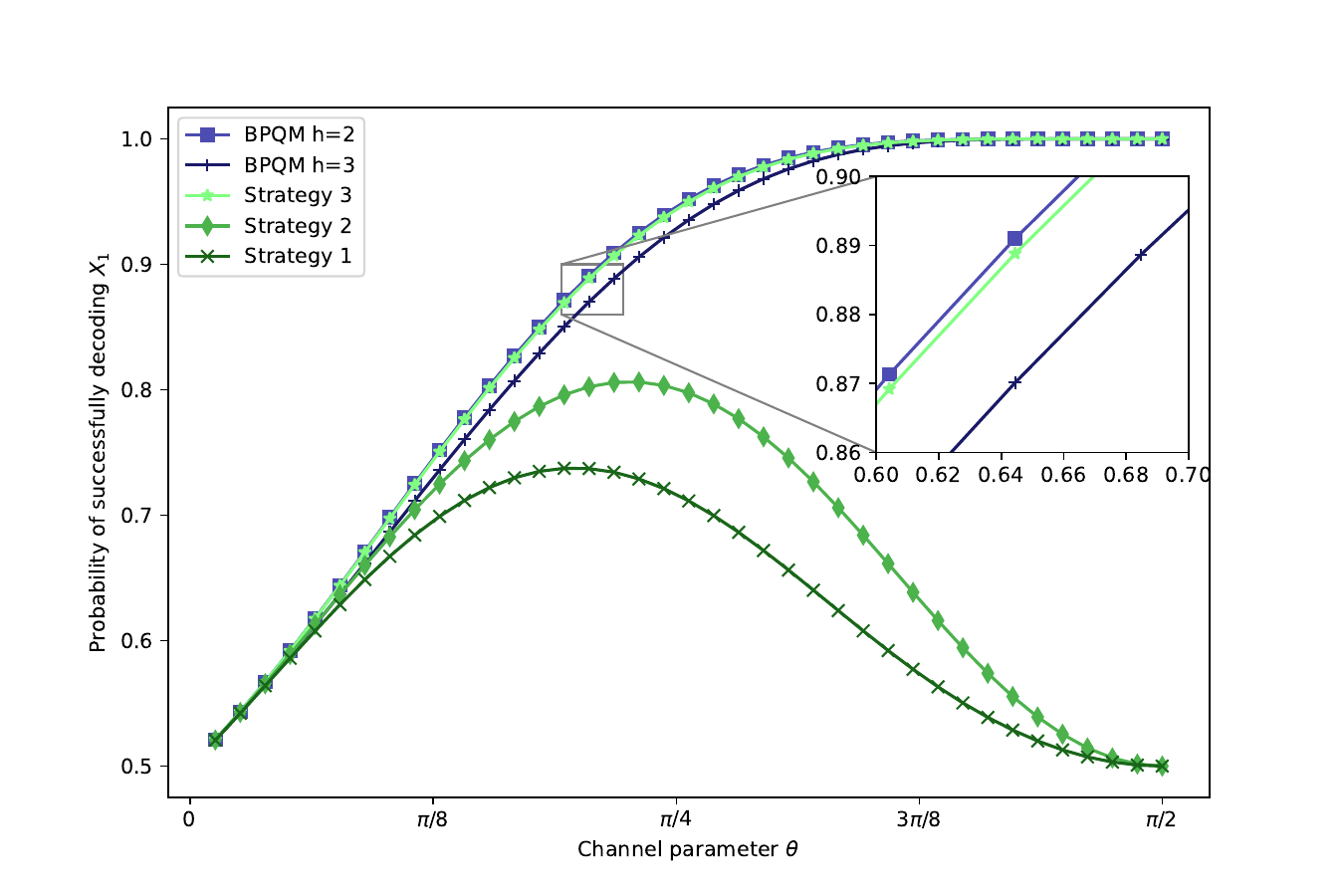}
  \caption{Decoding performance of the three candidate strategies depicted in~\cref{fig:alternative_decoders} for cloning-free decoders of the 8-bit code.}
  \label{fig:alternative_decoders_performance}
\end{figure}

All three strategies are outperformed by $h=2$, though strategy 3 performs only slightly worse than $h=2$. A more careful analysis of the strategies 1 and 2 reveals that their performance depends on the input codeword. They perform badly on codewords where $X_3=1$, and in fact have a zero percent success probability for these codewords in the limit $\theta\rightarrow \pi/2$. This makes sense---the Tanner graph on which their decoder is based on the implicit assumption that $X_7=X_4$, which only holds if $X_3=0$.

\printbibliography[heading=bibintoc,title={\large References}]

\end{document}